\documentclass[11pt]{article}

\usepackage{authblk} 

\usepackage{graphicx} 

\usepackage{fullpage}
\usepackage{svg}

\usepackage[
  textsize=scriptsize,
]{todonotes} 
\setuptodonotes{fancyline}
\usepackage[ruled]{algorithm2e}

\usepackage{amsmath}

\ifdefined\newcounteralias%
  \RemoveFromHook{package/amsthm/after}[firstaid/aliascounter]%
\fi
\usepackage{amsthm}
\usepackage[colorlinks=true,citecolor=blue,linkcolor=blue,hypertexnames=false]{hyperref}
\usepackage[capitalise]{cleveref}
\usepackage{cancel}
\usepackage[font={small}]{caption}
\usepackage{mathrsfs}

\usepackage{thmtools}

\declaretheorem[numberwithin=section]{theorem}
\declaretheorem[sibling=theorem,name=Lemma]{lemma}
\declaretheorem[sibling=theorem,name=Definition]{definition}

\declaretheorem[sibling=theorem,name=Proposition]{proposition}
\declaretheorem[sibling=theorem,name=Corollary]{corollary}
\declaretheorem[sibling=theorem,name=Remark]{remark}
\declaretheorem[sibling=theorem,name=Assumption]{assumption}

\declaretheorem[sibling=theorem,name=Fact]{fact}

\crefname{lemma}{Lemma}{Lemmas}\Crefname{lemma}{Lemma}{Lemmas}
\crefname{definition}{Definition}{Definitions}\Crefname{definition}{Definition}{Definitions}
\crefname{observation}{Observation}{Observations}\Crefname{observation}{Observation}{Observations}
\crefname{proposition}{Proposition}{Propositions}\Crefname{proposition}{Proposition}{Propositions}
\crefname{corollary}{Corollary}{Corollaries}\Crefname{corollary}{Corollary}{Corollaries}
\crefname{remark}{Remark}{Remarks}\Crefname{remark}{Remark}{Remarks}

\usepackage{amssymb}
\usepackage{microtype}

\usepackage[normalem]{ulem}
\usepackage[dvipsnames]{xcolor}

\usepackage[T1]{fontenc}
\usepackage{lmodern}
\usepackage{FiraSans}

\usepackage{tikz}

\AddToHook{cmd/appendix/before}{\crefalias{section}{appendix}}

\usetikzlibrary{arrows.meta,calc,decorations.pathreplacing}
\definecolor{heatred}{RGB}{200,40,40}
\definecolor{keyblue}{RGB}{35,90,180}
\definecolor{gateblue}{RGB}{90,160,230}
\definecolor{eraseorange}{RGB}{240,150,60}
\definecolor{garbagepurple}{RGB}{140,60,170}
\tikzset{
  wire/.style={thick},
  wirelabel/.style={fill=white, inner sep=0.5pt, font=\scriptsize},
  gate/.style={draw=keyblue!70!black, rounded corners=2pt, fill=gateblue!30, inner sep=0pt, font=\small},
  tightgate/.style={gate, inner sep=1.5pt, minimum height=0.3cm, font=\scriptsize},
  cerase/.style={gate, draw=eraseorange!70!black, fill=eraseorange!15},
}

\newcommand{\dd}[1]{{\leavevmode\color{teal}DD: #1}}

\newcommand{\boltzmann}{k}
\newcommand{\temp}{\mathbb{T}}
\newcommand{\dist}{\mathcal{D}}
\newcommand{\EP}{\mathsf{EP}}
\newcommand{\EF}{\mathsf{Q}}

\newcommand{\DKL}{D_{\mathrm{KL}}}
\newcommand{\entropy}{\mathcal{S}}

\newcommand{\heat}{Q}
\newcommand{\heatoptimal}{Q^{*}}
\newcommand{\heato}{\tilde Q}
\newcommand{\heatestimate}{\widehat{Q}}
\newcommand{\varLanguage}{\mathcal{L}}
\newcommand{\halfAccept}{{\SAT/2}}
\newcommand{\SAT}{\ensuremath{\textnormal{\textsc{Sat}}}}
\newcommand{\EFmin}{{\heat \text{-}\mathrm{min}}}
\renewcommand{\P}{\mathsf{P}}
\newcommand{\NP}{\mathsf{NP}}
\newcommand{\PP}{\mathsf{PP}}

\newcommand{\BPP}{\mathsf{BPP}}

\newcommand{\landauer}{\mathcal{L}}

\newcommand{\mutualInformation}{\mathcal{I}}

\newcommand{\support}{\mathrm{supp}}
\newcommand{\given}{\;|\;}
\newcommand{\circuit}{\mathcal{C}}
\newcommand{\varSpace}{W}
\newcommand{\varTime}{T}

\newcommand{\twodots}{\mathinner{\ldotp\ldotp}}

\newcommand{\E}{\mathbb{E}}

\newcommand{\R}{\mathbb{R}}

\newcommand{\Z}{\mathbb{Z}}

\newcommand{\calA}{\mathcal{A}}

\newcommand{\calO}{\mathcal{O}}

\newcommand{\calG}{\mathcal{G}}
\newcommand{\calU}{\mathcal{U}}

\newcommand{\calE}{\mathcal{E}}

\newcommand{\poly}{\mathrm{poly}}

\newcommand{\domain}{D}
\newcommand{\range}{R}

\renewcommand{\epsilon}{\varepsilon}

\newcommand\abs[1]{\left|#1\right|}
\newcommand{\floor}[1]{\left\lfloor #1 \right\rfloor}
\newcommand{\ceil}[1]{\left\lceil #1 \right\rceil}
\newcommand{\binaryEntropy}{\mathrm{H_2}}

\newcommand{\valil}[2]{w_{{#1}}(#2)}
\newcommand{\hide}[1]{}
\newcommand{\Prob}{\mathrm{Pr}}
\newcommand{\unif}{\mathrm{unif}}

\newcommand{\newbounda}{\alpha}

\newcommand{\gateSet}{G}

\newcommand{\idx}{I}

\newcommand{\alphabet}{\Sigma}

\newcommand{\gibbs}{\calG}

\newcommand{\counts}{C}
\newcommand{\countDist}{\mathscr{C}}
\newcommand{\EPlocal}{\EP^{\mathrm{Loc}}}

\newcommand{\compare}{\mathrm{Comp}}

\newcommand{\sort}{\mathrm{Sort}}
\newcommand{\distUniform}{\calU}

\newcommand{\evenSet}{\calE}
\newcommand{\oddSet}{\calO}
\newcommand{\xmin}{x_{\min}}
\newcommand{\xmax}{x_{\max}}
\newcommand{\hyperGeo}[3]{HG^{#1,#2}_{#3}}

\newcommand{\blank}[1]{}

\allowdisplaybreaks

\makeatletter
\newcommand{\compacttoc}[1][1]{%
  \begingroup
  \setcounter{tocdepth}{#1}%
  \setlength{\parskip}{0pt}%
  \renewcommand*{\l@section}{\@dottedtocline{1}{0em}{1.7em}}%
  \renewcommand*{\l@subsection}{\@dottedtocline{2}{1.7em}{2.5em}}%
  \small
  \begin{center}\normalsize\bfseries\contentsname\end{center}%
  \nobreak\vspace{-1.5ex}%
  \@starttoc{toc}%
  \endgroup
}
\makeatother

\begin{document}

\title{The Thermodynamic Complexity of Circuits\footnote{
    A preliminary version of this paper appeared in chapter 6 of Niels Kornerup's PhD thesis \cite{Kor25}. This version substantially improves, corrects, and builds on the results in that chapter.
}
}
\date{}

\author[1]{Anne Condon\thanks{Condon was supported by an NSERC Discovery grant.}}
\author[2]{David Doty\thanks{Doty was supported by NSF awards 2329909 and 2211793 and DoE award DE-SC0024467.}}
\author[3]{Niels Kornerup\thanks{Kornerup was supported by the LDRD Program at Sandia National Laboratories. Sandia is managed and operated by NTESS under DOE NNSA contract DE-NA0003525.}}
\author[4]{David Soloveichik\thanks{Soloveichik was supported by a Schmidt Sciences Polymath award.}}
\author[5]{David Wolpert\thanks{Wolpert was supported by NSF award 2626044.}}

\affil[1]{University of British Columbia}
\affil[2]{University of California, Davis}
\affil[3]{Sandia National Laboratories}
\affil[4]{University of Texas, Austin}
\affil[5]{Santa Fe Institute, New Mexico; Complexity Science Hub, Vienna; Arizona State University, Arizona; International Center for Theoretical Physics, Trieste}

\maketitle

\begin{abstract}

When any physical system executes a logically irreversible (non-injective) step of a computation, the excess heat released over the so-called Landauer limit is an energetic cost that cannot be recovered. In modern computers this cost vastly dwarfs the Landauer limit. Modern stochastic thermodynamics refers to this cost as the ``entropy production'' and has derived some very powerful lower bounds on its value, which strictly strengthen the second law of thermodynamics. Here we introduce a model of circuits using \emph{heat functions} and investigate its entropy production, deriving new results on the thermodynamic complexity of circuits. 

In our model, the circuits compute \emph{cleanly}, erasing their inputs and all intermediate values. We first show that computing the optimal heat functions for gates in a given circuit (the gates that minimize total entropy production) when the circuit's input is uniformly random is $\mathsf{PP}$-complete. We also show that multiplicative approximation of those heat functions is $\mathsf{NP}$-hard, but additive approximation admits a polynomial-time randomized algorithm. We give tight bounds for entropy production for certain specific families of circuits (on uniformly distributed inputs): an $n$-ary AND tree of depth $O(\log n)$ requires $\Omega(n)$ circuit entropy production, while a serial chain achieves $O(1)$, showing an energy-depth tradeoff. The Batcher odd-even mergesort network on $n$ bits has $\Theta(n)$ entropy production,  asymptotically less than its $\Theta(n \log^2 n)$ gate count. On random permutations of distinct integers its entropy production is $\Theta(n \log^2 n)$, asymptotically exceeding the Landauer cost of sorting; but preceding it with Leighton and Plaxton's almost-sorting network reduces the total heat released to $O(n \log n)$. Finally, we show that a clean sorting circuit (not a sorting network) can sort $n$ values with zero circuit entropy production over any exchangeable input distribution. Our framework also yields streamlined proofs of known results in stochastic thermodynamics on mismatch cost and modularity cost.
\end{abstract}

\tableofcontents



\section{Introduction}
\label{sec:intro}


A central part of the industrial revolution of
the 19th century was the invention of steam engines. This 
invention resulted in a
practical demand for powerful and efficient steam engines,
which in turn
ignited a mathematical revolution in physics:
the development of the theory of thermodynamics to understand and quantify how efficiently an engine
could accomplish mechanical objectives such as moving a large train.

Then in the 20th century there was a similarly
revolutionary invention, of computers. Accordingly,
the energy required to accomplish \emph{computational} objectives
has been a subject of academic inquiry for a century.
A deeper and more powerful understanding
of this thermodynamics of computation would not just
be a purely theoretical advance. As a practical matter,
we live in a world where the energetic costs of computing are critically important.
The International Energy Agency estimates that $485$ terawatt-hours was consumed running data centers in the year 2025 and projects that this number will roughly double to $950$ terawatt-hours in $2030$ \cite{IEA2026KeyQuestions}.
Despite this, our theoretical understanding of the fundamental thermodynamic costs of computation is rather limited, partially due to the lack of a strong and flexible model that captures many of the nuances of the thermodynamics of computation beyond logical reversibility.
While current computational devices expend energy at rates well above those mandated by the laws of thermodynamics, as technology improves a principled understanding of thermodynamics is key to optimizing the energetic costs of computation.

However, most of the study of the thermodynamics of
computation to date has 
been based on classical thermodynamics principles.
While the field of thermodynamics itself has continued to evolve in recent decades,
the study of the thermodynamics of computation has not kept pace.
In particular, there has been a revolution in non-equilibrium
statistical physics in the 21st century, resulting in what
is now called \textit{stochastic thermodynamics}. Whereas classical thermodynamics is, in its foundations, concerned only with systems at
equilibrium, stochastic thermodynamics focuses on systems arbitrarily
far from equilibrium. This makes it the ideal formalism for analyzing
the thermodynamics of computation.

This paper serves two purposes: 
First, we introduce such a new model for 
investigating the thermodynamics of circuits.
Our goal with this model
is to strike a balance between broad accessibility to a general computer science audience 
and correctness with respect to 
{stochastic thermodynamics},
which represents our current best understanding of the fundamental energetic cost of machines,
including computing machines.
Second, we show how this circuit model can be leveraged to prove a wide range of new results ranging from the computational complexity of optimizing Boolean circuits to asymptotically tight lower and upper bounds on the thermodynamic costs of fixed circuits.

We start by giving a detailed overview of the history of the thermodynamics of computation before diving into our model and results.

\subsection{Landauer's principle and logical reversibility}
\label{sec:intro-reversible-computing}
Most of the computer science literature on the fundamental energetic costs of computation studies \emph{reversible computing},\footnote{
    Notable exceptions abound in the physics literature,
    e.g.,~\cite{WK20,YCW24,YYW25,Wol19,WKL+24},
    but these papers generally use physics terms and assumptions that are often difficult for computer scientists to read.
}
This is an important and groundbreaking body
of work, which we now summarize.
Suppose a computing device executes a basic computational step that is \emph{logically irreversible},
meaning that the device goes from logical state $c_0$ 
(such as the configuration of a multitape Turing machine, for instance) 
to logical state $c_1$,
where $c_1$ has \emph{another} predecessor state $c_0'$ that also transitions to $c_1$ in one step.
Formally, the global transition function $\delta: C \to C$ mapping the machine's current state to its next state is not injective.
Landauer~\cite{Lan61} argued that this necessarily decreases entropy within the computing device, 
entropy must correspondingly increase elsewhere in the universe to satisfy the Second Law of thermodynamics.
Phrased in terms of \emph{work} (i.e., accessible energy), \emph{Landauer's principle} states that \emph{logically irreversible operations must expend work in order to release heat into their environments.}\footnote{
    In this paper we use the term ``release heat'' mostly synonymously with ``dissipate heat'',
    but generally reserve the term ``dissipate'' to mean a loss of heat/work/useful energy that is irreversible.
    However, we are careful to include a descriptor such as ``irreversible'' or ``irrecoverable'' when we want to identify that the lose of heat/energy cannot be undone; absent this descriptor,
    any release/transfer of heat/work/useful energy in principle may be recovered.
}
Given a transition from $n$ equally likely predecessor states to $m$ equally likely successors, $\boltzmann \temp \ln (n/m)$ heat must be released,
regardless of the physical implementation of the operation. Here 
$\temp$ denotes the system's temperature\footnote{In physics, it is common to specify a \emph{heat bath} of temperature $\temp$. For our purposes, we can think of $\temp$ as the temperature of both the system and its environment.}
and
$\boltzmann$ is Boltzmann's constant.\footnote{
    Throughout this paper, we choose the units for measuring energy so that $\boltzmann \temp \ln 2 = 1$. So we
    measure energy and heat in comparable units to bits.
}

Landauer argued~\cite{Lan61} that most real computation would necessarily execute many irreversible operations, implying an unavoidable cost of $\Omega(t \boltzmann \temp \ln 2)$ to execute $t$ steps of a computation regardless of the physical implementation of the computing device.
Bennett~\cite{Ben73} challenged this assumption, showing how to convert an arbitrary Turing machine $M$ into a 
logically reversible Turing machine $R$ computing the same function.
The configuration graph of $R$ is then a 1D line 
$c_0 \to c_1 \to \ldots \to c_t$. 
A completely unbiased random walk on this graph, requiring no work to bias it forwards, would then visit final configuration $c_t$ in $O(t^2)$ time on average.\footnote{
Bennett~\cite{Ben73} observed that readout from such a device would be impractical, since even after waiting $t^2$ time, there is only a $1/t$ probability of observing $R$ in its final output configuration.
He proposed dissipating small $\epsilon \ll \boltzmann \temp \ln 2$ energy per step to bias the computation forward just enough to ensure a constant probability of finding the machine in an output configuration.
This still implies $\epsilon t = \Omega(t)$ energy cost for a $t$-step computation. See~\cite{DKL+24} for arguments that this total cost need not scale linearly with $t$, essentially by using algorithmic design techniques to increase the probability of observing the machine in an output configuration even with totally unbiased transitions, obviating the need to supply energy for even a slight forward bias.
However, this trick (and the original unbiased walk outlined in \cite{Ben73}) doesn't account for the thermodynamic costs of resetting the device (see~\cite{Nor13,SCS+15}).
In \cite{DKL+24} the authors avoid these costs by designing a device that never needs to be reset.
}

Bennett's original construction~\cite{Ben73} simulates a $t$-step Turing machine in $O(t)$ steps 
but pays with an increase in memory complexity up to $\Theta(t)$, 
essentially needing to write down the entire history of the computation on a separate worktape to maintain logical reversibility,
bringing the memory use up to asymptotically match the time.
(Thus a polynomial-time, log-space computation, for example, pays exponentially extra memory to make the computation logically reversible.)
A series of papers over the next few decades~\cite{Bennett1989,LevineSherman1990,LiVitanyi1996,LiTrompVitanyi1998,LangeMcKenzieTapp2000,BuhrmanTrompVitanyi2001}
showed various improved simulation time/space tradeoffs in reversible Turing machines;
see~\cite{Bennett1982,Vitanyi2005,Morita2008} for useful surveys.
Toffoli and Fredkin~\cite{Toffoli1977,Toffoli1980,FredkinToffoli1982} pioneered the study of reversible computation in other computational models such as cellular automata and circuits.
Meanwhile in the physics community, there has been work on understanding the thermodynamics of logically reversible Turing machines~\cite{PhysRevE.92.042104,brittain2021buildthermodynamicallyreversibleuniversal}.

So what is there to say about the fundamental thermodynamic costs of computation that has not already been said in the reversible computing literature?
First, we may consider a situation in which we do not have access to a logically reversible gate set, 
in which we wish to understand the fundamental energetic costs associated with whatever computational devices we do have on hand.
While there has been progress towards actually building logically reversible computational devices (e.g. \cite{Ammer1998Scalable,DBLP:journals/cse/Frank02,DBLP:conf/icrc/FrankBTMH20}), these systems are not yet ready for widespread use and most computation today is performed in a highly irreversible manner.
This is a deep research issue in its own right, and more practically, reflects the situation in almost every deployed computational system.

Second, over the last few decades the field of non-equilibrium statistical physics has undergone major developments, culminating in a new subfield called \emph{stochastic
thermodynamics}~\cite{van2015ensemble,Seifert2025StochasticThermodynamics}. 
These developments have upended our 
understanding of the thermodynamics of physical systems that implement computation. In particular, they reveal
that the traditional understanding of 
logical reversibility and Landauer's principle is at best partially correct:
while the Second Law does indeed imply that decreasing entropy of a system locally requires increasing entropy outside the system,
nothing in the Second Law says that this must be a physically irreversible loss of work.

To give a simple example of a logically irreversible process that can be done in a physically/thermodynamically reversible way (i.e., can be done with arbitrarily little dissipated work), suppose we have a particle in a cylinder, and identify the particle being in the top half of the cylinder as a bit $1$, while the particle being in the bottom half is identified as bit $0$, 
with equal probability.
So if we compress the cylinder to its bottom
half using a piston, we perform a bit erasure, 
since afterwards the particle is guaranteed to be in the bottom half.
Moreover, in first-year physics we learn that if this compression is
done sufficiently slowly, then all the work expended in that compression can be recovered afterward simply by \emph{un}compressing the cylinder sufficiently slowly, returning it to its original volume,
and returning the particle's distribution to equally likely to be in either half (see Szilard engines for more details, \cite{Szi29,doi:10.1073/pnas.1406966111}).

This is precisely the property that defines a process as being thermodynamically reversible. In contrast, if one compresses the cylinder too quickly, 
no matter how one subsequently uncompresses the cylinder, one won't recover all the work that was used in the compression. This is what is meant by thermodynamic irreversibility. 
This example illustrates that logical and thermodynamic reversibility are independent properties of a process;
it is possible to have one without the other.
(See~\cite{Sag14} for a pedagogical elaboration of this point.)

Third, if the function $f$ being computed is not injective,
it is not possible to compute the mapping $x \mapsto f(x)$ in a logically reversibly way;
instead reversible computing models typically compute $x \mapsto (x,f(x))$.
But the mappings $x \to f(x)$ and $x \to (x,f(x))$ are fundamentally different.
If $f$ is non-injective or if $f^{-1}$ is prohibitively expensive to compute (i.e., if $f$ is a one-way function),
then it is impossible to efficiently and reversibly obtain just $f(x)$ from the pair $(x,f(x))$.\footnote{
    Bennett~\cite{Ben73} shows that,
    if $f$ is injective,
    it is possible to compute the mapping $x \mapsto f(x)$ reversibly 
    (i.e., not needing the input also to appear as part of the output). 
    However, this construction has time complexity as high as the max of computing $f$ \emph{and} computing $f^{-1}$,
    since it involves taking a direct computation of $f^{-1}$ by a Turing machine and making it logically reversible in the same way that $f$ is reversibly computed.
    But if we want the time complexity to be no worse than that of computing $f$,
    this rules out one-way functions in which $f^{-1}$ has time complexity fundamentally greater than that of $f$.
    Indeed, it is well-known folklore that the mapping $x \mapsto f(x)$ cannot be efficiently computed in a logically reversible manner if $f$ is a one-way function:
    if such a computation, for example a Boolean circuit, were completely logically reversible,
    starting only with $x$ and ending only with $f(x)$,
    then the mirror image of such a circuit would be an equally efficient circuit computing $f^{-1}$,
    contradicting that $f$ is one-way.
}
One might think it is okay to just leave an extra copy of $x$ in your computer's memory, but 
if you compute $f$ repeatedly and old values of $x$ are overwritten, the cost of erasing these old values must be paid
(see~\cite[Sec.\,11]{Wol19}).
One might hope to store values of $x$ generated by intermediate steps of a larger logically reversible computation and reversibly uncompute them during the uncomputation step, however, temporarily storing $x$ increases the space complexity of the overall computation (see~\cite{Bennett1989}).


Finally,
it remains open whether every time-$t$, space-$s$ Turing machine can be simulated by a reversible Turing machine using $O(t)$ time and $O(s)$ space;
the current best results asymptotically increase at least one of them~\cite{Vitanyi2005},
and there is an oracle relative to which this is necessary~\cite{FrankAmmer2017}.

In this paper we start to investigate how {stochastic thermodynamics} can help untangle these issues.

\subsection{Beyond logical reversibility: Stochastic thermodynamics}
\label{sec:intro-beyond-logical-reversibility}

Bennett's approach of requiring logical reversibility of every computational step is in one sense a very ambitious way to minimize energy costs by pushing the so-called \emph{Landauer cost} of executing logically irreversible operations to 0,
simply by removing all such operations.
Li and Vit\'{a}nyi~\cite{LiVitanyi1996} investigated how it is possible to trade time and space for `energy' by limiting the number of logically irreversible steps in a computation.
Demaine, Lynch, Mirano, and Tyagi~\cite{DemaineLMT16} built on this idea and proposed a generalization of logical reversibility they call \emph{semi-reversible computing},
in which irreversible operations $f:\domain \to \range$ are permitted and assumed to cost $\log \frac{|\domain|}{|\range|}$ bits of energy to execute
(this being the natural generalization of the Landauer cost of 1 bit when $|\domain|=2$ and $|\range|=1$,
assuming a uniform distribution on $\domain$).
The total energy cost (by this definition) is then treated as a complexity measure to be minimized.

However, the more modern field of \emph{stochastic thermodynamics}~\cite{van2015ensemble,Seifert2025StochasticThermodynamics}
has generalized our understanding of the interaction of entropy and fundamental energy use in physical devices,
beyond the overly simplistic assumption that each bit erasure costs $\boltzmann \temp \ln 2$ energy.
It has been argued~\cite{WKL+24} that stochastic thermodynamics can be fruitfully applied to computational complexity to better understand the relationship between fundamental thermodynamic costs and other traditional complexity measures such as time and space.

We now review
\emph{entropy production},
a basic concept from modern stochastic thermodynamics that lets us study the fundamental limits of energetic costs 
even for logically irreversible computation.
As applied to the study of computation, 
we propose entropy production as a computational complexity measure/cost to be minimized.
This is already the measure of interest for such systems in the physics community (i.e. \cite{Wol19,peliti2021stochastic,WKL+24,Seifert2025StochasticThermodynamics}).
We then analyze these costs in a 
particular circuit model, 
noting that our approach can be extended in a natural way to other models of computation.

We provide our full definitions in \Cref{sec:prelim} and give
a semi-formal summary of the underlying concepts here.
Briefly, the generalization of the Second Law of thermodynamics most central to our purpose is the equation

\begin{equation}
\label{eq:intro-entropy-production}
\EP = 
\EF 
- (\entropy(\dist_0) - \entropy(\dist_1)).
\end{equation}
Equation \eqref{eq:intro-entropy-production} appears later in \Cref{def:ep-gate} with notation specialized to the study of the thermodynamics of circuits.
The terms $\EP$ and $\EF$ respectively mean \emph{entropy production} and \emph{entropy flow}.
However, throughout this paper we focus
on the case of a single (thermal)
reservoir. That leads us to use the term \emph{heat released} for $\EF$.
Furthermore,
given our choice of units such that $\boltzmann \temp \ln 2 = 1$, we can equate information in bits with equal units of heat; we discuss these terms below.

To begin, we focus on 
the term $\entropy(\dist_0) - \entropy(\dist_1)$
 in~\cref{eq:intro-entropy-production}.
$\dist_0$ represents the probability distribution on states of a physical device\footnote{
    The stochastic thermodynamic perspective can be viewed as an ``intrinsically distributional'' complexity measure because it defines costs that are fundamentally tied to a distribution over inputs. 
    In other words, there is not in general a sensible notion of ``worst-case entropy production'' defined as $\max_{x \in \{0,1\}^n} \EP(x)$, as is typical for definitions of worst-case time or space complexity.
} prior to executing a (possibly logically irreversible) step, and $\dist_1$ represents the associated distribution after; $\entropy(\dist) = \sum_{x \in \mathrm{supp}(\dist)} \Pr[x] \log \frac{1}{\Pr[x]}$ is the \emph{Shannon entropy} 
of $\dist$.\footnote{For example, if the system has two bits $x_1,x_2$ each i.i.d.\,$0$ or $1$ with probability $\Pr_{\dist_0}[x_1=1] = \Pr_{\dist_0}[x_2=1] = 1/2$, and if the system computes $y = \mathrm{AND}(x_1,x_2)$, then implied output distribution $\dist_1$ on $y$ is $\Pr_{\dist_1}[y=1] = 1/4$ and $\Pr_{\dist_1}[y=0] = 3/4$.
}

Concretely, suppose we perform the most widely discussed logical operation in the thermodynamics of computing: erasing a single bit $b \in \{0,1\}$.
Here ``erasing $b$'' generally means ``setting $b$ to 0'', i.e., no matter whether $b$ was 0 or 1 before, afterward it is 0 and so the system has ``erased'' any knowledge of $b$'s previous value.
In this case, $\dist_1$ is the degenerate point-mass distribution with $\Pr_{\dist_1}[b=0] = 1$.
The prior distribution $\dist_0$ may be a number of things, but the worst case in terms of logical irreversibility is the uniform distribution $\dist$ with
$\Pr_{\dist_0}[b=0] = \Pr_{\dist_0}[b=1] = 1/2$, i.e., prior to erasing the bit, we have maximal uncertainty about its current value.
The term $\entropy(\dist_0) - \entropy(\dist_1)$, the difference in Shannon entropy between distributions $\dist_0$ and $\dist_1$, is the so-called \emph{Landauer cost}.\footnote{
\label{foot:landauer-cost-not-always-one-bit}
    Much literature on reversible computing assumes the Landauer cost of erasing one bit is, depending on units, 1 bit of entropy or $\boltzmann \temp \ln 2$ units of heat.
    (In this paper we will use units of entropy, i.e., bits.)
    However, this assumes $\dist_0$ is the worst-case uniform distribution;
    for other distributions the drop in entropy is strictly less than one bit:
    it is precisely the binary entropy $H_2(p)$ if one of the bits has probability $p$.
    As an extreme example, if $\Pr_{\dist_0}[b=0]=1$, then the Landauer cost is 0. 
    Indeed, obviously it costs nothing to ``set'' a bit to 0 when one knows for certain that it is already 0;
    no action at all is required for this ``computation''.
    More generally,
    Frank~\cite{frank2017foundations} points out that reversible computing need not require a function $f$ to be fully injective, but only injective on inputs with positive probability.
}
In this example, $S$ is the simple binary entropy function on $p=\Pr_{\dist_i}[b=0]$, so $\entropy(\dist_{0}) = 1$,
$\entropy(\dist_{1}) = 0$, and the Landauer cost is 1 bit of entropy.
Importantly, \eqref{eq:intro-entropy-production} means that the Landauer cost does not necessarily imply positive entropy production.
By focusing instead on entropy production, we open the door to modeling and analyzing sources of heat dissipation in computation, over and above the Landauer cost, that cause $\EP$ to be positive.

The other two terms $\EP$ and $\EF$ account for how this local drop in entropy
relates to the outside universe. $\EP$, the \emph{entropy production}, represents the total increase in entropy of the universe.
The Second Law dictates that $\EP \ge 0.$
This is related to the \emph{local} decrease in entropy $\entropy(\dist_0) - \entropy(\dist_1)$ by the \emph{heat released} $\EF$.\footnote{The physics literature often refers to ``the environment'' as a ``heat bath'' or ``heat reservoir''. 
We note that the term ``heat'' in theoretical physics does not necessarily imply anything actually getting hotter,
and even in that community there can be forms of entropy flow not interpretable as heat transfer.
In principle $\EF$ does not necessarily represent heat dissipation.
It could correspond to discharging a rechargeable battery or increasing degrees of freedom that do not correspond to computational states (i.e. the device might write bits to a tape outside of the computer).
However, in most physical devices,
\emph{released heat} is the form entropy flow takes,
and we will use this term synonymously for entropy flow,
especially since it is less easily confused with entropy \emph{production}.
}
$\EF$ corresponds to the entropy change of the environment due to heat exchange (releasing heat increases the entropy of the environment) and encompasses the underlying physical implementation of our computation.
Note that to maintain the Second Law ($\EP \ge 0$), we must have $\EF \ge \entropy(\dist_0) - \entropy(\dist_1)$.
If the local entropy strictly dropped through a logically irreversible step, 
i.e., if $\entropy(\dist_0) > \entropy(\dist_1)$,
then in this case $\EF$ must be strictly positive to maintain $\EP \ge 0$.

The ability of randomized computation to have a negative Landauer cost has implications for energy use as well, allowing a computation to use randomized operations to counteract the heat released in other steps by logically irreversible operations.
As a simple example of these principles in action, and to illustrate how the laws of physics allow for computation with zero dissipated work \emph{even if} a logically irreversible step such as erasing a bit is performed, consider a situation in which we first generate a random bit and then erase it.
Imagine we have some physical sort of randomness/entropy available to generate a random bit $b \in \{0,1\}$.
Suppose we generate this bit in a memory location where we know in advance that a 0 is already written.
Applying equation~\eqref{eq:intro-entropy-production} to the first step, we have that $\dist_0$ is the point-mass $\Pr_{\dist_0}[b=0] = 1$,
whereas afterward we have $\Pr_{\dist_1}[b=0] = \Pr_{\dist_1}[b=1] = 1/2$.
Thus the Landauer cost $\entropy(\dist_0) - \entropy(\dist_1) = -1$ is actually \emph{negative}.
Physically, this corresponds to the fact that in principle, one removes $\boltzmann \temp \ln 2$ heat from the environment to randomize this bit,
releasing the same amount of heat on the subsequent erasure step that takes us from $\dist_1$ ($b$ is equally likely to be 0 or 1) back to $\dist_0$ ($b=0$).
This is essentially ``taking a loan and then paying it back''.

In the remainder of this paper, we focus on deterministic computation, in which every state has at most one successor.
Our model may describe only part of a larger physical device: other parts may perform randomized operations with negative Landauer cost and draw heat from the environment to offset heat released by the deterministic computation.
When $f:\domain \to \range$ is a (deterministic, single-valued) function being computed, we let $\landauer_f(\dist)$ 
(\Cref{defn:landauer-cost})
denote $\entropy(\dist) - \entropy(f(\dist))$,
where $f(\dist)$ denotes the distribution on $\range$ induced by distribution $\dist$ on $\domain$.
In this case, the data processing inequality (\Cref{prop:entropy-data-processing}) implies that 
$\landauer_f(\dist) \ge 0$; 
i.e., a deterministic process cannot increase the entropy of a distribution.

The previous example shows that a logically irreversible operation may have zero entropy production.
Next we discuss a ``dual'' example:
an overall operation can be logically reversible,
yet require positive entropy production due to ``locality'',
i.e., due to the fact that portions of the computation must be done in physically separate places,
so that despite entropy staying constant globally,
local operations must produce entropy due to lack of knowledge of rest of the state of the system.
Consider three entities Alice, Bob, and Eve.
Alice and Bob are at physically separate locations with no ability to communicate.
Eve has a uniformly random bit $b$,
and sends a copy of $b$ to each of Alice and Bob.
Alice must erase her bit, setting it to 0.
Globally, this is logically reversible,
since the mapping $(b,b) \mapsto (0,b)$ is injective,
but since Alice ``locally forgets'' $b$,
she must release a unit of heat into the environment to erase her bit,
even though far-away Bob remembers $b$.
Though this example may seem contrived and hand-wavey,
it is in fact a good example of what we sometimes called \emph{modularity cost}
(also known as \emph{modularity mismatch cost} in the stochastic thermodynamics literature, e.g.,~\cite{wolpert2019uncertainty}):
the fact that a self-contained computational module,
such as a gate in a circuit,
must release heat to compute an irreversible \emph{local} operation,
despite the fact that elsewhere in the circuit,
information about the gate's input is retained and so globally,
the operation is reversible.

To summarize, \eqref{eq:intro-entropy-production} applied to deterministic computation of a function $f$ obeys $\EP = \EF - \landauer_f(\dist)$.
$\landauer_f(\dist)$ is a purely information theoretic term depending only on $f$ and the input distribution $\dist$:
the loss of entropy $\entropy(\dist) - \entropy(f(\dist))$ between the input distribution $\dist$ to $f$ and the induced output distribution $f(\dist)$.
Reversible computing literature is concerned with making $\landauer_f(\dist) = 0$ for every computational step; 
in this paper we allow it to be positive,
though we freely and regularly use reversible computing techniques as well.
$\EF$ is the heat released into the environment during this computation;
this depends on the physics of the device computing $f$.
The Second Law of Thermodynamics states that $\EP \ge 0$,
so since $\EP = \EF - \landauer_f(\dist)$, 
this implies $\EF \ge \landauer_f(\dist)$.
From a complexity-theoretic perspective,
we wish to minimize $\EP$,
since it represents a physically irreversible loss of work (useful energy).
However, given a fixed $f$ to compute and distribution $\dist$,
we have no control over $\landauer_f(\dist)$,
so minimizing $\EF$ is equivalent to minimizing $\EP$.


\subsection{Prior work on thermodynamics of circuits}
\label{sec:intro-related-work-thermodynamics-of-circuits}
One perspective on the thermodynamic costs of circuits is provided by works of Michael Frank~\cite{frank2017foundations, DBLP:conf/rc/Frank18}.
In \cite{frank2017foundations}, Frank generalizes the conventional computer science notion of logical reversibility to include operations that, while not reversible in general, are so under a restricted subset of their inputs.
He denotes such operations as being ``conditionally reversible'' and shows that they can be implemented to release arbitrarily little heat, so long as the gates never receive an input violating the reversibility precondition.
Frank is primarily interested in conditional reversibility since it gives logically reversible gates that may be easier to implement than networks of Toffoli gates.
Our work can be seen as a further generalization of this idea; in addition to allowing conditionally reversible gates to release zero expected heat on valid distributions over inputs, we consider gates that are finely tuned for the actual distribution over inputs they expect to receive.
In \cite{DBLP:conf/rc/Frank18}, Frank also points out that $\EF$ is not quite the correct measure of thermodynamic cost.
He points out that the thermodynamic costs associated with Landauer's principle is not merely the fact that $\EF$ must be positive.
Rather it is that information which is dissipated as heat will rapidly become unrecoverable to the computational process, resulting in something similar to what we call a modularity cost.
Frank justifies the study of logically reversible computing since it is a sufficient condition to avoid Landauer related thermodynamic costs. 

Prior work by Wolpert, Kolchinsky, Caravelli, Yadav, and Yousef~\cite{WK20,YCW24,YYW25} investigated the thermodynamics of Boolean circuits using the theorems involving \emph{mismatch cost} --- a metric that quantifies the additional entropy production incurred by running a process on a different input distribution than the ``ideal'' one for given hardware.
In this case,
the physics of the device computing the gate may, 
in a sense,
be ``summarized information-theoretically'' by an \emph{optimal} input distribution $\dist^*$,
the distribution that minimizes the entropy production according to \eqref{eq:intro-entropy-production}.
They show lower bounds on the entropy production induced by running a circuit on a (potentially mismatched) input distribution $\dist$, where the lower bound depends only on the logical behavior of the gates with respect to an optimal input distribution $\dist^*$.
In other words, it is possible to prove non-trivial lower bounds on the $\EP$ of a circuit without knowing any details of the physical implementation of its gates other than their optimal distributions over its inputs.\footnote{
    The physical implementation does influence what is the optimal distribution.
    The significance of mismatch cost is that it can be measured knowing \emph{only} what is the optimal distribution, i.e., the physical implementation influences the mismatch cost \emph{only} through how it influences what is the optimal distribution.
}

The perspective of these works is that we know an optimal input distribution for a gate, and proceeds from there to measure mismatch costs (which are a lower bound on $\EP$).
Mismatch costs are a powerful tool for analyzing thermodynamic costs without needing to specify the underlying physics, 
but they are cumbersome to work with.
The thermodynamic accounting they provide requires viewing everything through the lens of distributions over state and costs are neither linear (with respect to the probabilities of inputs) nor local to the gates performing computation.
Moreover, the optimal distribution for a gate is a counter-intuitive way to specify its behavior and such a distribution is not unique even when the dynamics of the gate are deterministic.\footnote{There is a well defined way to handle this with mismatch costs, but it requires analyzing the pre-image of each possible output independently, but this can be messy to work with.}

Nevertheless, these methods are particularly nice as with mild assumptions on the behavior of gates in a circuit (i.e. modularity or time-invariant behavior), they can prove entropy lower bounds that are truly independent of the underlying hardware.


\subsection{Our results}
\label{sec:intro-results}
Here we provide a high level overview of our results and the techniques we use to obtain them.

\paragraph{The thermodynamic circuit model}
We introduce a model of circuits where each gate releases heat determined by its input value in order to perform 
logically irreversible computation.
For a circuit to compute a function $f$, we require it to map each input state $(x,\vec{0})$ to output state $(f(x),\vec{0})$;
the $\vec{0}$'s (which have different lengths if $|x| \ne |f(x)|$)
denote ``ancilla'' wires that may be used as extra storage during the computation.
In other words, the circuit must erase any ``dirty'' bits introduced during the computation including the original value of $x$.
The amount of heat a gate releases can depend on that gate's input; imagine that the gate must undergo some physical process to map each of its possible input states to the correct output.
Thus, we can assign a ``heat function'' to each gate $\heat: \{0,1\}^n \to \R_{\geq 0} \cup \{\infty\}$ that maps each input to an associated amount of released heat.
Technically, the amount of heat released by this process may not be deterministic, so one should think of these values as being expected costs,
though sometimes for brevity we may say a gate releases $\heat(x)$ heat on input $x$, understood to mean \emph{expected} heat.
Thanks to the linearity of expectation, this does not make a difference in the details of our model.
We say a heat function for a gate is \emph{valid} if it obeys the Second Law of thermodynamics (i.e. the $\EP$ in \eqref{eq:intro-entropy-production} is always non-negative) on every marginal input distribution.
From this perspective, we connect our model to mismatch costs in \cref{sec:rederive-mismatch-cost} of the appendix, and show that some of the results in \cite{WK20,YCW24,YYW25} have more intuitive proofs using heat functions.
We also show how our model can be used to prove new results about the thermodynamic costs of circuits.


Recall from \Cref{sec:intro-beyond-logical-reversibility} that our goal is to minimize $\EP$,
and that once $\landauer_f(\dist)$ is fixed,
minimizing $\EP$ is equivalent to minimizing $\EF$.
Despite this equivalence once $\landauer_f(\dist)$ is fixed, 
when we prove a one-way bound,
we generally prove lower bounds on $\EP$ 
(i.e., negative results stating that some computation \emph{must} dissipate a certain minimum amount of work that cannot be recovered)
and upper bounds on $\EF$
(i.e., positive results stating that it is \emph{possible} to achieve a certain computation while only releasing a small amount of heat).
For deterministic computation, $\EP \leq \EF$ so a lower bound on $\EP$ is also a lower bound on $\EF$ and an upper bound on $\EF$ is also an upper bound on $\EP$.
As is typical in theoretical computer science, we will use $O$-notation to suppress constant factors in our bounds.
Note that constant factors are important to track in the case where $\EF = \Omega(h(n))$ and $\landauer_f(\dist) = O(h(n))$.
Depending on the hidden constants, the $\EP$ for such a computation could be anywhere between $0$ and $\Omega(h(n))$.

There are two important caveats with our approach in this paper. First, when discussing $\EP$, we are implicitly assuming that the input distribution to our circuit is sampled independently from any larger context.
If a copy of the input to a circuit is retained somewhere outside of the circuit, then the drop in entropy in \eqref{eq:intro-entropy-production} should actually be zero and we get that $\EP = \EF$.
When proving lower bounds on $\EP$ this is a non-issue, as the additional context can only increase the $\EP$.
When proving upper bounds on $\EP$, we specify that we are interested in the \emph{circuit entropy production}, which we define as $\EP$ under this assumption on the input distribution.
Second, when giving upper bounds on either the $\EF$ or the $\EP$ of a circuit, we are operating under the ``axiom'' that every gate could in principle be built with any valid heat function.
While the validity constraint protects our gates from violating the Second Law of thermodynamics, other physical or engineering constraints may theoretically prevent such a gate from existing.
For example, speed limit theorems~\cite{PhysRevLett.121.070601} imply that, with bounds on transition rates, $\EP$ must be strictly positive in order for gates to compute their outputs in finite time\footnote{As is common in physics, we sidestep this issue by assuming that all gates are computed \emph{quasistatically} (i.e. infinitely slowly). In the limit as time to implement a gate goes to infinity, this lower bound on the $\EP$ goes to zero. We could still implement all gates in a circuit in finite time while making this cost a small constant independent of the number of gates, but we choose to do things quasistatically to simplify our analysis.}  and current CMOS technology is far from the thermodynamic limit imposed by validity.
Thus our upper bounds should be interpreted under the ``spherical cow'' assumption that validity is the only constraint on realizable heat functions.


\paragraph{Deriving optimal heat functions}
In \Cref{sec:opt-heat-func}, \Cref{prop:ef-min-heat-fn} shows how to derive the optimal valid heat function for a known distribution over inputs.
In particular if $p_x$ is the probability of input $x$ and $q_y$ is the probability of output $y$, then any optimal heat function for computing $f$ on this distribution over inputs must satisfy $\heat(x) = \log(q_{f(x)}/p_x)$ on all $x$ where $p_x > 0$.
This in turn is a key ingredient in proving our results in Sections \ref{sec:pp-hard}, \ref{sec:heat-functions-fixed-circuits}, and  \ref{sec:zero-entropy-compaction-and-sorting-circuits}.

\paragraph{Hardness results and approximation algorithms for computing optimal heat functions}

\Cref{sec:pp-hard} is concerned with computing the optimal heat functions for gates in a given Boolean circuit (so as to minimize $\EP$), assuming the i.i.d.~uniform distribution for the circuit's inputs.
We first show that the problem of computing the optimal heat function for each gate (so as to minimize total entropy production of the entire circuit) is unfortunately $\PP$-complete
(\Cref{lem:finding-opt-Q-is-PP-hard,lem:finding-opt-Q-in-PP}).
Since \Cref{prop:ef-min-heat-fn} shows how to compute the optimal heat function for a \emph{single} gate and a given input distribution to that gate,
this hardness follows essentially from the fact that computing the induced input distribution for a given gate in a circuit, knowing that the circuit's inputs are i.i.d.~with probability $1/2$, is itself $\PP$-complete.
The (straightforward) reduction is from the $\PP$-complete language $\varLanguage_\halfAccept$ of Boolean circuits $\circuit$ where at least half the assignments satisfy $\circuit$:
by adding a single erasure gate $g$ on the output, with an optimal heat function $Q_g$, $Q_g(1) \le 1$ if and only if $\Pr[\text{output of } \circuit \text{ is }1] \ge 1/2$, 
i.e., if and only if 
$\circuit \in \varLanguage_\halfAccept.$
Containment in $\PP$ directly follows from an algorithm that accepts or rejects depending on the relative probabilities a gate receives a specified input $x$ or another input $x' \neq x$ where $f(x') = f(x)$.


Since $\PP$-hardness rules out polynomial-time exact algorithms,
this motivates the question of whether polynomial-time approximation algorithms exist.
We show that this depends crucially on the form of error.
With \emph{multiplicative} error,
i.e., for some constant $c > 1$, to compute heat functions with expected value at most $c$ \emph{times} optimal, 
then this is $\NP$-hard, i.e., untractable unless $\NP \subseteq \BPP$.
There is an easy argument for this:
suppose we have a Boolean formula $\phi$ (a special type of circuit) and want to know whether $\phi$ is satisfiable.
If we then add an erasure gate $g$ to $\phi$'s output,
then for an optimal heat function $\heatoptimal_g$,
if $\phi$ is unsatisfiable, $\heatoptimal_g(0) = 0$,
i.e., it costs 0 heat to erase the output since we know it is 0.
If $\phi$ is satisfiable, then $\heatoptimal_g(0) > 0$.
Thus any approximation algorithm outputting heat value $q$ will tell us whether $\phi$ is satisfiable depending on whether $q > 0.$
However, this is a somewhat degenerate case,
in the sense that \emph{any} optimization problem permitting an optimal value of 0 has the property that 
no multiplicative approximation is good unless it is exact. The more interesting case is when a wire might have both bits 0 or 1 ($\support(\dist) = \{0,1\}$),
so that we have a positive optimal expected heat $\E_{x \sim \dist}[\heatoptimal(x)] > 0$.
We can in fact prove $\NP$-hardness of multiplicative approximation continues to hold even in this case
(\Cref{lem:mult-approx-hard}).

No $\NP$-hard problem has a polynomial-time \emph{randomized} algorithm unless $\NP \subseteq \BPP$;\footnote{
    The famous and widely-believed conjecture that $\P \ne \NP$ is the statement that $\NP$-complete problems such as $\SAT$ have no polynomial-time deterministic algorithms that are guaranteed to get the answer correct.
    The statement $\NP \not \subseteq \BPP$ generalizes this to the also widely-believed conjecture that such problems have no polynomial-time \emph{randomized} algorithms that are allowed to get the answer wrong with small probability.
}
\Cref{lem:mult-approx-hard} implies this applies to the problem of approximating optimal heat functions with multiplicative error.
However, considering \emph{additive} error in approximation,
we show there is a polynomial-time randomized algorithm that outputs a heat function at most $\epsilon$ \emph{plus} optimal with probability of failure at most $\delta$;
the running time is polynomial in the input, $1/\epsilon$, and $1/\delta$
(\Cref{lem:additive-approx-easy}).
The idea is simply to sample uniformly random inputs to the whole circuit $\circuit$ 
to empirically estimate a distribution on inputs to a gate $g$ in $\circuit$,
from which a heat function is computed via \Cref{prop:ef-min-heat-fn}.

\paragraph{Bounds on thermodynamic costs of fixed circuits}
Having shown the difficulty in computing optimal heat functions for arbitrary given circuits, \Cref{sec:heat-functions-fixed-circuits} then analyzes particular circuits and derives both lower bounds on $\EP$ and upper bounds on $\EF$.
\Cref{sec:heat-functions-AND-entropy-depth-tradeoff} begins with a warm-up exercise analyzing the computation of the $n$-input AND function $x_1 \wedge x_2 \wedge \ldots \wedge x_n$ on uniformly distributed inputs in $\{0,1\}^n$, using only input-preserving binary-AND gates of the form $\bigwedge(x_1,x_2,z) = (x_1, x_2, x_1 \wedge x_2)$ and erasure gates $E(x) = 0$.
We show the (possibly) counterintuitive result that the most natural implementation, a tree of AND gates of depth $O(\log n)$, along with appropriate erasure gates to enable the circuit to cleanly compute the 1-bit output $x_1 \wedge x_2 \wedge \ldots \wedge x_n$, has circuit entropy production $\Omega(n)$,
whereas the ``serial'' variant with depth $\Theta(n)$, can be done with circuit entropy production that is only $O(1)$.

\Cref{sec:batcher-bits} then analyzes the $O(n \log^2 n)$-gate, $O(\log^2 n)$-depth Batcher even-odd mergesort sorting network on Boolean inputs.
When applied to Boolean values, sorting networks are circuits where every gate is a 2-input, 2-output \emph{comparator} gate: computing the identity on inputs $00,01,11$,
but on input $10$,
swapping the bits to output $01$.
The main results of this section, \Cref{thm:batcher-heat-released-upper-bound,thm:batcher-ep-omega-n},
show an $O(n)$ upper bound on $\EF$ using only $O(\log n)$ distinct heat functions and a matching $\Omega(n)$ lower bound on $\EP$ over a uniform distribution of inputs in $\{0,1\}^n$ when using the optimal valid heat function on every gate.
This bound is notable in that it is asymptotically smaller than the number $\Theta(n \log^2 n)$ of gates in the circuit, despite none of the gates performing logically reversible operations.
The $O(n)$ upper bound (\cref{thm:batcher-heat-released-upper-bound}) is proven by slicing the network into regions of gates that all either (1) contain a small number of gates or (2) have one input that is overwhelmingly likely.
In contrast, the $\Omega(n)$ lower bound (\cref{thm:batcher-ep-omega-n}) follows directly from observing that the first $3$ layers of the Batcher's odd even mergesort network consist of $n/4$ independent copies of sorting networks on $4$ bits each and that each of these networks has $\Omega(1)$ entropy production.
In \cref{sec:batcher-overwriting} we prove a similar upper bound on the $\EF$ for sorting bits in a more restricted overwriting model of circuits considered in other works (i.e., \cite{YCW24,YYW25}).

\Cref{sec:batcher-integers} analyzes the more interesting case of the Batcher network applied to $n$ distinct values drawn uniformly at random from an ordered set $\domain$.
In this case, we are able to show that the entropy production must be $\Omega(n \log^2 n)$ which is tightly matched by a trivial $O(n \log^2 n)$ upper bound on $\EF$ given using any valid heat function where the expected released heat is $O(1)$ on all inputs.
This result is rather technical to prove, but the key ingredient is \cref{lem:batcher-integer-swap-bound}, which shows that for a constant fraction of the gates in the network, if that gate receives input $(x_1,x_2)$ with probability $p_{x_1,x_2}$ then whenever $p_{a,b} > 0$ it must be the case that $1/\alpha \leq p_{b,a}/p_{a,b} \leq \alpha$ for some universal constant $\alpha$.
Since such gates must release $\Omega(1)$ expected heat each, we get an $\Omega(n \log^2 n)$ lower bound on $\EF$.
Since the drop in entropy associated with sorting is only $O(n \log n)$, we can conclude that $\EP$ is $\Omega(n \log^2 n)$ in \cref{cor:batcher-integer-any-set}.

\cref{sec:better-than-batchers} presents a straightforward method for constructing a sorting network with better thermodynamic costs than Batcher on $n$ distinct inputs drawn uniformly at random from an ordered set $\domain$.
This circuit has $\EF$ that is $O(n \log n)$ which is asymptotically optimal.\footnote{Technically \cite{AjtaiKS83} automatically gives us this result since the network only has $O(n \log n)$ gates, but this network is known to have prohibitively large constant factors that render it impractical.}
Our circuit works by first using an $O(n \log n)$ gate network developed by Leighton and Plaxton~\cite{LP98} that sorts all but a small fraction of the inputs.
After running that circuit, we apply Batcher to its output.
The result, as we show in \cref{thm:pre-sort-helps-upper-bound}, is a sorting network with $O(n \log^2 n)$ gates that only releases $O(n \log n)$ expected heat.
By first using Leighton and Plaxton's network, the input to Batcher's will already be sorted most of the time.
Thus we can use a heavily biased heat function for the gates in Batcher's that releases less heat whenever the input is already sorted.

\paragraph{Compaction and sorting with zero entropy production}

\cref{sec:clean-compaction-zero-ep}
shows that we can do better than Batcher by designing a bit-sorting 
(a.k.a., \emph{bit-compaction}) circuit
specifically to minimize $\EP$,
and in fact can achieve zero circuit entropy production 
(\cref{thm:compaction-zero-entropy}) on any exchangeable distribution over inputs.
Note that this construction is not a sorting network, as it uses ancilla wires and more general gates than only comparators.

We do so by using standard Bennett-style ``compute/copy output/uncompute''
(applied to circuits)
to convert any circuit computing the mapping $x \mapsto f(x)$
to one reversibly computing 
$(x,\vec{0}) \mapsto (x,f(x))$,\footnote{We actually perform a slightly different logically reversible mapping for this problem to efficiently handle payloads, see the proof of \cref{thm:compaction-zero-entropy} for more details.}
with only a constant increase in gate and depth complexity.\footnote{
    This is useful because, as mentioned,
    reversible computation is automatically zero $\EP$ assuming a heat function of $Q(x)=0$ for all gates.
}
Once we have $(x,f(x))$ in memory, we give a way to pick heat functions for erasing each bit of $x$ such that the expected heat released is exactly the drop in entropy from mapping $x$ to $f(x)$.
The choice of heat function for erasing each $x_i$ is determined based on the value of $f(x)$ and the suffix of $x$ after index $i$, which can be computed and uncomputed in a logically reversible fashion even after $x_i$ has been erased.
We design such a circuit based on any sorting network $C$ with $t(n)$ gates and depth $d(n)$ that compacts $n$ inputs with $m$ bit payloads using only constant arity gates with only $O(t(n)m + n^2)$ gates, $O(d(n) \log m + n \log n)$ depth, and $O(t(n) + nm)$ ancilla wires.
In \Cref{sec:sorting-zero-ep},
we show how this approach can be extended to give zero circuit entropy production circuits for sorting on any exchangeable distribution over inputs.

\paragraph{New proofs of known results on mismatch cost and circuit mutual information using heat functions}

Following some basic lemmas in \cref{sec:basic-lemmas} of the appendix and the aforementioned analysis of Batcher's sorting algorithm in \cref{sec:batcher-overwriting}, we prove new results connecting heat functions to mismatch costs and reprove existing results on mismatch costs using the
language of heat functions in
\Cref{sec:rederive-mismatch-cost,sec:modularity}.
This is intended to share these insights from stochastic thermodynamics with the computer science community and better connect our heat function model with prior work on the thermodynamics of circuits.
\Cref{sec:rederive-mismatch-cost} concerns \emph{mismatch cost}, which is informally defined as the \emph{extra} entropy production incurred by using a gate with a heat function optimal for distribution $\dist^*$, when in fact $\dist \ne \dist^*$ is the true distribution.
The mismatch cost of course depends on the physical implementation of the gate.
Wolpert and Kolchinsky~\cite{WK20} showed that the way this physics influences the extra entropy production is entirely contained in its optimal distribution $\dist^*$, in the sense that the mismatch cost is precisely the information-theoretic terms $\DKL(\dist || \dist^*) - \DKL(f(\dist) || f(\dist^*))$, where $\DKL$ refers to the Kullback-Leibler divergence, see \cref{sec:prelim}.
We show how this result admits a clean proof using heat functions.

In \Cref{sec:modularity},
we give a simple proof of a result of Wolpert and Kolchinsky~\cite{WK20},
who showed that the entropy production of
applying a gate in a circuit is at least the drop in mutual information between the wires in this gate and
the state of the remaining wires of the circuit.
We reprove this result (\Cref{lem:mutual-info-positive-ep}) in the language of heat functions 
showing that it is an immediate consequence of the locality of valid heat functions,
i.e., the fact that heat functions reference a single gate independently of the correlation of its inputs to other gates in the circuit.

\section{Preliminaries}
\label{sec:prelim}
We use $\log$ to denote $\log_2$, 
and fix $2^{-\infty} = 0$, $-\log 0  = \infty$, as well as $z/0 = \infty$ when $z>0$.
Let $\dist$ be a distribution over domain $\domain$ and $x \in \domain$. We use $\Pr_{\dist}[x]$ to denote the probability that a random variable $X \sim \dist$ has value $x$ (i.e. $\Pr_{X \sim \dist}[X = x]$).
Let $S \subseteq \domain$.
We will also use $\dist_{|S}$ to denote the distribution $\dist$ conditioned on the sample being in $S$.
More formally, $\Pr_{\dist_{|S}}[x] = \Pr_{X \sim \dist}[X = x | X \in S]$.
When $S \cap \support(\dist) = \emptyset$, we say that $\dist_{|S}$ is undefined.
We use $\E_{x \sim \dist}[f(x)]$ to denote $\sum_{x \in \support(\dist)} \Pr_{\dist}[x]f(x)$.

\begin{definition}
The \emph{Shannon entropy} of a distribution $\dist$ is given by
\[
\entropy(\dist) 
= -\E_{x \sim \dist} \left[ \log \Pr_\dist[x] \right]
= - \sum_{x \in \support(\dist)} \Pr_\dist[x] \log \Pr_\dist[x].
\]
\end{definition}

We use $e_i$ to refer to a vector with the value $1$ in the $i$'th index and zero everywhere else.
When thinking about function computation, we will often use $\domain$ to refer to the (finite sized) domain and $\range$ to refer to the (finite sized) range.
When needed, we will often specify $\domain = \alphabet^n$ and $\range = \alphabet^m$ for some finite alphabet $\Sigma$ containing an element $0$.
We use $x_i$ to denote the $i$'th coordinate of $x$ and $f_i$ to denote the function $i$'th coordinate function of $f$.
For $\idx \subseteq [n]$ we use $x_{\idx}$ to denote the length-$|\idx|$ vector consisting of $x$ restricted to its coordinates on indices in $\idx$,
and $\Sigma^I$ to denote the set of all such vectors.

In our model, gates apply a function $f$ to some subset $I$ of the $\varSpace$ total wires in the circuit.
It is useful to extend this function $f$ to an $f_\uparrow$ that acts on all $\varSpace$ wires of the circuit, but only modifies the wires in $I$ according to $f$, leaving the other wires unchanged.
Thus we define the following extension of a function $f$.

\begin{definition}
[lifting function $f_{\uparrow}$]
\label{defn:lifting-f}
    If $f: \alphabet^\idx \to \alphabet^\idx$
    for some finite set $\idx ={i_1, \ldots, i_k} \subseteq [\varSpace]$,
    then we can define the function $f_{\uparrow}: \alphabet^{\varSpace} \to \alphabet^{\varSpace}$ as:
    \begin{displaymath}
        (f_{\uparrow})_i(x) = \begin{cases}
            f_i(x_{\idx}) & i \in \idx\\
            x_i & i \not \in \idx.
        \end{cases}
    \end{displaymath}
\end{definition}

The following definition gives us clean notation for the distribution over states after applying a gate computing some function $f$.
\begin{definition}
[induced output distribution]
\label{defn:induced-output-distribution}
    If $f: \domain \to \range$ and $\dist$ is a distribution over $D$, 
    then we use $f(\dist)$ to denote the distribution over $\range$ produced by first sampling $x \sim \dist$ and then outputting $f(x)$.
\end{definition}

Note that we have defined two distinct ways to apply functions to distributions, reflecting whether the co-domain is a scalar or an entire distribution: 
$\landauer_f(\dist) \in \R$ is an example
of the former, and \Cref{defn:induced-output-distribution},
a new distribution resulting from applying the function $f$ to a sample, is an example of the latter. 
Our intended meaning should be clear based off the context. However,
in general we will use capital letters such as $\entropy$ and $\landauer_f$ are functions that compute some property of the input distribution whereas lower case letters are used for mappings that define new distributions.

\begin{definition}
    The \emph{mutual information} between two random variables $X,Y$ is given by:
    \begin{displaymath}
        \mutualInformation(X;Y) = \entropy(X) + \entropy(Y) - \entropy(X,Y).
    \end{displaymath}
\end{definition}

\begin{definition}
[Data processing inequality, \cite{Cov06}]
\label{prop:entropy-data-processing}
    Let $f$ be any function with domain $\domain$ and $\dist$ be any distribution over $\domain$.
    Then $\entropy(\dist) \geq \entropy(f(\dist))$.
\end{definition}

\begin{definition}
The \emph{Kullback-Leibler divergence} between two distributions $\dist$ and $\dist'$ is 
\[
\DKL(\dist || \dist') = \E_{x \sim \dist}\left[ \log \frac{\Pr_\dist[x]}{\Pr_{\dist'}[x]}\right].
\]
\end{definition}

\begin{proposition}[Data processing inequality for KL divergence, \cite{Cov06}]\label{prop:data-processing}
    Let $f: \domain \to \range$ be an arbitrary function and $\dist, \dist'$ be two distributions over $\domain$.
    If $\support(\dist) \subseteq \support(\dist')$, then:
    \begin{displaymath}
        \DKL(\dist || \dist') \geq \DKL(f(\dist) || f(\dist')).
    \end{displaymath}
\end{proposition}

\begin{definition}
    If $X$ and $Y$ are random variables over $\domain$ and $\range$ respectively, then the \emph{conditional} entropy of $X$ \emph{given} $Y$ is $$\entropy(X | Y) = \sum_{x \in \domain, y \in \range} \Pr[X=x,Y=y] \log\left(\frac{\Pr[Y=y]}{\Pr[X=x,Y=y]}\right)$$.
\end{definition}
\begin{proposition}
    If $X$ is a random variable and $f$ is a deterministic mapping, then $$\entropy(X | f(X)) = \entropy(X) - \entropy(f(X))$$.
\end{proposition}

\begin{proposition}[Gibbs' Inequality]\label{prop:gibbs-inequality}
    Let $\dist, \dist'$ be two distributions over $\domain$.
    Then:
    \begin{displaymath}
        \DKL(\dist || \dist') \geq 0
    \end{displaymath}
    with equality only when $\dist=\dist'$.
\end{proposition}


\begin{definition}
    Let $f: \domain \to \R_{\geq 0} \cup \{\infty\}$ be a function.
    We define the \emph{Gibbs distribution} $\gibbs_f$ over $\dist$ with respect to $f$, where for any $x \in \domain$,
    (recall $2^{-\infty} = 0$)
    \[
    \Pr_{\gibbs_f}[x] 
    = \frac{2^{-f(x)}} 
    {\sum\limits_{x' \in \domain} 2^{-f(x')}}.
    \]
\end{definition}

We prove the following lemma in \cref{sec:basic-lemmas}.
\begin{restatable}{lemma}{gibbsopt}
\label{lem:gibbs-opt}
    Let $f: \domain \to \R_{\geq 0} \cup \{\infty\}$ be a function.
    For any set $S \subseteq \domain$ where $(\gibbs_f)_{|S}$ is defined, it is the unique choice that minimizes $\E_{x \sim \dist_{|S}}[f(x)] - \entropy(\dist_{|S})$ over all distributions supported on $S$.
\end{restatable}

We use $\binaryEntropy(p) = p \log(1/p) + (1-p) \log(1/(1-p))$ to denote the
\emph{binary entropy} of a Bernoulli random variable with probability $p$.
We use the following standard inequality on
$\binaryEntropy$.

\begin{proposition}
    \label{lem:binentropy-ineq}
        $\binaryEntropy(p) \leq p \log(e/p)$.
\end{proposition}
\begin{proof}
    \begin{align*}
        \binaryEntropy(p) & = p \log(1/p) + (1-p) \log(1/(1-p))                                        \\
                          & = p \log(1/p) + (1-p) \frac{\ln(1/(1-p))}{\ln 2}                           \\
                          & \leq p \log(1/p) + \frac{p}{\ln 2} \qquad \textrm{since $\ln(x) \leq x-1$} \\
                          & = p \log(1/p) + p \log e                                                   \\
                          & = p \log (e/p).
        \qedhere
    \end{align*}
\end{proof}

We also require the following standard inequality.

\begin{proposition}[Hoeffding's inequality~\cite{hoeffding1963probability}]
    \label{prop:hoeffding}
    Let $\distUniform$ be the uniform distribution on $n$ bits and $t \leq n/2$.
    Then
    \begin{displaymath}
        \Pr_{x_1, \ldots, x_n \sim \distUniform}\left[\sum_{i \in [n]} x_i \leq n/2 - t\right] = \Pr_{x_1, \ldots, x_n \sim \distUniform}\left[\sum_{i \in [n]} x_i \geq n/2 + t\right] \leq e^{-2t^2/n}.
    \end{displaymath}
\end{proposition}

\subsection{Our model}
While we will focus on circuits in this document, the ideas used in this section are sufficiently general to apply to any computational model that performs deterministic local transformations at fixed time steps.
\paragraph{Heat and entropy flow}
To simplify our math, we will assume that all computation is performed at a fixed temperature $\temp$ and units are chosen such that $k_B \temp \ln 2 = 1$.
Throughout this paper we analyze how much entropy a computational device must release into its environment to obey the second law of thermodynamics.
We will be assuming that each logical state of the machine has the same energy level and number of microstates.\footnote{
    We use the term ``microstate'' in this paper to mean the various different physical states a machine can be in while still representing the same logical state.
    As a concrete example, a Boolean gate implemented with semiconductors has a logical state identified by the bits on the input and output wires, which generally will correspond to many different physical microstates, depending on the precise positions and motion of electrons in the wires.
}
In principle this entropy can take many forms, but for our analysis it is safe to assume that it is always released as heat.
This motivates what we will call \emph{heat functions} as a way to abstract away messy details for how computations release excess entropy.

\paragraph{Thermodynamic circuits}
We start by giving two working definitions of circuits that we use in this paper.
We start with our model of circuits, which is heavily inspired by `serial reinitialized' model introduced in \cite{WK20}.

This model of circuits should be familiar to readers acquainted with the quantum circuit model.
For those unfamiliar with such models, recall that a standard model defines a Boolean circuit as a directed acyclic graph, where each edge represents a wire connecting exactly two gates (nodes).
Instead, in our model, all
information is stored in a fixed set of $n$ wires, and gates operate on a subset (typically of $O(1)$-size)
of these wires to update their values.
However, unlike quantum circuits, we do not require gates to be unitary.
Instead, they can perform any arbitrary mapping on their associated wires.
While we use the language of circuits with wires and gates, this formulation works equally well to describe straight line programs where wires correspond to registers and gates correspond to operations on these registers.
While we define such circuits in abstract, as is common in complexity theory, it will often be interesting to restrict them to specific gate sets 
(thus each gate would operate on $O(1)$ wires independently of the total number of wires).
This rules out the possibility of implementing arbitrary functions with a single gate acting simultaneously on all wires of the circuit and lets us analyze the thermodynamic costs of computations with restricted gate sets.\footnote{
\label{foot:constant-size-gates-assumption}
    As with the standard DAG-based model of Boolean circuits, if one allows larger gates than $O(1)$, even up to a single gate taking all inputs and computing an arbitrary Boolean function in one step, then this will not give informative complexity measures. In this paper, restricting gates to operate on $O(1)$ wires will enable stronger lower bounds on entropy production,
    in an analogous sense that restricting standard Boolean circuits to have $O(1)$-size gates enables nontrivial lower bounds on gate complexity.
}

\begin{definition}
\label{def:sr-circuit}
    A deterministic \emph{circuit} $\circuit$ on $n$ inputs with $\varTime$ gates and $\varSpace \geq n$ wires is described as an 
    ordered collection
    of $\varTime$ gates $g_t = (\idx_{g_t}, f^{g_t})$ for $t \in \{1, \ldots, \varTime\}$.
    Each gate specifies a set $\idx_{g_t} \subseteq [\varSpace]$ and a function $f^{g_t}: \alphabet^{\idx_{g_t}} \to \alphabet^{\idx_{g_t}}$.
    On input $x_1, \ldots x_n \in \alphabet^n$ the circuit starts in the state $\circuit^0 \in \alphabet^\varSpace$ where $\circuit^0_i = x_i$ for $i \in [n]$ and $\circuit^0_i = 0$ for $i \in [\varSpace] \setminus [n]$.
    The \emph{state} of the circuit at step $t$ is given by $\circuit^t = f^{g_t}_\uparrow(\circuit^{t-1})$.
    We say that wire $i \in [\varSpace]$ of $\circuit$ \emph{evaluates to} $x \in \alphabet$ when $\circuit^T_i = x$.
\end{definition}

Note that \Cref{def:sr-circuit} allows to have more wires ($\varSpace$) than input bits ($n$), to allow extra ``workspace'' for computation. It requires all non-input wires to be initially 0. The next definition captures our notion of computing output, which similarly allows the number of output bits $m$ to be strictly less than the number of wires $\varSpace$, and similarly requires that at the end of the computation, all non-output wires are 0.\footnote{
    This requirement of zeroing out non-output wires is particularly important in the context of thermodynamics of computing, because it forces on the circuit the obligation to erase any non-output bits left over from the computation. Absent this requirement, a circuit could appear na\"{i}vely to produce less entropy than one would expect, hiding the cost of erasing those leftover ``garbage'' bits. However, if one wished to reuse the circuit for a subsequent computation, since those bits would at some point need to be set to 0 to prepare the circuit for the next input. The definition of clean computation forces the circuit to prepare itself for reuse. (Note that we do \emph{not} account for the thermodynamic cost of writing the input, recording the output, or replacing the old output with the new input.)
}

\begin{definition}
\label{defn:cleanly-compute}
    We say that a circuit $\circuit$ \emph{cleanly computes} a function $f: \alphabet^n \to \alphabet^m$ if on every input $x \in \alphabet^n$, wire $i$ of $\circuit$ evaluates to $f_i(x)$ if $i \leq m$ and evaluates to $0$ if $i > m$.\footnote{
    The choice to make the \emph{first} $m$ wires correspond to output is arbitrary.
    If a circuit places its output on other wires, logically reversible swap gates can move the output to the first $m$ wires with no extra entropy production.
}
\end{definition}
An example of a circuit that cleanly computes a the and of its inputs is shown in \cref{fig:and-example-circuit}.

\begin{figure}[tb]
    \centering
    \includegraphics[width=0.8\linewidth]{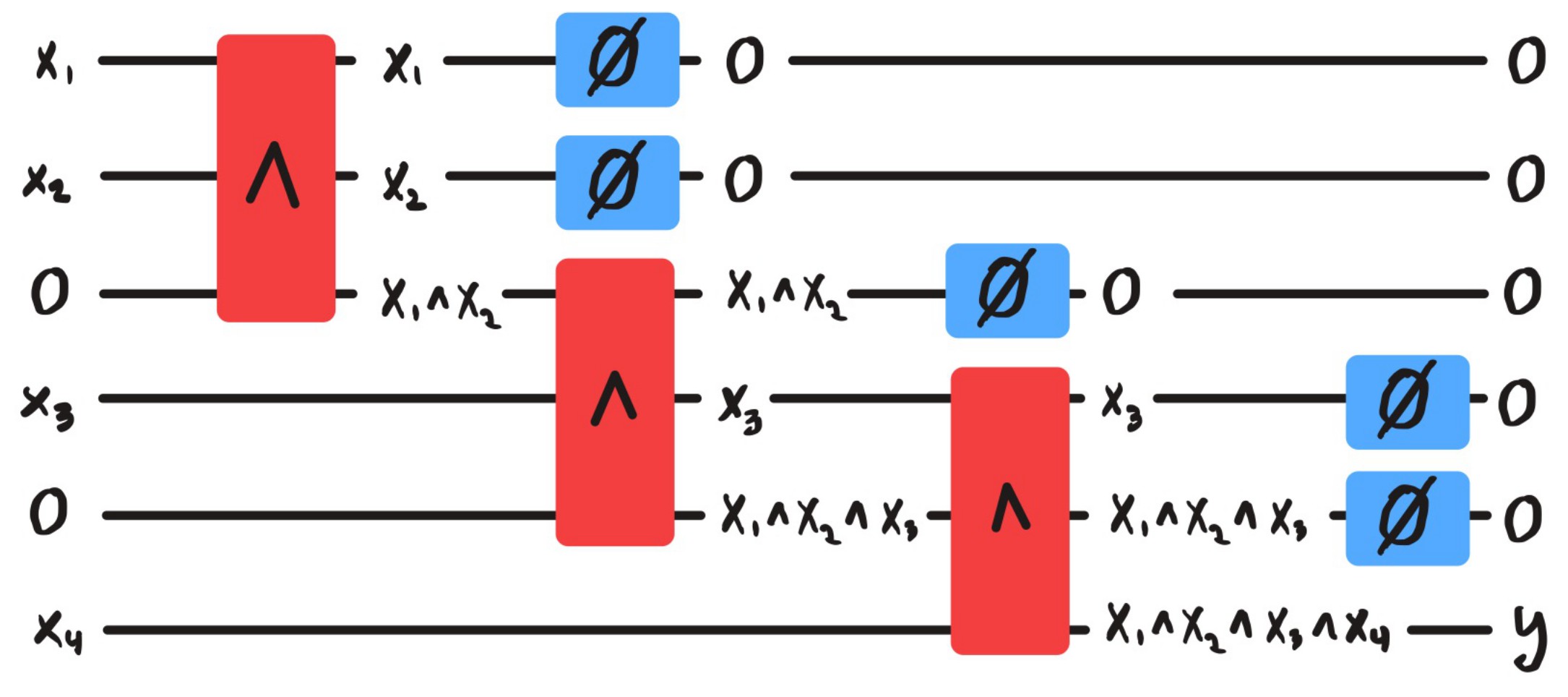}
    \caption{An example of a circuit that cleanly computes $y=x_1 \wedge x_2 \wedge x_3 \wedge x_4$ up to permutation of the wire indices.
    Here the red $\wedge$ gates perform the mapping $(a,b,c) \to (a,b,a\wedge b)$ and the blue $\emptyset$ erasure gates perform the mapping $a \to 0$.
    Each gate overwrites its inputs (values on wires to the left of the gate) with its outputs (values on wires to the right of the gate).}
    \label{fig:and-example-circuit}
\end{figure}


Next we define how heat functions can augment a circuit to model its thermodynamic costs.

\begin{definition}
    A \emph{thermodynamic circuit} is a circuit $\circuit$ where each gate $g$ computing a function with domain $\domain$ is assigned a \emph{heat function} $\heat_g:\domain \to \R_{\geq 0}\cup \{\infty\}$.
    On input $x\in \domain$, the gate releases $\heat_g(x)$ expected heat into the environment.
\end{definition}
\begin{remark}
    As we are fixing the temperature, it is reasonable to view $\heat$ as expected heat release.
    A more generally accurate description of $\heat(x)$ would be the expected entropy flow out of the system on input $x$.
\end{remark}

As we will see in \cref{sec:modularity}, it is important to note that $\heat$ is only allowed to depend on the inputs to a gate.
This represents the fact that a gate is modular and oblivious to the larger state of the circuit.
Otherwise, these heat functions can in principle be an arbitrary function of $\domain$.
A gate that releases $\infty$ heat on some input $x$ can be thought of being designed so it is not physically possible to specify $x$.
This conveys the notion of restricted reversibility in \cite{frank2017foundations}, where Frank points out that some logically irreversible operations can become reversible with such restrictions on their domain.
However, not all choices of heat functions are valid.
The Second Law of Thermodynamics (or Landauer's costs) places constraints on the possible heat functions that can be assigned to each gate.

\begin{definition}
[Landauer cost, \cite{Lan61}]
\label{defn:landauer-cost}
    For any function $f:\domain \to \range$ and distribution $\dist$ over $\domain$, we define the \emph{Landauer cost} function of $f$ on $\dist$:
    \begin{displaymath}
        \landauer_f(\dist) = \entropy(\dist) - \entropy(f(\dist)).
    \end{displaymath}
\end{definition}

The following is a corollary of the Second Law of Thermodynamics as applied to computation of functions.
As a physical law, it cannot be proven in full generality, though it can be proven within specific models.
We state it formally to motivate the subsequent \Cref{def:valid-heat}.

\begin{proposition}
[Second law, \cite{Lan61,Seifert2005Trajectory}]
\label{prop:second-law}
    Let $f: \domain \to \range$ be a function.
    Any physical process that maps every state $x \in \domain$ to the state $f(x)$ must 
    release at least $\landauer_f(\dist)$ expected heat when $x$ is sampled from distribution $\dist$ over $\domain$.
\end{proposition}

In the context of physical systems evolving according to $f$,
this is simply the lower bound on the expected heat release (or more generally entropy flow) out of a system where all logical configurations contain the same number of microstates as given by
the second law of thermodynamics.
The following definition essentially formalizes our notion of a heat function being compatible with the second law.

\begin{definition}
\label{def:valid-heat}
    Let $g$ be a gate in circuit computing a function 
    $f$ with domain $\domain$.
    A heat function $\heat:\domain \to \R_{\geq 0} \cup \{\infty\}$ is \emph{valid} for $g$ if and only if, for every input distribution $\dist$ over $\domain$, we have that:
    \begin{displaymath}
        \landauer_f(\dist) \leq \E_{x \sim \dist}[ \heat(x)].
    \end{displaymath}
\end{definition}

\Cref{def:valid-heat} ``directly'' captures what we mean by a valid heat function---no matter the distribution it obeys the Second Law---and as stated intuitively appears to require validation over all input distributions over $\domain$.
The following equivalent formulation of \Cref{def:valid-heat} shows that validity in fact is solely a property of the value $Q(x)$ on \emph{individual} inputs $x$.
This is the deterministic specialization of Turgut’s Theorem 1 (Eq. (1)) in the simplified setting where each logical state has the same internal entropy~\cite{turgut2009relations}.
Zhu~\cite{zhu2024} subsequently rederived the same stochastic realizability criterion from a differential fluctuation theorem; with deterministic dynamics, Zhu’s condition likewise reduces to \Cref{lem:turgut-lemma}.
We give a self-contained proof of \Cref{lem:turgut-lemma} in \Cref{sec:basic-lemmas}.

\begin{restatable}{lemma}{turgut}
\label{lem:turgut-lemma}
    Let $g$ be a gate in a circuit computing a function $f:\domain \to \range$.
    A heat function $\heat:\domain \to \R_{\geq 0} \cup \{\infty\}$ is valid for $g$ if and only if, for every $y \in \range$,
    \begin{displaymath}
        \sum_{x \in f^{-1}(y)} 2^{-\heat(x)} \leq 1.
    \end{displaymath}
\end{restatable}

In practice, there is a large gap between valid heat functions 
(those obeying the Second Law)
and physically realizable heat functions 
(which may release far \emph{more} heat than the minimum dictated by the Second Law, 
e.g., current CMOS technology).
Additional heat is released by computation to drive the physical processes that implement the logical mapping of the gate in the desired direction in finite time \cite{PhysRevLett.121.070601}.
However, valid heat functions let us prove lower bounds on the thermodynamic costs of a computation and provide a limit for the physical implementation of the computation.
In this paper, we state upper bounds on thermodynamic costs that make the idealized assumption that any combination of logical gate $g$ and valid heat function $\heat$ for $g$ could actually be constructed.
The known laws of physics neither explicitly construct such gates nor rule out their existence (in the quasistatic limit).
Thus our upper bounds all implicitly require the following ``spherical cow'' like assumption.
\begin{assumption}[Idealized realizability assumption]\label{as:ideal-real}
Let $f: \domain \to \range$ be any function and $\heat: \domain \to \R_{\geq 0} \cup \{\infty\}$ be any valid heat function for computing $f$.
Then we assume that it is possible to fabricate a gate that computes $f$ while releasing $\heat(x)$ expected heat on each input where $\heat(x) \neq \infty$ in the quasistatic (i.e. infinitely slowly) limit.
We assume that whenever $\heat(x) = \infty$ the gate is implemented in a way where it physically cannot receive such an input. In other words, the gate actually computes a restriction of $f$ on inputs where $\heat(x)$ is finite.\footnote{The fact that $\heat(x) = \infty$ already makes this domain restriction implicit in our model.
Any distribution with support on such an input would immediately make $\EF$ infinite.
Thus heat functions where $\heat(x) = \infty$ on any input occurring with positive probability trivially satisfy lower bounds and would automatically make any resulting upper bound trivial.
}
\end{assumption}
This is well motivated.
Theorem $2$ of Turgut~\cite{turgut2009relations} shows how, in an idealized model, it is possible to implement any fixed mapping from $\domain$ to $\range$ where the expected released heat on every input exactly matches any valid $\heat$.
If one wishes to reject this assumption, then they should interpret all of our upper bounds on $\EF$ and $\EP$ as theoretical limits on the best possible lower bounds for these quantities that could be proven in our model.
All of our lower bounds on $\EP$ and $\EF$ hold independently of this axiom.

Excess heat released by a physical process is 
called the entropy production~\cite{Seifert2012StochasticThermodynamics}.
More formally (recall \Cref{defn:lifting-f} of $f_\uparrow$),

\begin{definition}[gate entropy production]
\label{def:ep-gate}
    Let $\circuit$ be a $\varSpace$-wire circuit and $g$ be a gate in $\circuit$ computing a function $f$ with domain $\domain = \alphabet^{\idx}$ for $\idx \subseteq [\varSpace]$ that is assigned heat function $\heat$.
    The \emph{entropy production} $\EP_g$ of $g$ on distribution\footnote{
        Note $\dist$ is not the input distribution to the whole circuit, but the induced distribution on all $W$ wires \emph{just prior} to applying gate $g$.
    }
    $\dist$ over $\alphabet^{\varSpace}$ is given by:
    \[
        \EP_{g}(\dist) = \entropy(f_{\uparrow}(\dist)) - \entropy(\dist) 
        + \E_{x \sim \dist}[\heat(x_{\idx})]
    \]
\end{definition}

Note $\entropy(f_{\uparrow}(\dist))$ is used in \Cref{def:ep-gate} instead of $\entropy(f(\dist))$ since $\dist$ is a distribution on all $\varSpace$ wires.
Note that $\entropy(f_{\uparrow}(\dist)) - \entropy(\dist)$ is the \emph{negation} of the Landauer cost of $f$ on $\dist$ (\Cref{defn:landauer-cost}),
thus is nonpositive by \Cref{prop:entropy-data-processing}.
$\E_{x \sim \dist}[\heat(x_{\idx})]$ is the \emph{entropy flow}/\emph{heat released} (see \Cref{sec:intro} for discussion of this term in the broader field of stochastic thermodynamics) in the context of circuits that exchange heat with the environment through their gates.

It is necessary to define $\EP_{g}$ in the context of the larger circuit, 
since $\entropy(f_{\uparrow}(\dist)) - \entropy(\dist)$ is a global property of the circuit, 
while $\E_{x \sim \dist}[\heat(x_{\idx})]$ is a property of the marginal distribution of $\dist$ on the coordinates in $\idx$.
We can also define a local equivalent, termed $\EPlocal_{g}$, 
to $\EP_{g}$, which gives a lower bound on the 
actual entropy production $\EP_{g}$.
In \cref{sec:modularity} we show how these two quantities are related
(\Cref{lem:mutual-info-positive-ep}).

\begin{definition}[local gate entropy production]
\label{def:local-ep}
    Let $g$ be a gate computing a function $f$ with domain $\domain$ that is assigned a heat function $\heat$.
    The \emph{local entropy production} $\EPlocal_{g}$ of $g$ on input distribution $\dist$ over $\domain$ is given by:
    \begin{displaymath}
        \EPlocal_{g}(\dist) = \entropy(f(\dist)) - \entropy(\dist) + \E_{x \sim \dist}[\heat(x)].
    \end{displaymath}
\end{definition}

We are ultimately interested in characterizing the entropy production of a circuit.
When a circuit receives some distribution $\dist$ over its inputs, 
this could in principal be a marginal distribution 
where the circuit's inputs are
only a subset of the variables in a larger system,
with which the inputs may have mutual information that would call for different heat functions for ``globally optimal'' heat release for each gate.
We give the following definition for the entropy production of a circuit under the assumption that its input distribution $\dist$ is independent from the outside world.
\begin{definition}[circuit entropy production]
\label{defn:circuit-ep}
    Let $\circuit$ be a thermodynamic circuit with gates $g \in G$ that cleanly computes a function $f$ with domain $\domain$.
    Let $\dist$ be a distribution over $\domain$ and $\dist_g$ be the induced marginal distribution over inputs for gate $g$ which has heat function $\heat_g$.
    The \emph{circuit entropy production} of $\circuit$ on distribution $\dist$ is given by
    \[
        \EP_\circuit(\dist) = \entropy(f(\dist)) -\entropy(\dist) + \sum_{g \in G} \E_{x \sim \dist_g}[\heat_g(x)]
    \]
\end{definition}

$\EP_\circuit(\dist)$ as defined above is a lower bound on the actual entropy production involved with running the circuit, as we show in \cref{sec:modularity}
(\cref{lem:mutual-info-positive-ep}).
This bound is tight when $\dist$ is an independent source of randomness uncorrelated with the rest of the universe.
We call this circuit entropy production because it is exactly the entropy production of the circuit as if it were isolated from its environment.
We note that this is the same notion of entropy production for circuits used in prior works
(i.e., \cite{WK20,YCW24,YYW25}).

\begin{remark}
    We have defined $3$ separate types of entropy production in this section.
    Informally, $\EPlocal_g$ is the entropy production due to using a sub-optimal heat function for gate $g$ given its induced marginal distribution over inputs.
    This captures what the entropy production of $g$ would be if its input distribution $\dist$ were sampled independently of the rest of the universe.
    $\EP_g$, which is always at least $\EPlocal_g$, captures $\EPlocal_g$ plus additional costs due to $\dist$ having correlations with other wires in the circuit.
    Finally $\EP_\circuit$ is the total entropy production across all gates in the circuit under the assumption that $\dist$ is sampled independently from the rest of the universe.
    We will be proving lower and upper bounds on $\EP_\circuit$ throughout the body this paper and we use both $\EP_g$ and $\EPlocal_g$ primarily in \cref{sec:rederive-mismatch-cost,sec:modularity} where we  discuss mismatch costs.
\end{remark}

\paragraph{Overwriting circuits.}
It is worth noting that in \cite{YCW24,YYW25} the authors use a different notion of a thermodynamic circuit than we consider here.
In their \emph{overwriting} model, a circuit is defined (as is common in most non-quantum circuit complexity classes) as a DAG where the nodes are labeled with gates and the in-edges define the inputs for each gate.
The authors are interested in repeated runs of their circuit where the old logical configuration of each gate is overwritten one gate at a time during the next run,
so the thermodynamic costs are incurred incrementally as each gate is updated and replaces its old output with the new one.
The overwriting model is an accurate abstraction of the CMOS architecture; 
however, it places some restrictions on what is feasible with circuits from a theoretical perspective.
For example a logically reversible gate in the our circuit model admits the valid heat function $\heat(x)=0$ on all inputs.
In the overwriting model, this gate must release heat when run on independently distributed inputs unless it is computing a constant function.\footnote{In the overwriting model, the validity of the heat function for a deterministic gate is entirely determined by the drop in entropy caused by erasing the old value.}
In fact, we can see that any circuit in the overwriting model can be simulated in our circuit model,
simulating a $k$-input, $1$-output gate 
by using $k$ `input' wires and one `output' wire in each of our gates, setting the other $k-1$ wires to $0$.
Each gate in the simulated overwriting circuit is now a gate in our circuit model that treats the output wires of its parents (or the input) as read-only inputs and only modifies its own output wire.
At the end of the computation, the values stored in all output wires for intermediate gates are erased.\footnote{As long as an overwriting circuit is run on independently distributed inputs, it is safe to assume without loss of generality that it erases the values stored in every gate before computing the new output value.}

\paragraph{Cycles of computation and changing inputs.}
An important subtlety with the thermodynamics of computation is how to think about the state of the device before and after a computation has been completed.
Logical reversibility gives a black box way to, for any efficiently computable function $f$, efficiently perform the mapping $x \to (x,f(x))$ without releasing any heat.
However, even if $f$ is injective (hence ``globally'' logically reversible), this does not imply a black box way to compute the mapping $x \mapsto f(x)$ unless $f^{-1}$ is is also efficiently computable.\footnote{
    For instance, it is well-known folklore that one-way functions cannot be efficiently computed in the forward direction (outputting only $f(x)$ cleanly as in \Cref{defn:cleanly-compute}, without retaining a copy of $x$ anywhere) using logically reversible computation;
    Bennett's construction~\cite[Table 2]{Ben73} for efficiently and reversibly computing the mapping $x \mapsto f(x)$ for bijective $f$ requires both $f$ and $f^{-1}$ to be efficiently computable, ruling out one-way functions.
}
More generally any circuit that produces garbage in addition to its output must erase that garbage before it can be run on a new input.
Failing to do so causes real issues; if one extends \cref{def:valid-heat} to probabilistic gates, then there are valid heat functions where $\heat(x)$ may be negative.
By using such gates, you can define circuits that appear to violate the second law of thermodynamics until you account for the increase in entropy in their mapping of inputs to outputs.

Much of the complexity associated with reuse/cycling can be simplified by thinking of the circuit as a cog in a larger machine that goes through cycles.
Let's consider a circuit computing a function $f: \alphabet^n \to \alphabet^m$ with $\varSpace$ total wires.
At the start of a cycle, all $\varSpace$ wires are initialized to zero.
Then some upstream process writes a new input $x$ into the first $n$ wires of our circuit.
The circuit will perform its computation in accordance with \cref{def:sr-circuit}, ending in a state where the first $m$ wires contain $f(x)$ and all other wires are reset to zero.
The value $f(x)$ can then be swapped with $m$ fresh zeroes from the downstream process that takes $f(x)$ as its input at no cost.
Resetting these values is a task for the downstream process.
Note that $f$ can be defined to keep a copy of its input.
Thus the standard logically reversible mapping $f(x)= (x,h(x))$ is allowed in this model, but then $x$ is explicitly part of the output of $f$ and must be handled as such by downstream processes.

\section{Deriving optimal heat functions}
\label{sec:opt-heat-func}

Now we show how to derive the optimal heat function for a known distribution over inputs.
Recall that for a logically reversible operation, i.e., computing an injective function $f: \domain \to \range$,
an optimal heat function is 0 on all inputs.
In this case, no matter the input distribution $\dist$,
since each input $x$ maps to its own unique output $f(x),$
$\Pr_{f(\dist)}[f(x)] = \Pr_{\dist}[x],$
thus $\log \left(\Pr_{f(\dist)}[f(x)] / \Pr_{\dist}[x] \right) = 0$.
If $f$ is not injective,
it remains the case that any input $x$ mapping to its own unique output $f(x)$ has zero expected heat released on input $x$.
``Colliding'' inputs will release more heat the more inputs map to a single output (more precisely, the greater the total probability $\Pr_{f(\dist)}[f(x)]$ of inputs $x$ sharing output $f(x)$):
still
$\log \left(\Pr_{f(\dist)}[f(x)] / \Pr_{\dist}[x] \right)$ (now strictly positive) is the expected heat released.

\begin{lemma}
\label{prop:ef-min-heat-fn}
    Let $f$ be an arbitrary function with domain $\domain$ and $\dist$ be a distribution over $\domain$.
    Then any valid heat function for a gate computing $f$ that minimizes the expected heat release on input distribution $\dist$ satisfies for all $x \in D$:
    \begin{align*}
        \heat(x) = \begin{cases} \log (\Pr_{f(\dist)}[f(x)] / \Pr_{\dist}[x]) & \Pr_{\dist}[x] > 0 \\ 
        \infty & \Pr_{\dist}[x] = 0, \Pr_{f(\dist)}[f(x)] > 0 \end{cases}
    \end{align*}
    on all $x$ where $f(x) \in \support(f(\dist))$.\footnote{
        Note that equivalently $\log (\Pr\limits_{f(\dist)}[f(x)] / \Pr\limits_{\dist}[x]) = -\log \Pr_\dist[x | \text{ output is } f(x)]$.
    }
    On $x$ where $f(x) \not \in \support(f(\dist))$, $\heat$ can behave arbitrarily so long as $\heat$ still satisfies \cref{lem:turgut-lemma} on these inputs.
\end{lemma}






Note that in the finite case, $Q(x)$ is the log of the one divided by the probability that the input is $x$, conditional on the output being $f(x)$.
In other words, the optimal heat function is unique up to the choice of values for inputs with zero probability under the input distribution $\dist$ and whose outputs also have zero probability under the output distribution $f(\dist)$.
While an earlier version of this proof appeared in \cite{Kor25}, it contained some bugs.
Here we give a corrected proof of optimality and uniqueness up to the behavior of $\heat$ on such inputs.
\begin{proof}
Picking $\heat$ as in the statement of \cref{prop:ef-min-heat-fn} exactly satisfies the best possible bound of:
\begin{equation}\label{eqn:opt-heat-land-eq}
    \landauer_f(\dist) = \E_{x \sim \dist}[\heat(x)].
\end{equation}
To see that this choice is valid, we show that $\heat$ satisfies \cref{lem:turgut-lemma}.
For every $y \in \support(f(\dist))$ we have:
\begin{displaymath}
    \sum_{x \in f^{-1}(y)} 2^{-\heat(x)} = \sum_{x \in f^{-1}(y)} \Pr_{\dist}[x] \Big/ \Pr_{f(\dist)}[y] = 1.
\end{displaymath}
And for every $y \not \in \support(f(\dist))$ the only constraint we place on $\heat$ is that this quantity is at most 1.

Next we show that any heat function that minimizes the heat release on $\dist$ must satisfy \cref{prop:ef-min-heat-fn}.
Let $\heatoptimal$ be some other valid heat function that minimizes the expected heat release on input distribution $\dist$.
Then per \cref{lem:turgut-lemma} for every $y,\sum_{x \in f^{-1}(y)} 2^{- \heatoptimal(x)} \leq 1$.
Moreover if $\sum_{x \in f^{-1}(y)} 2^{- \heatoptimal(x)} < 1$ for some $y \in \support(f(\dist))$ then there must be some $x \in \support(\dist_{|f^{-1}(y)})$ where $\heatoptimal(x)$ could have been smaller while remaining valid, and therefore $\heatoptimal$ could not have been optimal.
Thus we can assume that for every $y \in \support(f(\dist))$:
\begin{displaymath}
    \sum_{x \in f^{-1}(y)} 2^{-\heatoptimal(x)} = 1.
\end{displaymath}
This means that we can define the conditional Gibbs distribution $\Pr_{(\gibbs_{\heatoptimal})_{|f^{-1}(y)}}[x] = 2^{-\heatoptimal(x)}$ for every $y \in \support(f(\dist))$.
Since $\heatoptimal$ is valid we know that for every $y \in \support(f(\dist))$:
\begin{displaymath}
    \E_{x \sim (\gibbs_{\heatoptimal})_{|f^{-1}(y)}}[\heatoptimal(x)] - \entropy((\gibbs_{\heatoptimal})_{|f^{-1}(y)}) \geq 0.
\end{displaymath}
And since $\heatoptimal$ is a valid heat function that minimizes the heat release on the distribution $\dist$ we know that for every $y \in \support(f(\dist))$,
\begin{displaymath}
    \E_{x \sim \dist_{|f^{-1}(y)}}[\heatoptimal(x)] - \entropy(\dist_{|f^{-1}(y)}) = 0,
\end{displaymath}
as \cref{def:valid-heat} requires that this expression be non-negative for every $y \in \support(f(\dist))$ and we need the average of these terms to be 0 to match \cref{eqn:opt-heat-land-eq}.
But by \cref{lem:gibbs-opt}, we know that this implies that for every $y$ where $\dist_{|f^{-1}(y)}$ is defined, $\dist_{|f^{-1}(y)} = (\gibbs_{\heatoptimal})_{|f^{-1}(y)}$ and hence for every $x$ where $f(x) \in \support(f(\dist))$,
\begin{displaymath}
    2^{-\heatoptimal(x)} = \Pr_{\dist}[x] \Big/ \Pr_{f(\dist)}[f(x)] = 2^{- \heat(x)}.
\end{displaymath}
For $x$ where $\Pr_{f(\dist)}[f(x)] = 0$ then both $\heat(x)$ and $\heatoptimal(x)$ can take on arbitrary non-negative values without affecting the optimality of the heat function, so long as it remains valid.
\end{proof}

\section{The complexity of optimizing heat functions in a given circuit}
\label{sec:pp-hard}

Here we show that computing optimal heat functions for gates in a fixed circuit is $\PP$-hard.\footnote{ See~\cite[Chapter
        9]{arora2009computational} for a discussion of the complexity class $\PP$.
    Briefly, it is the class of decision problems/languages $\varLanguage$ for which
    yes-instances $x$, that like $\NP$ problems, have a polynomial-length $k =
        \mathrm{poly}(|x|)$ witness $w \in \{0,1\}^k$, furthermore have at least
    \emph{half} of all $w \in \{0,1\}^k$ as witnesses; see for example the language
    $\varLanguage_\halfAccept$ below. A probabilistic characterization ($\PP$
    actually stands for ``probabilistic polynomial-time'') is that there is a
    randomized polynomial-time algorithm for $\varLanguage$ that answers `yes' on input $x$ with probability at least $1/2$ if and only if $x \in \varLanguage$.} In light of \cref{prop:ef-min-heat-fn}, this
hardness result is hardly surprising: computing an optimal heat function requires
computing probabilities induced by the circuit. Nonetheless, we view this as an
overall strength of the heat function framework.

\subsection{$\PP$-complete language for optimizing heat functions}\label{sec:pp-hard-exact}
We start with the language $\varLanguage_{\EFmin}$ that we will prove
$\PP$-complete.
\begin{definition}
    $\varLanguage_{\EFmin}$ is a set of tuples $(\circuit, g, x, a, b)$ where $\circuit$ is a Boolean circuit, $g$ is a gate in that circuit, $x$ is a potential input to that gate, and $q = a/b$ is a positive rational number.
    We say that $(\circuit, g, x, a, b) \in \varLanguage_{\EFmin}$ if, when $\circuit$ receives the uniform distribution over inputs, there exists some valid heat function $\heat$ for gate $g$ that minimizes the expected heat release on its induced marginal input distribution such that $\heat(x) \geq \log q$.
\end{definition}

We note that our choice to express threshold $q=a/b$ as a rational number whose log corresponds to the optimal expected released heat is for simplicity but is not actually needed.
Per \cref{prop:ef-min-heat-fn} the optimal value of $q$ is always the ratio of two integer valued fractions whose numerators are at most $2^n$ and whose denominators are both exactly $2^n$.
Thus we can safely round up any $q$ obtain a rational $q'$ that is the ratio of two $\leq n+1$ bit integers $a'$ and $b'$ without changing membership in $\varLanguage_{\EFmin}$.

To show that $\varLanguage_{\EFmin}$ is $\PP$-hard, we will reduce from the
following $\PP$-hard problem~\cite{akmal2022majority}.\footnote{A version of this problem where `at least' is replaced with `more than' was proven to be $\PP$-hard in \cite{Gill74}. That paper also uses the convention that $\PP$ algorithms must accept with probability strictly larger than $1/2$ on yes instances. In \cite{akmal2022majority} the authors show how to convert between these definitions of $\PP$.}

\begin{proposition}[\cite{akmal2022majority}]
    \label{prop:sat-half-pp-hard}
    Let $\varLanguage_\halfAccept$ be the subset of $\SAT$ instances that evaluate to $1$ on at least half of all inputs.
    Deciding membership in $\varLanguage_\halfAccept$ is $\PP$-hard.
\end{proposition}

\begin{lemma}
    \label{lem:finding-opt-Q-is-PP-hard}
    $\varLanguage_{\EFmin}$ is $\PP$-hard.
\end{lemma}
\begin{proof}
    We prove this by Karp reduction from $\varLanguage_{\halfAccept}$.
    Given a Boolean circuit $\circuit$, assume its output is the first wire.
    We modify $\circuit$ by appending a single-wire gate $g$ that erases its input to the first wire at the end of the circuit and call our new circuit $\circuit'$.
    Observe that, due to
    \Cref{prop:ef-min-heat-fn}, the optimal heat function for the added gate $g$ releases at least $1$ expected heat on input $0$ precisely when the probability of that input is at most $1/2$.
    Thus, $(\circuit', g, 0, 2, 1) \in \varLanguage_{\EFmin}$ precisely when $\circuit \in \varLanguage_{\halfAccept}$.
\end{proof}

Next we show that this language is in $\PP$.
\begin{lemma}
\label{lem:finding-opt-Q-in-PP}
    $\varLanguage_{\EFmin} \in \PP$.
\end{lemma}
\begin{proof}
    Consider input $(\circuit, g, x, a, b)$.
    Let $\dist$ be the marginal distribution over inputs to gate $g$ when $\circuit$ receives the uniform distribution over inputs, $f$ be the function computed by gate $g$, $y=f(x)$, and $q = a/b$.
    The value $q \leq 1$ asks if the optimal heat function releases at least $0$ expected heat on input $x$.
    Since such a condition is trivially satisfied by the definition of heat functions, we can automatically accept such inputs.
    Now we assume that $a \geq b$.
    By \cref{prop:ef-min-heat-fn}, we know that deciding if $(\circuit, g, x, a, b)\in \varLanguage_{\EFmin}$ is equivalent to deciding if $q =a/b \leq \Pr_{f(\dist)}[y]/\Pr_{\dist}[x]$.
    We will now construct a non deterministic program that accepts with probability at least $1/2$ exactly when $q \leq \Pr_{f(\dist)}[y]/\Pr_{\dist}[x]$.
    We start by observing that this condition is equivalent to
    (multiplying both sides by $b \Pr_{\dist}[x]$, then subtracting $b \Pr_{\dist}[x]$ from both sides)
    \begin{displaymath}
        (a-b)\Pr_{\dist}[x] \leq b
        \left(
        \Pr_{f(\dist)}[y] - \Pr_{\dist}[x]
        \right).
    \end{displaymath}
    Thus by designing a non deterministic program whose success probability is at least $1/2$ when the above inequality is satisfied, we have a $\PP$ algorithm for $\varLanguage_{\EFmin}$.
    Importantly, note that $\Pr_{f(\dist)}[y] - \Pr_{\dist}[x]$ is exactly the probability that $\dist$ selects some input $x' \not = x$ where $f(x') = y$.
    We start by picking an input for $\circuit$ uniformly at random and evaluating it on $\circuit$.
    Let $x'$ be the observed input to gate $g$ and $y' = f(x')$ be the observed output.
    We then sample a Bernoulli random variable $v$ that has outcome $1$ with probability $b/a$.\footnote{While this is not possible in general under standard definitions of $\PP$, we can get the desired behavior by flipping $\ceil{\log(a+b)} + 1$ coins. Treat all but the last of these coins as a number $z$ in binary. If $z < b$ we set $v = 1$. If $b\leq z <a$ we set $v=0$. If $a\leq z$ then we use the final coin to immediately accept or reject without changing the correctness of our algorithm. Post-conditioned on not immediately halting, this protocol simulates the desired Bernoulli random variable.}
    If we have $v=1, x'\neq x,$ and $y' = y$ then we accept.
    If we have $v=0$ and $x' = x$ then we reject.
    On all other patterns we accept or reject with probability $1/2$.
    We note that this protocol accepts with probability exactly $1/2$ precisely when $q=\Pr_{f(\dist)}[y]/\Pr_{\dist}[x]$
    or $\Pr_{f(\dist)}[y] = 0$.
    In both cases we have that our input is in $\varLanguage_\EFmin$ and thus should accept; when $q=\Pr_{f(\dist)}[y]/\Pr_{\dist}[x]$ we have that $q$ is precisely the optimal value and when $\Pr_{f(\dist)}[y] = 0$ every choice of $q$ can be optimal.
    Given our definition of $\PP$, this is the intended behavior.
\end{proof}

\subsection{Approximation}\label{sec:approx-heat-funcs}
We first make a simple argument that if circuits are permitted to have wires where only a 0 bit can ever appear 
(i.e., such as the output of an unsatisfiable formula), 
then it is $\NP$-hard even to approximate an optimal heat function.
Fix a circuit $\circuit$ and let $g$ be an erasure gate acting on the first
wire at the end of that circuit. Let $\dist$ be the marginal induced input
distribution for $g$ when $\circuit$ receives a uniform distribution over its
inputs and $\heatoptimal$ be a valid heat function that minimizes the expected heat release
on distribution $\dist$. If $\Pr_{\dist}[0] = 1$ then we can safely set
$\heat(0) = 0, \heat(1) = \infty$ and observe that $\E_{x
        \sim \dist}[\heat(x)] = 0$. 
However, if there is any input to
$\circuit$ that causes the first wire to evaluate to $1$ before applying $g$,
then this choice of $\heat$ will cause $\E_{x \sim
        \dist}[\heat(x)] = \infty$. Thus, for any $c \geq 1$, this choice of
$\heat$ is the unique choice that satisfies $\E_{x\sim
        \dist}[\heat(x)] \leq c \E_{x\sim \dist}[\heatoptimal(x)]$ when
$\Pr_{\dist}[0] = 1$ and $\heat$ can never satisfy this condition
when $\Pr_{\dist}[1] > 0$. Therefore by fixing $\circuit$ to be the $\circuit'$
in \Cref{lem:finding-opt-Q-is-PP-hard}, finding any multiplicative
approximation to the optimal heat function for the added erasure gate is at
least as hard as solving $\SAT$.

However, this is a somewhat degenerate case,
in the sense that \emph{any} optimization problem permitting an optimal value of 0 has the property that no multiplicative approximation is good unless it is exact. 
The more interesting case is when a wire might have both bits 0 or 1 ($\support(\dist) = \{0,1\}$),
so that we have a positive optimal expected heat $\E_{x \sim \dist}[\heatoptimal(x)] > 0$.
Using a bit of machinery, we can prove that $\NP$-hardness of multiplicative
approximation holds even in this case.

\begin{remark}[Pointwise versus expected approximation]
\label{rem:approximation-pointwise-vs-expectation}
The motivation for studying the thermodynamics of a circuit is to reduce the amount of work (i.e.\ useful energy) that must be irreversibly spent each time circuit is run.
Ultimately we care about finding heat functions that minimize the circuit entropy production, or equivalently (since the mapping of inputs to outputs is fixed) the net expected released heat by each gate.
Thus, for our approximation results, the natural computational objective is to produce \emph{any} heat function where the expected released heat is within an additive (or multiplicative) factor of optimal, rather than to approximate the individual entries of the optimal heat function.
\end{remark}

We first require a technical lemma about converting a formula $\phi$ to another $\psi$ such that, 
depending on whether $\phi$ is satisfiable,
$\psi$ has either 1 or a much larger number of satisfying assignments.
This follows by adding $t$ dummy variables to a construction by Blass and Gurevich~\cite{DBLP:journals/iandc/BlassG82} that converts any $\SAT$ formula $\phi$ on $n$ variables into another formula $\psi$ on $n+1$ variables which contains exactly one additional satisfying assignment.

\begin{lemma}
\label{lem:sat-promise-conv}
    Let $\phi(x_1, \ldots, x_n)$ be a $\SAT$ instance on $n$ variables with $Z(\phi)$ satisfying assignments.
    For every $t \geq 0$, there exists another $\SAT$ instance $\psi(x_1, \ldots, x_n, y_1, \ldots, y_t, z)$ on $n+t+1$ variables that has $Z(\psi) = 1 + 2^t Z(\phi)$ satisfying assignments.
\end{lemma}
\begin{proof}
    We fix $\psi(x_1, \ldots, x_n, y_1, \ldots, y_t, z) = \left(\left(\bigwedge_{i \in [n]} x_i\right) \wedge \left(\bigwedge_{i \in [t]} y_i\right) \wedge z \right) \vee (\phi(x_1, \ldots, x_n) \wedge \neg z)$.
    When $z=1$, $\psi$ is satisfied exactly when all $x_i$ and $y_i$ are assigned $1$.
    When $z=0$, each assignment of $x_1, \ldots, x_n$ that satisfy $\phi$ now results in $2^t$ assignments of $y_1, \ldots, y_t$ that all satisfy $\psi$.
    Thus the total number of satisfying assignments is now $1+ 2^t Z(\phi)$.
\end{proof}

Next, we prove $\NP$-hardness of multiplicative approximation even for gates whose inputs have full support.

\begin{lemma}
\label{lem:mult-approx-hard}
    Fix any $c > 1$.
    Let $\circuit$ be a circuit receiving the uniform distribution over inputs and $g$ be a gate in $\circuit$.
    Let $\dist$ be the induced marginal distribution over inputs to $g$.
    Let $\heatoptimal$ be a valid heat function for $g$ that minimizes $\E_{x \sim \dist}[\heatoptimal(x)]$.
    Finding a valid heat function $\heat$ for $g$ such that $\E_{x \sim \dist}[\heat(x)] \leq c \E_{x \sim \dist}[\heatoptimal(x)]$ that is specified by a pairs $a_x,b_x$ of $n+1$ bit integers such that $\heat(x) = \log(a_x/b_x)$ is impossible in randomized polynomial time even with the promise that $\E_{x \sim \dist}[\heatoptimal(x)] > 0$ unless 
    $\NP \subseteq \BPP$.
\end{lemma}
\begin{proof}
    We prove this via an oracle reduction from $\SAT$.
    Let $\phi$ be a $\SAT$ instance on $n$ variables and $\psi$ be the resulting instance on $n+t+1$ variables after applying \cref{lem:sat-promise-conv} for $t=\alpha c n$ with constant $\alpha>1$ that we will choose later.
    By construction, if $\phi$ had $Z(\phi)$ satisfying assignments, then $\psi$ must have $Z(\psi) = 1+2^t Z(\phi)$ satisfying assignments.
    This implies that $1 \leq Z(\psi) \leq 1 + 2^{n+t} < 2^{n+t+1}$.
    We convert $\psi$ into a Boolean circuit $\circuit$ whose output is the first wire and then append an erasure gate to that wire to obtain circuit $\circuit'$.
    Let $\heatoptimal(x)$ be some heat function that minimizes the expected heat release on the induced input distribution $\dist$ to the added erasure gate.
    Since $\psi$ evaluates to both true and false on some input,
    we know that $\E_{x \sim \dist}[\heatoptimal(x)] > 0$.
    If $\phi$ was unsatisfiable, then $\psi$ has exactly $1$ satisfying assignment and $\Pr_{\dist}[1] = 1/2^{n+t+1}$.
    If $\phi$ had some satisfying assignment then $\Pr_{\dist}[1] \geq (1 + 2^t)/2^{n+t+1}$.
    Let $\heat$ be a valid heat function such that $\E_{x \sim \dist}[\heat(x)] \leq c \E_{x \sim \dist}[\heatoptimal(x)]$.
    By \cref{prop:ef-min-heat-fn} we know that $\heatoptimal(x) = -\log(\Pr_{\dist}[x])$ and thus if $\phi$ was unsatisfiable we have:
    \begin{displaymath}
        (1-1/2^{n+t+1}) \heat(0) \leq \E_{x\sim \dist}[\heat(x)] \leq c \binaryEntropy(1/2^{n+t+1}) \leq c (n+t+1+\log e)/2^{n+t+1},
    \end{displaymath}
    which implies $\heat(0) \leq \frac{c(n+t+1+\log e)}{2^{n+t+1}-1} \leq \frac{c(n+t+1+\log e)}{2^{n+t}}$.
    Since $\heat$ must be valid \cref{lem:turgut-lemma} gives us:
    \begin{align}
        \heat(1) & \geq - \log(1-2^{-\heat(0)})\nonumber\\
                 & \geq - \log(\heat(0) \ln 2)\nonumber\\
                 & \geq \log \left(\frac{2^{n+t}}{c(n+t+1+\log e) \ln 2}\right) \nonumber\\
                 &= n+ t - \log\left(c(n+t+1 + \log e)\ln(2)\right)\nonumber\\
                 &= \alpha cn+n - \log\left(c((\alpha c+1)n+1 + \log e)\ln(2)\right)\nonumber\\
                 &> \alpha cn - 2\log( 5\alpha cn) \qquad \textrm{since $\alpha,c,n \geq 1$} \label{eq:heat-one-upper}
    \end{align}
    But if $\phi$ were satisfiable then we must have that:
    \begin{displaymath}
        \Pr_{\dist}[1] \heat(1) \leq  \E_{x \sim \dist}[\heat(x)] \leq c \binaryEntropy(\Pr_{\dist}[1]) \leq c \Pr_{\dist}[1] \log(e/\Pr_{\dist}[1])
    \end{displaymath}
    Which implies:
    \begin{align*}
        \heat(1) & \leq c \log(e/\Pr_{\dist}[1])\\
        &= c(\log(e) - \log(\Pr_{\dist}[1]))\\
        &\leq c\left(\log(e) - \log\left(\frac{1+2^t}{2^{n+t+1}}\right)\right)\\
        &=c\left( n + t + 1 + \log e - \log(1+2^t)\right)\\
        &\leq c(n+1+\log e)\\
        &< 4cn.
    \end{align*}
    As long as $\alpha c n \geq 18$, \cref{eq:heat-one-upper} is at least $\alpha c n / 4$.
    For simplicity, we will fix $\alpha = 18$ to get that $\heat(1) > 9 c n /2$ when $\phi$ is unsatisfiable.
    But, if $\phi$ has a single satisfying assignment then $\heat(1) < 4 cn$.
    Thus finding a valid heat function $\heat$ for this erasure gate such that $\E_{x \sim \dist}[\heat(x)] \leq c \E_{x \sim \dist}[\heatoptimal(x)]$ tells us whether $\phi$ has a satisfying assignment.
\end{proof}

The $\NP$-hardness of multiplicative approximation is unfortunate, as being able to estimate the optimal heat functions for each of $T$ gates in a circuit to multiplicative error $c$ immediately gives us a thermodynamic circuit that is within a multiplicative factor $c$ of optimal.
Like all $\NP$-hard problems,
it also rules out a polynomial-time \emph{randomized} algorithm unless $\NP \subseteq \BPP$.

There \emph{is} a polynomial-time randomized algorithm for the problem of finding optimal heat functions with \emph{additive} error on the expected released heat.
Finding heat functions for each of the $T$ gates with additive error at most $\varepsilon$ will give us a thermodynamic circuit where the additive error (i.e. expected excess heat dissipation) across all gates is at most $\varepsilon T$.
By setting $\varepsilon = \varepsilon'/T$ we can obtain an additive error $\varepsilon'$ for the full circuit that doesn't scale with $T$.
The basic idea is simple: randomly sample inputs to the circuit and use them to empirically estimate the true input distribution to the gate with add-one (or Laplace) smoothing.\footnote{The smoothing factor gives us a straightforward way to ensure that our produced heat function $\heatestimate$ never sets $\heatestimate(x) = \infty$ on any input.}

\begin{lemma}
\label{lem:additive-approx-easy}
    Fix any $\varepsilon>0$.
    Given a circuit $\circuit$ where all gates act on a constant number of bits and one such gate $g$ in $\circuit$, let $\dist$ be the induced marginal distribution over the inputs to gate $g$ when $\circuit$ receives the uniform distribution over its inputs.
    Let $\heatoptimal$ be a valid heat function for $g$ that minimizes $\E_{x \sim \dist}[\heatoptimal(x)]$.
    There is a $\poly(n, 1/\varepsilon, 1/\delta)$ time randomized algorithm that produces a valid heat function $\heatestimate$ such that $\E_{x \sim \dist}[\heatestimate(x)] \leq \E_{x \sim \dist}[\heatoptimal(x)] + \varepsilon$ with probability at least $1-\delta$.
\end{lemma}
\begin{proof}
    Let $f$ be the function computed by gate $g$, $\domain$ be the domain of $f$, and $d= \abs{\domain}$.
    We pick $t$ inputs to $\circuit$ uniformly at random, evaluate $\circuit$ on these inputs, and count the number of times we see each input to $g$.
    Let $\counts: \domain \to \Z_{\geq 0}$ be the obtained counts for each input $x \in \domain$ to $g$.
    From these counts, we define some empirical distribution $\widehat{\dist}$ and let $\heatestimate(x) = \log(\Pr_{f(\widehat{\dist})}[f(x)] / \Pr_{\widehat{\dist}}[x])$ be the optimal heat function (per \cref{prop:ef-min-heat-fn}) if $\widehat{\dist}$ were the actual input distribution to $g$.
    We use $\heatestimate_\counts, \widehat{\dist_\counts}$ to denote the choices of $\widehat{\dist}$ and $\heatestimate$ respectively resulting from observing counts $\counts$.
    We use $\countDist$ to denote the distribution over the observed counts.
    This gives us that:
    \begin{align*}
    & \mathrel{\phantom{=}}
        \E_{\counts \sim \countDist}[\E_{x \sim \dist}[\heatestimate_\counts(x) - \heatoptimal(x)]]
    \\ & =
        \E_{\counts \sim \countDist}\left[ \E_{x \sim \dist} \left[
            \log\left( \frac {\Pr_{f(\widehat{\dist_\counts})}[f(x)]} {\Pr_{\widehat{\dist_\counts}}[x]} \right)
            - \log\left( \frac {\Pr_{f(\dist)}[f(x)]} {\Pr_{\dist}[x]} \right)
            \right]\right]
        \quad \hspace{-10pt}\text{definition of $\heatestimate_\counts$; \cref{prop:ef-min-heat-fn} for $\heatoptimal$}
    \\ & =
        \E_{\counts \sim \countDist}\left[ \E_{x \sim \dist} \left[
            \log\left( \frac {\Pr_{\dist}[x]} {\Pr_{\widehat{\dist_\counts}}[x]} \right)
            - \log\left( \frac {\Pr_{f(\dist)}[f(x)]} {\Pr_{f(\widehat{\dist_\counts})}[f(x)]} \right)
            \right]\right]
        \qquad \text{rearrange terms}
    \\ & =
        \E_{\counts \sim \countDist}\left[ \E_{x \sim \dist} \left[
            \log\left( \frac {\Pr_{\dist}[x]} {\Pr_{\widehat{\dist_\counts}}[x]} \right)
            \right]
            - \E_{y \sim f(\dist)} \left[ \log\left( \frac {\Pr_{f(\dist)}[y]} {\Pr_{f(\widehat{\dist_\counts})}[y]} \right)
            \right]\right]
        \quad \hspace{-5pt} \text{since $x \sim \dist \implies f(x) \sim f(\dist)$}
    \\ & =
        \E_{\counts \sim \countDist}\left[\DKL(\dist || \widehat{\dist_\counts})
            - \DKL(f(\dist) || f(\widehat{\dist_\counts}))
            \right]
        \qquad \text{definition of $\DKL$}
    \end{align*}
    This connection between excess heat dissipation and the drop in KL-divergence associated with computing $f$ on a device optimized for the wrong distribution over inputs is well known in the stochastic thermodynamics of computation literature.\footnote{A universal lower bound on $\EP$ of this form was first identified by Kolchinsky and Wolpert~\cite{KW17} and was later called a \emph{mismatch cost}.
    The same authors expanded on this concept in \cite{WK20} and we include and re-prove their generalized version as \cref{prop:mismatch-cost}.}
    Since KL-divergence is non-negative, $E_{\counts \sim \countDist}\left[\DKL(\dist || \widehat{\dist_\counts}) - \DKL(f(\dist) || f(\widehat{\dist_\counts}))\right] \leq E_{\counts \sim \countDist}\left[\DKL(\dist || \widehat{\dist_\counts})\right]$ and the expected excess heat dissipation from using $\widehat{\dist}$ to pick our choice of $\heatestimate$ is at most the expected KL-divergence between $\dist$ and $\widehat{\dist}$.
    It is known that the add-one smoothing estimator is well suited to this task.
    For count vector $\counts \sim \countDist$, this estimator sets $\Pr_{\widehat{\dist_{\counts}}}[x] = (\counts(x)+1)/(t+d)$.
    Braess, Forster, Sauer and Simon~\cite{DBLP:conf/alt/BraessFSS02} showed that for any true distribution $\dist$, this estimator after $t$ samples satisfies
    \begin{displaymath}
        \E_{\counts \sim \countDist}\left[ \DKL(\dist || \widehat{\dist_\counts})\right] \leq \log\left(1+\frac{d-1}{t+1}\right),
    \end{displaymath}
    which is at most $(d-1) \log(e) /(t+1)$.

    Thus by using this choice of $\widehat{\dist}$ and $\heatestimate$ we can set the number of samples $t = \ceil{\log(e) d / (\varepsilon \delta)}$ to get that $\E_{\counts \sim \countDist}[\E_{x \sim \dist}[\heatestimate_\counts(x) - \heatoptimal(x)]]  \leq \varepsilon \delta$.
    By Markov's inequality, this means that the probability over $\counts \sim \countDist$ that $\E_{x \sim \dist}[\heatestimate_\counts(x) - \heatoptimal(x)] \geq \varepsilon$ is at most $\delta$.
    Since $g$ operates on a constant number of bits, $d$ is also a constant and this can be done in a polynomial number of steps.
\end{proof}

\section{Bounds on thermodynamic costs of fixed circuits}
\label{sec:heat-functions-fixed-circuits}


\Cref{sec:pp-hard} was about efficient algorithms for finding optimal heat functions in a given circuit in order to minimize heat released on the uniform input distribution.
\Cref{lem:additive-approx-easy} aside, the hardness results imply efficient algorithms for determining optimal thermodynamic parameters are unlikely to exist in general.
Nevertheless we can analyze fixed families of circuits to derive bounds in these fixed cases.
We first warm up by analyzing two simple implementations of an $n$-ary AND function:
a balanced tree (depth $\log n$),
and a linear chain of (depth $n$),  of 2-input AND gates
(\Cref{sec:heat-functions-AND-entropy-depth-tradeoff}).
We then study the Batcher even-odd sorting networks,
analyzing first the Boolean case (\Cref{sec:batcher-bits})
then the case where each wire carries a whole integer
(\Cref{sec:batcher-integers}).

For all of these results, the upper and lower bounds we show on entropy production are for $\EP_\circuit$ (i.e. they assume that the input distribution is sampled independently from the outside world).
This assumption is perhaps stronger than it looks.
Since $\EP$ measures the total increase in entropy of the universe,
this is essentially saying the inputs,
which presumably came from somewhere in the universe,
are distributed independently of the state of the rest of the universe.
Nevertheless, upper bounds on $\EP_\circuit$ which are smaller than $\EF$ are still worth studying.
When $\EP_\circuit \ll \EF$ we know that most of the released heat is offsetting the drop in entropy of the marginal input distribution.
When $\EP_\circuit = 0$ we know that no circuit could compute the same mapping of inputs to outputs while releasing less expected heat.

\subsection{Warm-up: entropy-depth tradeoff for computing AND}
\label{sec:heat-functions-AND-entropy-depth-tradeoff}

We note that the circuits we consider in this example are compatible with the overwriting model.
Thus the bounds we show here apply equally well to that model of computation under the assumption that inputs to the circuit across runs are independently distributed.

Assume that the input is drawn from the uniform distribution $U_n$ over $\{0,1\}^n$ and say that you want to cleanly compute the function $\text{AND}(x_1, \ldots, x_n) = (x_1 \wedge \ldots \wedge x_n)$ where $n$ is a power of two using only and gates of the form $\bigwedge(x_1,x_2,z) = (x_1, x_2, x_1 \wedge x_2)$ and erasure gates $E(x) = 0$.
Per \cref{prop:ef-min-heat-fn} we know that the optimal heat functions for each of these gates depend on their distributions over inputs.
Since $\bigwedge(x_1,x_2,z)$ leaves inputs $x_1,x_2$ untouched and it overwrites $z$ with a deterministic function only of $x_1,x_2$, when inputs $x_1, x_2$ are independently distributed from $z$, the optimal heat function for this gate only depends on the marginal probability distribution over $z$.
We can see this as follows for such a product input distribution $\dist = \dist_{x_1,x_2}  \otimes \dist_z$:
\begin{displaymath}
    \heat_\wedge(x_1,x_2,z) = \log \left(\frac{\Pr_{\bigwedge(\dist)}[(x_1,x_2, x_1 \wedge x_2)]}{\Pr_{\dist}[(x_1,x_2,z)]}\right) = \log \left(\frac{\Pr_{\dist_{x_1,x_2}}[(x_1,x_2)]}{\Pr_{\dist_{x_1,x_2}}[(x_1,x_2)]\Pr_{\dist_{z}}[z]}\right) = \log \left(\frac{1}{\Pr_{\dist_z}[z]}\right)
\end{displaymath}
Similarly, per \cref{prop:ef-min-heat-fn}, the optimal heat function for $E$ on input distribution $\dist$ is given by:
\begin{displaymath}
    \heat_E(x) = \log\left(\frac{1}{\Pr_{\dist}[x]}\right)
\end{displaymath}
Importantly we also have that the expected heat release of these gates is given by $\E_{x \sim \dist}[\heat_E(x)] = \binaryEntropy(\Pr_{\dist}[0]) = \binaryEntropy(\Pr_{\dist}[1])$.

With these optimal heat functions in mind, we will now analyze the local entropy production of two different fixed Boolean circuits computing $\text{AND}(x_1, \ldots, x_n) = (x_1 \wedge \ldots \wedge x_n)$ using these gates.

The first we consider is the binary tree composed of $\bigwedge(x_1,x_2,z)$ gates on pairs of inputs $x_1,x_2$ and an initially zeroed wire $z$.
After executing each of these gates, we apply $E(x)$ gates on the wires containing the values of $x_1$ and $x_2$.
Since each $\bigwedge(x_1,x_2,z)$ gate in this circuit is run with $z=0$, they can all be implemented with the heat function $\heat_\wedge(x_1,x_2,0) = 0, \heat_\wedge(x_1,x_2,1) = \infty$ and have an expected heat release of zero.
Thus it is only necessary to keep track of the expected heat release of the erasure gates.
The erasure gates deployed after applying the $t$'th layer of $\bigwedge(x_1,x_2,z)$ gates are each erasing bits that have the value $1$ with probability $2^{-2^t}$.
Thus the expected heat release of each of these gates is $\binaryEntropy(2^{-2^t})$.
Since there are $n/2^{t}$ such erasure gates in layer $t$, the expected heat released in this layer is $(n/2^{t})\binaryEntropy(2^{-2^t})$.
Thus the total expected heat release is
\begin{displaymath}
    \sum_{t=0}^{\log (n)-1} \frac{n}{2^t}\binaryEntropy(2^{-2^t}).
\end{displaymath}
Meanwhile the circuit maps a state with $n$ bits of entropy to only $\binaryEntropy(2^{-n})$ bits of entropy.
By \cref{def:local-ep}, the circuit entropy production is therefore:
\begin{displaymath}
    \EP_\circuit = \binaryEntropy(2^{-n}) - n + \sum_{t=0}^{\log (n)-1} \frac{n}{2^t}\binaryEntropy(2^{-2^t}) = \sum_{t=1}^{\log(n)}\frac{n}{2^t}\binaryEntropy(2^{-2^t}) \geq \frac{n}{2}\binaryEntropy(1/4).
\end{displaymath}
Which is $\Omega(n)$.

Next we consider the serial circuit where the $i$'th $\bigwedge(x_1,x_2,z)$ gate has the AND of the first $i$ inputs as its first argument, $x_{i+1}$ as its second argument, and 0 as its third argument.
The circuit then applies erasure gates on the wires containing inputs $x_1$ and $x_2$.
This circuit is illustrated for $n=4$ in \cref{fig:and-example-circuit}.
Again these $\bigwedge(x_1,x_2,z)$ gates can be implemented with 0 expected heat release under the promise that $z=0$, so the only source of heat release are the erasure gates.
There are a total of $n$ erasure gates in this circuit that erase the individual inputs and $n-2$ erasure gates that erase the partial products $x_1 \wedge x_2$ through $x_1 \wedge \ldots \wedge x_{n-1}$.
Since the inputs are unbiased bits,
the expected heat release of these gates is $\binaryEntropy(1/2) = 1$ per gate and $n$ across all gates.
For the partial product erasures, the $i$'th such gate erases the value of $x_1 \wedge \ldots \wedge x_{i+1}$.
Again we know that the $x_i$ are independently and uniformly distributed bits and therefore the $i$'th partial erasure gate is erasing a bit that has the value $1$ with probability $2^{-i-1}$.
The expected heat release for this gate is therefore $\binaryEntropy(2^{-i-1})$.
Thus our total expected heat release for the circuit is $n + \sum_{i = 1}^{n-2} \binaryEntropy(2^{-i-1})$.
Since we have the same overall mapping of inputs to outputs as the binary tree case, our change in entropy is $\binaryEntropy(2^{-n})-n$.
This gives us a circuit entropy production of
\begin{displaymath}
    \EP_\circuit = (\binaryEntropy(2^{-n}) - n) + n + \sum_{i = 1}^{n-2} \binaryEntropy(2^{-i-1}) = \sum_{i = 2}^{n} \binaryEntropy(2^{-i}) \leq \sum_{i=2}^n 2^{-i} (i +\log(e)) \leq (3 + \log(e))/2
\end{displaymath}
which is $O(1)$.

Another way one could compute the circuit entropy production of these circuits would be to use \cref{lem:mutual-info-positive-ep} and sum up the drops of mutual information associated with executing each gate.
Since the gates we are considering here are each locally optimal, this would give us the same values of $\EP_\circuit$ that we computed here.
Doing so is the primary approach used in prior works (i.e. \cite{YYW25}) to produce lower bounds on the local entropy production of circuits.
By using heat functions, we can exactly compute the circuit entropy production by only analyzing the local behavior of the individual gates.
This lets us ignore more global properties of the circuit like the drops in mutual information, making our approach better suited to analyzing complex circuits where these mutual information terms are hard to quantify.

\subsection{Batcher's odd-even mergesort network}
\label{sec:batcher-bits}
We now give bounds on the circuit entropy production of Batcher's odd-even mergesort network~\cite{batcher1968}.
We note that the Landauer cost of a sorting network acting on bits chosen uniformly at random is $\Omega(n)$,
since for an input with exactly $k \approx n/2$ $0$'s,  it collapses $\binom{n}{k} \geq (n/k)^k$ different inputs to the single output $0^k 1^{n-k}$.
As a result, by \Cref{prop:second-law} applied to the whole circuit, any sorting network necessarily
releases heat at least $\Omega(n)$.
In this section we show that the Batcher network gates can be assigned heat functions meeting this bound
on heat released:
$O(n)$
(\Cref{thm:batcher-heat-released-upper-bound}).
It remains an open question whether entropy production is necessarily $\Omega(n)$ in any sorting network
but
we are able to show that the Batcher's odd-even mergesort network on binary inputs does have cost $\Omega(n)$ (\Cref{thm:batcher-ep-omega-n}). We then turn to the integer case, where the Landauer cost (i.e. drop in entropy) of a sorting network acting on a uniformly random permutation of $\{0,1,\ldots,n-1\}$
is $\Theta(n \log n)$. In this case, however, we show in \cref{thm:batcher-entropy-integer-high} that, when the input is 
the distribution resulting from picking $n$
elements of an ordered set without replacement uniformly at random, then
Batcher's network releases heat $\Theta(n \log^2 n)$ with optimal heat functions, which is directly proportional to the size of the network. In \cref{thm:pre-sort-helps-upper-bound} we leverage a sorting network of Leighton and Plaxton 
~\cite{LP98}, to obtain a sorting network with size $\Theta(n \log^2 n)$ but only $\Theta(n \log n)$ heat dissipation.
In contrast, 
by generalizing to a sorting circuit with more than $n$ wires,
it \emph{is} possible to achieve zero circuit entropy production; see \Cref{sec:clean-compaction-zero-ep}.

The Batcher network is illustrated in \Cref{fig:batcher} with $n=16$ inputs.
Each vertical line represents a comparator, which sorts two bits, putting the
smaller of the two on the ``upper" line, i.e., the line with the smaller index,
and the larger on the ``lower" line. If $a \leq b$ then we say that the
comparator gate $\compare(a,b) = \compare(b,a) = (a,b)$.
For simplicity in this section, we will assume that the number of sorting
inputs $n$ is a power of two.

More formally, we will index wires $w_i$ for $i \in \{0, \ldots n-1\}$ for $n$ being a power of two.
We use $w(\ell)$ to denote the value stored in wire $w$ after layer $\ell$.
By the definition of Batcher's, we know that $w_0(0) \leq \ldots \leq w_{n/2-1}(0)$ and $w_{n/2}(0) \leq \ldots \leq w_{n-1}(0)$ as these two halves of the input were sorted during the previous recursive step.
We denote a comparator gate acting on indices $i$ and $j$ with $i < j$ as $(i,j)$.
Let $k_\ell = \log(n) - \ell$.
For layer $1$ our set of gates is $$\gateSet_1 = \{(0, n/2), \ldots, (n/2-1, n-1)\}.$$
The set of gates at layer $\ell > 1$ are
\begin{equation}\label{eqn:gate-locs}
    \gateSet_\ell = \{(i, i+2^{k_\ell}) \; | \; 2^{k_\ell} \leq i < n-2^{k_\ell} \text{ and } \floor{i/2^{k_\ell}} \equiv 1 \bmod 2\}.
\end{equation}

To account for the total number of comparators needed, consider first the
number $M(n)$ in the non-recursive part (layers 1 to $\log n$). We have that
$M(2) = 1$ and
\[
    M(n) = 2M(n/2) + n/2-1,
\]
where the two $M(n/2)$ terms account for the separate sorting of odd-numbered
bits and even-numbered bits, and the $n/2-1$ accounts for the last layer.
Therefore, $M(n) = \Theta(n \log n).$ Now let $C(n)$ be the total number of
comparator gates needed to sort $n = 2^k$ bits. We have that $C(2) = 1$ and for
$n = 2^k \ge 4$,
\[
    C(n) = 2C(n/2) + M(n) = 2C(n/2) + \Theta(n \log n).
\]
Therefore $C(n) = \Theta(n \log^2 n)$.

\begin{figure}[tb]
    \centering
    \includegraphics[width=0.8\linewidth]{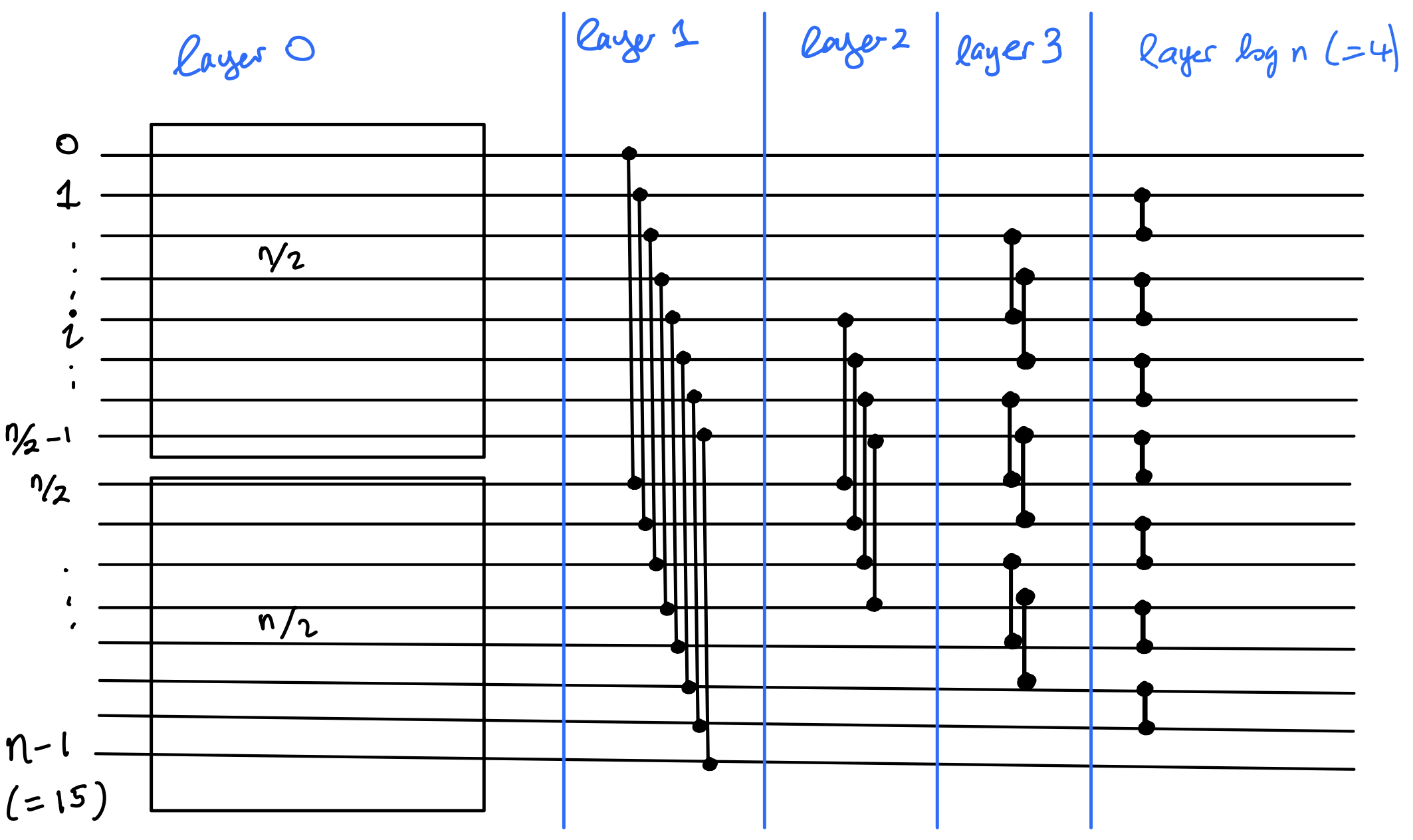}
    \caption{Recursive case of Batcher's odd-even mergesort network for 16 inputs. The base case is $n=4$ is shown in \Cref{fig:batcher-n-4}.}
    \label{fig:batcher}
\end{figure}

\subsubsection{Binary inputs}
We first consider the case where the input to be sorted is a length $n$
bitstring selected uniformly at random. It is known that the correctness of a sorting network
can be determined by only checking its behavior on such inputs.
\begin{proposition}[\cite{Knuth98a}]\label{prop:zero-one-principle}
    A sorting network on $n$ inputs from any ordered set is correct iff it correctly sorts all inputs in $\{0,1\}^n$.
\end{proposition}

The binary input restriction allows for comparator gates to be implemented in a
more thermodynamically efficient way. This is because the inputs $(0,0)$ and
$(1,1)$ are the unique pre-images of outputs $(0,0)$ and $(1,1)$ respectively,
meaning that we can set $\heat(0,0) = \heat(1,1) =0$ and we only need to
optimize the values of $\heat(0,1)$ versus $\heat(1,0)$ which both map to
output $(0,1)$.

\begin{theorem}\label{thm:batcher-heat-released-upper-bound}
    When given the uniform distribution over input $x$ in $\{0,1\}^n$, there is an assignment of heat function to every gate that results in the circuit releasing only $O(n)$ expected heat across all gates.
    Moreover, the construction need only use $2 \log(n) + 1$ distinct heat functions.
\end{theorem}
\begin{proof}
    We start by bounding the expected heat released in each layer of the non-recursive merge step.
    We will use three different types of heat functions in this construction.
    $\heat^\unif$ is defined such that $\heat^\unif(x_i,x_j) = 0$ when $x_i=x_j$ and $\heat^\unif(x_i,x_j) = 1$ when $x_i \neq x_j$.
    This is an unbiased heat function that we will use either when inputs $(0,0)$ and $(1,1)$ are overwhelmingly common or when the number of gates in a group is sufficiently small.
    This heat function has the advantage that neither input $(0,1)$ or input $(1,0)$ releases too much heat in expectation.
    We also define $\heat^{(0,1)}_\varepsilon$ parameterized by $\varepsilon \leq 1/2$ such that $\heat^{(0,1)}_\varepsilon(0,0) = \heat^{(0,1)}_\varepsilon(1,1) = 0$, $\heat^{(0,1)}_\varepsilon(0,1) = \log(1/(1-\varepsilon))$, and $\heat^{(0,1)}_\varepsilon(1,0) = \log(1/\varepsilon)$.
    Intuitively, these functions are used to reduce the expected heat released by gates where the first input is heavily biased towards input $0$ while the second input is heavily biased towards input $1$.
    In particular, if a gate receives a marginal induced input distribution where the probability of input $(1,0)$ is at most $\varepsilon \leq 1/2$, then the expected heat released by this gate is at most $(1-\varepsilon) \log(1/(1-\varepsilon)) + \varepsilon \log(1/\varepsilon) \leq 2 \varepsilon \log(1/\varepsilon)$.
    Likewise we define $\heat^{(1,0)}_\varepsilon$ such that $\heat^{(1,0)}_\varepsilon(a,b) = \heat^{(0,1)}_{\varepsilon}(b,a)$ for $\varepsilon \leq 1/2$.
    These gates are used when the first input is heavily biased towards $1$ while the second input is heavily biased towards $0$.
    By the same reasoning as $\heat^{(0,1)}_\varepsilon$, if a gate receives a marginal induced input distribution where the probability of input $(0,1)$ is at most $\varepsilon \leq 1/2$, then the expected heat released by this gate is at most $2 \varepsilon \log(1/\varepsilon)$.
    In each layer $\ell$, we split gates up into one of three regions depending on the value $i$ for the gate $(i,i+2^{k_\ell})$ and use the same heat function for every gate in the region.
    We use closed intervals on regions to cleanly handle non-integer values; this has the side-effect of counting some wires as belonging to multiple regions and thus using different heat functions.
    We will analyze such wires as being a part of both regions and account for costs associated with both heat functions.
    This only increases our upper bound and significantly simplifies the analysis.
    Let $\alpha$ be some value in $(0,1)$ that we pick later and for now assume that $n \geq 32$.
    
    \paragraph{Layer 1}
    In this layer, we have that $2^{k_\ell} = n/2$.
    The \emph{upper region} contains all gates with $i \in \left[0, (n-n^{(1+\alpha)/2})/4\right]$ and uses the heat function $\heat^\unif$.
    Since we know that $w_0(0) \leq \ldots \leq w_{n/2-1}(0)$ and $w_{n/2}(0) \leq \ldots \leq w_{n-1}(0)$ and the input was distributed uniformly at random, \cref{prop:hoeffding} gives us that the probability that $w_i(0) = 1$ as well as the probability that $w_{i+n/2}(0) = 1$ are each at most $e^{-n^{\alpha}/4}$.
    Thus by a union bound, the probability that gate $(i,i+n/2)$ receives input $(0,0)$ is at least $1-2 \cdot e^{-n^{\alpha}/4}$ and so the expected released heat by every such gate is at most $2  e^{-n^{\alpha}/4}$.
    
    The \emph{middle narrow band} contains all gates with $i \in \left[(n-n^{(1+\alpha)/2})/4, (n+n^{(1+\alpha)/2})/4\right]$ and uses the heat function $\heat^\unif$.
    There are a total of $n^{(1+\alpha)/2}/2 + 1$ gates in this band that each release at most $1$ expected heat.

    The \emph{lower region} contains all gates with $i \in \left[(n+n^{(1+\alpha)/2})/4, n/2-1\right]$ and uses the heat function $\heat^\unif$.
    Using a symmetrical argument to the upper region case, \cref{prop:hoeffding} gives us that the expected released heat of every gate in this region is at most $2e^{-n^{\alpha}/4}$.
    
    Putting all of this together, the heat released in the first layer is at most $n^{(1+\alpha)/2}/2 + 2ne^{-n^{\alpha}/4}+1$.

    \paragraph{Layer 2}
    In this layer, we have that $2^{k_\ell} = n/4$.
    The \emph{upper narrow band} contains all gates with $i \in \left[n/4, (n+n^{(1+\alpha)/2})/4\right]$ and uses the heat function $\heat^\unif$.
    There are at most $n^{(1+\alpha)/2}/4 +1$ gates in this band that each only release at most $1$ expected heat.

    The \emph{middle region} contains all gates with $i \in \left[(n+n^{(1+\alpha)/2})/4,(n-n^{(1+\alpha)/2})/2\right]$ and uses the heat function $\heat^{(1,0)}_\varepsilon$ for some $\varepsilon \leq 1/2$ specified later.
    By the construction of Batcher's odd-even mergesort, we know that $w_{i+n/4}(1) = \max(w_{i+n/4}(0),w_{i-n/4}(0))$.
    Thus, by \cref{prop:hoeffding}, the probability that $w_{i+n/4}(1) = 1$ is at most the probability that $w_{i+n/4}(0) = 1$ or $w_{i-n/4}(0) = 1$ which is at most $2e^{-n^\alpha/4}$.
    Thus the probability each gate in this region receives the input $(0,1)$ is also at most $2 e^{-n^\alpha / 4}$.
    Picking $\varepsilon = 2 e^{-n^\alpha/4}$ and $\alpha = 1/2$ gives us that $\varepsilon \leq 1/2$ and thus the expected heat released by any gate in this region is at most $\log(e) e^{-n^\alpha/4} n^\alpha$.

    The \emph{bottom narrow band} contains all gates with $i \in [(n-n^{(1+\alpha)/2})/2, n/2-1]$ and uses the heat function $\heat^\unif$.
    There are only $n^{(1+\alpha)/2}/2$ gates in this band that each only release at most $1$ expected heat.

    Putting all of this together, the heat released in the second layer is at most $n^{(1+\alpha)/2}/2 + \log(e) e^{-n^\alpha/4} n^{1+\alpha} + 1$.

    \paragraph{Layer 3 and beyond}
    The following claim gives us that, after layer $2$, most $0$'s are already in the top half of wires and most $1$'s are already in the bottom half of wires.
    \begin{restatable}{claim}{gateshavebiasedhalvesclaim}\label{claim:gates-have-biased-halves}
        Let $\alpha \in (0,1)$.
        Except with probability $3n e^{-n^{\alpha}/4}/2$, the following hold for every $\ell \geq 2$:
        \begin{enumerate}
            \item The number of $i$ such that $i < n/2$ and $w_i(\ell) = 1$ is at most $n^{(1+\alpha)/2}/2$.
            \item The number of $i$ such that $i \geq n/2$ and $w_i(\ell) = 0$ is at most $n^{(1+\alpha)/2}/2$.
        \end{enumerate}
    \end{restatable}
    \begin{proof}
        We start by proving this for $\ell = 2$.
        
        For $i < n/4$, $w_i(2) = w_i(1) = \min(w_i(0), w_{i+n/2}(0))$.
        We have already established using \cref{prop:hoeffding} that when $i \leq (n-n^{(1+\alpha)/2})/4$, each $w_i(0)=0$ except with probability at most $e^{-n^\alpha/4}$.
        By a union bound, the probability that there exists such an $i$ with $w_i(2) = 1$ is at most $n e^{-n^\alpha/4}/4$.
        We give no bounds on the values of $w_i(2) = w_i(1)$ for $i \in [(n-n^{(1+\alpha)/2})/4, n/4-1]$, but we note that there are only $n^{(1+\alpha)/2}/4$ such wires.

        For $i \in [n/4, n/2-1]$, $w_i(2) = \min(w_i(1), w_{i+n/4}(1)) = \min(w_i(1), \max(w_{i+n/4}(0), w_{i-n/4}(0)))$.
        Again we have already established using \cref{prop:hoeffding} that, when $i \in [n/4, n/2 - n^{(1+\alpha)/2}/4]$, both $w_{i+n/4}(0)$ and $w_{i-n/4}(0)$ each equal $0$ except with probability at most $e^{-n^\alpha/4}$.
        Thus by a union bound, their max is $0$ except with probability at most $2 e^{-n^{\alpha}/4}$.
        By a union bound, the probability that there exists such an $i$ with $w_i(2)=1$ is at most $n e^{-n^\alpha/4}/2$.
        We give no bounds on the values of $w_i(2)$ when $i \in [n/2 - n^{(1+\alpha)/2}/4, n/2-1]$, but we again note that there are only $n^{(1+\alpha)/2}/4$ such wires.

        For $i \in [n/2,n-1]$ the proof is symmetric.
        By a union bound, the total failure probability across both halves of the argument is at most $3n e^{-n^{\alpha}/4}/2$ and we have excluded only $n^{(1+\alpha)/2}/2$ wires from each half of the analysis.
        Note that after applying each comparator gate, the number of $1$'s in the first $n/2$ wires can only decrease while the number of $1$'s in the last $n/2$ wires can only increase.
        Thus the bound will automatically extend to all $\ell > 2$.
    \end{proof}

    Next, we partition the gates in layer $\ell \geq 3$ into three groups each.
    The \emph{upper gates} are $(i,i+2^{k_\ell})$ such that $i+2^{k_\ell} < n/2$ and they will use the heat function $\heat^\unif$.
    By \cref{claim:gates-have-biased-halves}, except with probability $3n e^{-n^{\alpha}/4}/2$, at most $n^{(1+\alpha)/2}/2$ upper gates receive an input other than $(0,0)$.
    Thus, except with probability at most $3n e^{-n^{\alpha}/4}/2$, the upper gates collectively release at most $n^{(1+\alpha)/2}/2$ expected heat.
    With probability at most $3n e^{-n^{\alpha}/4}/2$, every upper gate releases at most $1$ heat.
    Since there are at most $n/2$ such gates in a layer, the expected heat released by all upper gates is at most $n^{(1+\alpha)/2}/2 + 3n^2 e^{-n^{\alpha}/4}/4$.

    The \emph{lower gates} are $(i,i+2^{k_\ell})$ such that $i \geq n/2$ and they will use the heat function $\heat^\unif$.
    The reasoning is symmetric to the upper gates, so we get that the expected heat released by all lower gates is also at most $n^{(1+\alpha)/2}/2 + 3n^2 e^{-n^{\alpha}/4}/4$.

    The \emph{crossover gates} are $(i, i+2^{k_\ell})$ where $i < n/2$ but $i+2^{k_\ell} \geq n/2$ and they will use the heat function $\heat^{(0,1)}_\varepsilon$ for some $\varepsilon \leq 1/2$ that we will specify later.
    By \cref{claim:gates-have-biased-halves}, except with probability $3n e^{-n^{\alpha}/4}/2$, there are at most $n^{(1+\alpha)/2}/2$ crossover gates that receive the input $(1,0)$.
    When this is the case, the expected heat released by all $< n$ such gates is less than $n \log(1/(1-\varepsilon)) + n^{(1+\alpha)/2} \log(1/\varepsilon)/2$.
    When this is not the case, there are $<n$ gates that each release at most $\log(1/\varepsilon)$ expected heat.
    Thus the overall expected heat released by the crossover gates is less than $n \log(1/(1-\varepsilon)) + n^{(1+\alpha)/2} \log(1/\varepsilon)/2 + 3n^2 e^{-n^{\alpha}/4} \log(1/\varepsilon)/2 \leq 2 \log(e) \varepsilon n + n^{(1+\alpha)/2} \log(1/\varepsilon)/2 + 3n^2 e^{-n^{\alpha}/4} \log(1/\varepsilon)/2$.
    Picking $\varepsilon = n^{(\alpha -1)/2}$ when $\alpha = 1/2$ makes $\varepsilon \leq 1/2$ and therefore this expression at most $2 \log(e)n^{(1+\alpha)/2} + n^{(1+\alpha)/2} \log(n^{(1-\alpha)/2})/2 + 3n^2 e^{-n^{\alpha}/4} \log(n^{(1-\alpha)/2})/2$.

    Thus, in total, the expected heat released in each layer $\ell \geq 3$ is less than $(\log(n^{(1-\alpha)/2})/2+4 \log(e))n^{(1+\alpha)/2} + (\log(n^{(1-\alpha)/2})+1)3n^2 e^{-n^{\alpha}/4}/2$.

    \paragraph{Putting it all together}
    The total expected heat released in layers $1$ and $2$ is $O(n^{(1+\alpha)/2})$ while the total for layers $\ell \geq 3$ are each $O(n^{(1+\alpha)/2} \log (n))$, which gives a total of $O(n^{(1+\alpha)/2} \log^2 n)$ across all non-recursive layers.
    Fixing our choice of $\alpha = 1/2$ makes this expression $O(n^{3/4} \log^2 n)$.
    Since layer $0$ is two recursive calls on the halves of the input, we have that if $\heat(n)$ is the heat released for inputs of size $n$, then $\heat(n) = 2 \heat(n/2) + O(n^{3/4} \log^2 n)$ for all $n\geq 32$.
    When $n<32$ we use $\heat^\unif$ for all gates and, since the number of gates in this circuit is constant, we also get a constant upper bound for the expected heat released in this base case.
    Since the last term in the recursive case is sub-linear, the master theorem gives us that $\heat(n)$ is $O(n)$.
    We use heat function $\heat^\unif$ across all recursive calls and we use both $\heat^{(0,1)}_\varepsilon$ and $\heat^{(1,0)}_\varepsilon$ for at most $1$ unique value of $\varepsilon$ each per recursive call.
    Thus there are at most $2 \log(n) + 1$ unique heat functions used in this construction.
\end{proof}

In \cref{thm:batcher-heat-released-upper-bound-overwriting} of the appendix, we give a similar proof showing that in the overwriting circuit model (see \cref{sec:prelim}), Batcher's sorting network can still be implemented with only $O(n)$ total expected heat released.

The bound in \cref{thm:batcher-heat-released-upper-bound}  is asymptotically tight.
We can show that the circuit entropy production of sorting $n$ bits is also $\Omega(n)$.

\begin{theorem}
    \label{thm:batcher-ep-omega-n}
    When given the uniform distribution over inputs in $\{0,1\}^n$, for any choice of valid heat functions for each gate in Batcher's odd-even mergesort, the circuit entropy production is $\Omega(n)$.
\end{theorem}
\begin{proof}
    We show that each base case of sorting $n=4$ inputs has constant circuit entropy production.
    Since this base case occurs $n/4$ times in the overall sorting network and circuit entropy production is always non-negative, the entropy production of the entire circuit must be $\Omega(n)$.
    We now analyze the $n=4$ sorting network, as shown in \Cref{fig:batcher-n-4}, on the uniform distribution over inputs.
    The two gates in the first layer each receive inputs $00,01,10$, and $11$ with probability $1/4$.
    By \cref{prop:ef-min-heat-fn} this means that the expected heat release of each of these gates when given the optimal heat function is $1/2$ and the expected heat released in the first layer is $1$.
    After the first layer, the first and third wires independently have value $1$ with probability $1/4$ each while the second and fourth wires independently have value $0$ with probability $1/4$ each.
    Thus by \cref{prop:ef-min-heat-fn} we get that the expected heat release of these gates when given the optimal heat function is $3/8$ and and the expected heat release in the second layer is $3/4$.
    The final gate is a bit more tricky to analyze.
    Looking at all $16$ possible inputs to the circuit, this gate receives input $00$ on all inputs to the circuit with $3$ or more $0$'s.
    Likewise the gate receives input $11$ on all inputs to the circuit with $3$ or more $1$'s.
    The gate receives input $01$ on inputs $0011$ and inputs $1100$ to the circuit while the gate receives input $10$ on each of the four remaining inputs with $2$ zeroes and $2$ ones to the circuit.
    Thus by \cref{prop:ef-min-heat-fn} the expected heat release of this gate is $(3/8) \binaryEntropy(1/3)$.
    Putting all of this together, the expected heat release of a $4$ input sorting networks is $7/4 + (3/8) \binaryEntropy(1/3)$.
    Meanwhile the circuit maps a uniformly random length $4$ bit-string to a binomial distribution on $4$ unbiased coin flips.
    The entropy of this output equals
    \begin{displaymath}
        (2\log(16) +  8\log(4) + 6\log(8/3))/16 = (8+16 + 18 + 6 \log(1/3))/16 = 21/8 - (3/8) \log(3)
    \end{displaymath}
    and so the drop in entropy is $4-(21/8 - (3/8) \log(3)) = 11/8 + (3/8) \log(3)$.
    Since the released heat is $7/4 + (3/8) \binaryEntropy(1/3)$, the circuit entropy production of the $4$ input Batcher's sorting network with optimal heat function is at precisely $1/8$.
    Since the general Batcher's odd-even mergesort network features this sub-circuit on this input distribution $n/4$ times, the overall circuit entropy production of the network is at least $n / 32$ which is $\Omega(n)$.
\end{proof}

\begin{figure}[tb]
    \centering
    \includegraphics[width=.4\linewidth]{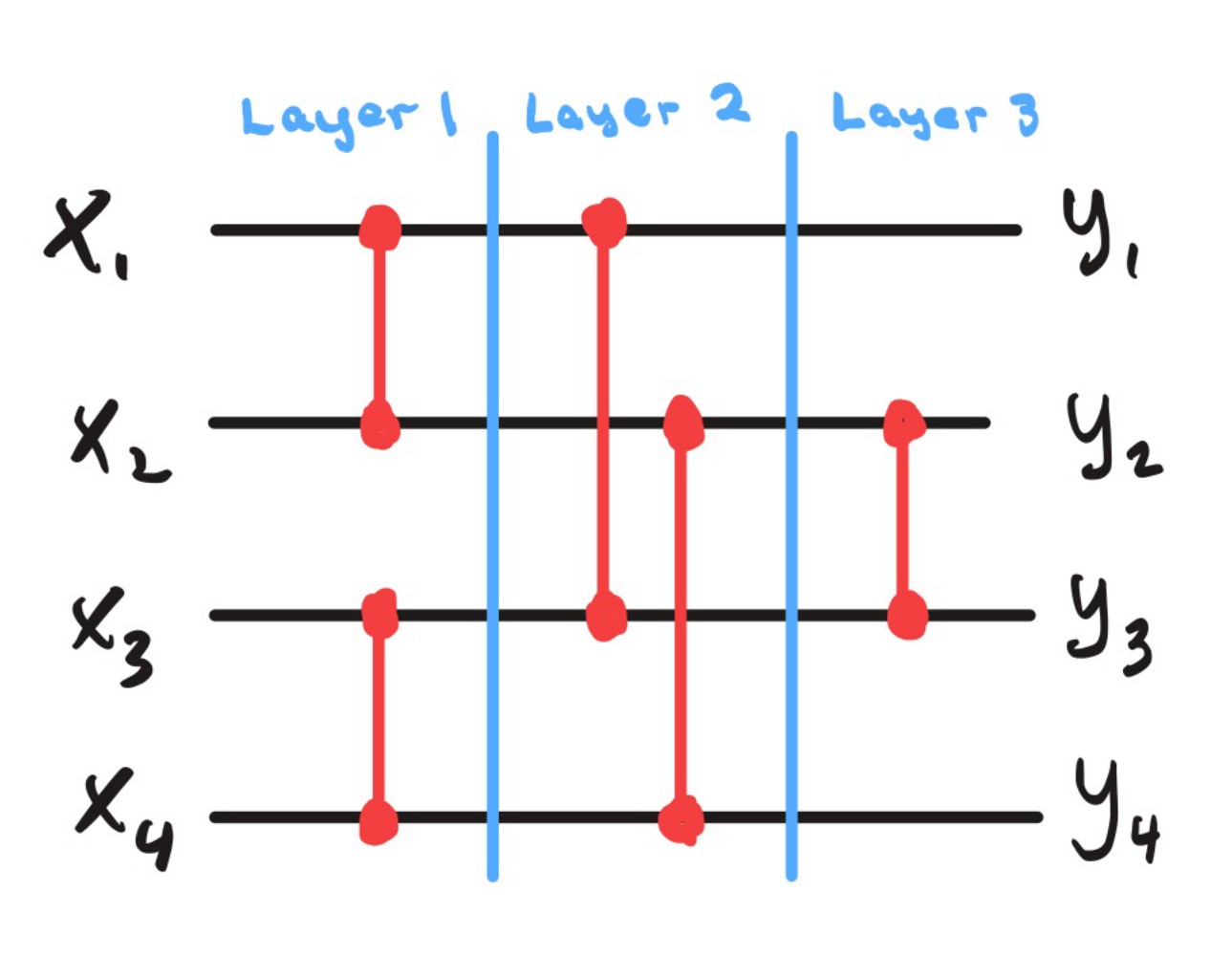}
    \caption{\label{fig:batcher-n-4} The $n=4$ base case of Batcher's odd-even mergesort.}
\end{figure}

\subsubsection{Arbitrary inputs}\label{sec:batcher-integers}
We now consider using Batcher's odd even mergesort for sorting a random permutation of the set $\{0, \ldots, n-1\}$.
This case is actually sufficient to upper bound the minimal heat dissipation for sorting arbitrary list of $n$ distinct items from an ordered set $\domain$.
\begin{lemma}\label{lem:sorting-network-domain-extension}
    Let $\circuit$ be a sorting network on $n$ inputs. Let $\dist_{n}$ be the uniform distribution over permutations of $\{0, \ldots, n-1\}$ and $\dist_\domain$ be the distribution obtained by sampling $n$ elements of $\domain$ uniformly at random without replacement.
    Then the minimum expected heat released by any gate $g$ in $\circuit$ on input distribution $\dist_{n}$ is at most the minimum expected heat released by $g$ in $\circuit$ on input distribution $\dist_\domain$.
\end{lemma}
\begin{proof}
    Let $X_g$ be a random variable over the inputs to gate $g$ induced by $\circuit$ receiving input $X$ sampled from $\dist_{\domain}$ and $V(X)$ be the unordered set of values in $X$.
    Using \cref{prop:ef-min-heat-fn}, the optimal heat function for $g$ would give a minimum expected heat release of
    \begin{align*}
        \sum_{x_1,x_2} \Pr[X_g = (x_1,x_2)] \heat(X_g) &= \sum_{x_1,x_2} \Pr[X_g = (x_1,x_2)] \log\left(\frac{\Pr[\compare(X_g) = \compare(x_1,x_2)]}{\Pr[X_g = (x_1,x_2)]}\right)\\
        &= \entropy(X_g \;| \;\compare(X_g))\\
        & \geq \entropy(X_g \;| \;\compare(X_g), V(X)).
    \end{align*}
    After conditioning on $V(X)$, we can re-label the elements in $\domain$ so that the elements in $V(X)$ get mapped to $0, \ldots, n-1$ in their relative rank order and all remaining elements get mapped to distinct values that are all at least $n$.
    After this relabeling, $X_g$ conditioned on $V(X)$ is just a random variable for the induced marginal input distribution when $\circuit$ has input distribution $\dist$.
    Thus the expected minimum heat dissipation on $\dist_\domain$ is at least that on $\dist_n$.
\end{proof}

Thus we will be assuming input distribution $\dist_n$ in the rest of this section.
By \cref{prop:ef-min-heat-fn} the optimal heat function for our comparator gates on marginal induced input distribution $\dist$ when $x_1 \neq x_2$ is given by:
\begin{align}
        \heat(x_1, x_2) & = \log\left(\Pr_{\compare(\dist)}[\compare(x_1,x_2)] / \Pr_{\dist}[(x_1,x_2)]\right) \nonumber \\
                        & = \log\left((\Pr_{\dist}[(x_1,x_2)] + \Pr_{\dist}[(x_2,x_1)]) / \Pr_{\dist}[(x_1,x_2)]\right) \nonumber \\
                        & = \log(1+\Pr_{\dist}[(x_2,x_1)] / \Pr_{\dist}[(x_1,x_2)]). \label{eqn:int-sort-heat-fn}
\end{align}

This means that to prove a $\Omega(n \log^2 n)$ lower bound on the expected heat release, it is sufficient to show that there are $\Omega(n \log^2 n)$ comparator gates where, on their induced marginal input distribution $\dist$, for every pair $(x_1,x_2) \in \support(\dist)$ we have that there are universal constants $a,b > 0$ such that $a \leq \frac{\Pr_{\dist}[(x_1,x_2)]}{\Pr_{\dist}[(x_2,x_1)]} \leq b$.
In particular, the gates that we show have this property will be ones that are centrally located and in the later layers.
An empirical plot of the expected heat released per gate in Batcher's odd-even mergesort when $n=256$ is shown in \cref{fig:batcher-n-32-heat}.

\begin{figure}[tb]
    \centering
    \includegraphics[width=1\linewidth]{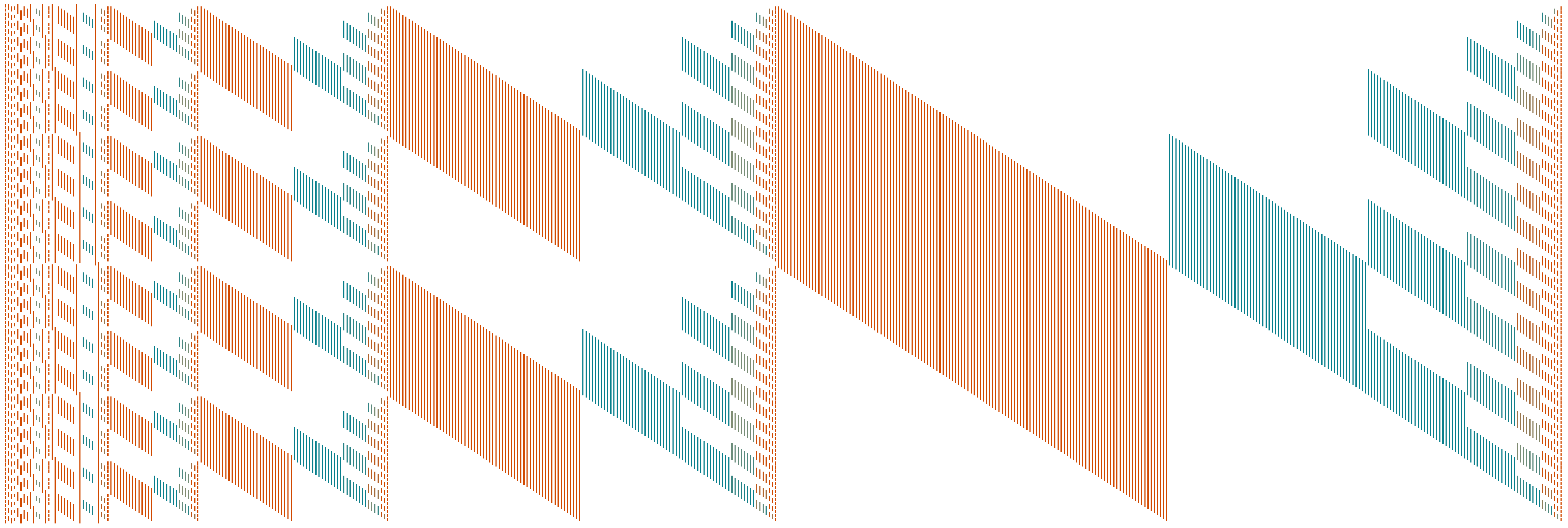}
    \caption{A plot of empirically calculated expected released heat per gate for sorting randomly permuted integers $0,\ldots,n-1$. The algorithm in \cref{lem:additive-approx-easy} is used to find approximate heat functions for each gate and then a fresh source of randomness is used to compute the empirical expected heat release. Teal indicates low heat release and orange high heat release. While it appears that there are many orange gates at the beginning of each recursive step, in the later layers this number is only $O(n)$. Thus we use the centrally located orange gates in the later layers of each recursive step to obtain our bound.
    }
    \label{fig:batcher-n-32-heat}
\end{figure}

\paragraph{Notation and network properties}
In order to show that these gates have similar probabilities for inputs $(x_1,x_2)$ and $(x_2,x_1)$, we need some properties about the Batcher's odd-even mergesort network. Some of the stated properties below are already established (see, e.g., \cite{batcher1968,Knuth98a}), but the key technical lemma that we require (i.e., \cref{lem:batcher-integer-swap-bound}) appears to be new.

\begin{fact}\label{lem:batchers-shuffle-invar}
    For every $0 \leq r < 2^{k_\ell}$, we have that $$\{w_{j}(0) \;|\; j \equiv r \bmod 2^{k_\ell}\} = \{w_{j}(\ell) \;|\; j \equiv r \bmod 2^{k_\ell}\}.$$
\end{fact}
\begin{proof}
    For every $\ell' \leq \ell$ we know that every $(i,j) \in \gateSet_{\ell'}$ has $2^{k_\ell} \; | \; (j-i)$.
    Thus, comparator gates up to level $\ell$ can only shift the index of elements within these congruency classes.
\end{proof}

The following basic fact about Batcher's odd-even sorting network follows directly from the original proof sketch of the network's correctness by Batcher~\cite{batcher1968}.
Hinze and Martin~\cite{DBLP:journals/jfp/HinzeM18} give a proof of this result similar to the one we include here.

\begin{proposition}\label{lem:batchers-partial-sorts}
    After layer $\ell$ of Batcher's odd-even mergesort, the values stored on wires whose indices are congruent mod $2^{k_\ell}$ are in their relative sorted order.
\end{proposition}
\begin{proof}
    We prove this by induction on $\ell$ increasing.
    In the base case of $\ell =1$, we have that $2^{k_{\ell}} = n/2$ and the congruent pairs are the ones that were just compared.
    Now assume that this holds after layer $\ell -1$ and let $m=n/2^{k_{\ell-1}}$.
    Fix any $r$ such that $0 \leq r < 2^{k_\ell}$. We will show that this holds for the wires whose indices are congruent to $r$ mod $2^{k_\ell}$.
    Let $e_j$ be the $j$'th wire whose index is congruent to $r$ mod $2^{k_{\ell-1}}$ and $o_j$ be the $j$'th wire whose index is congruent to $r + 2^{k_{\ell}}$ mod $2^{k_{\ell-1}}$.
    In other words $e_j$ is the wire at index $r + 2j\cdot 2^{k_{\ell}}$ and $o_j$ is the wire at index $r+(2j+1)2^{k_{\ell}}$.
    Note that for all $0\leq j<m-1$, $(o_j,e_{j+1}) \in \gateSet_\ell$.
    By the inductive hypothesis, we know that
    \begin{displaymath}
        e_j(\ell-1) \leq e_{j+1}(\ell-1) \text{ and } o_j(\ell-1) \leq o_{j+1}(\ell-1)
    \end{displaymath}
    for all $j$.
    We also know that $e_j(0) \leq o_j(0)$ since they are on the same half of the input after the recursive call in Batcher's.
    Since the $e_j(\ell-1)$ are a sorted version of the $e_j(0)$'s and likewise the $o_j(\ell-1)$ are a sorted version of the $o_j(0)$'s due to \cref{lem:batchers-shuffle-invar}, this implies that
    \begin{displaymath}
        e_j(\ell-1) \leq o_j(\ell-1)
    \end{displaymath}
    for all $j$.
    By a similar argument we know when $j\neq m/2 - 1$ it must be the case that
    $o_{j}(0) \leq e_{j+1}(0)$ by the recursive call in Batcher's.
    This maximum of one violation comparing $o_j(0)$ to $e_{j+1}(0)$ gives us that in the sorted orders before applying layer $\ell$,
    \begin{displaymath}
        o_j(\ell -1) \leq e_{j+2}(\ell-1)
    \end{displaymath}
    for all $0 \leq j < m-2$.
    Putting this together we have that:
    \begin{align*}
        e_j(\ell)& \leq \max(o_{j-1}(\ell -1), e_j(\ell-1))\\
        \leq& \min(o_{j}(\ell -1), e_{j+1}(\ell-1)) = o_j(\ell)\\
        \leq & \max(o_{j}(\ell -1), e_{j+1}(\ell-1)) = e_{j+1}(\ell)\\
        \leq &\min(o_{j+1}(\ell -1), e_{j+2}(\ell-1)) \leq  o_{j+1}(\ell)
    \end{align*}
    for all $1 \leq j < m-2$.
    We now have that $e_1(\ell) \leq o_1(\ell) \leq \ldots \leq e_{m-2}(\ell) \leq o_{m-2}(\ell)$.
    We know $o_0(\ell) \leq e_1(\ell)$ and $o_{m-2}(\ell) \leq e_{m-1}(\ell)$ since pairs $o_0(\ell), e_1(\ell)$ and $o_{m-2},e_{m-1}$ are compared in layer $\ell$.
    Also we know that $e_0(\ell) \leq o_0(\ell)$ and $e_{m-1}(\ell) \leq o_{m-1}(\ell)$ since $e_0(\ell-1) = e_0(\ell)$ is smaller than both possible values of $o_0(\ell)$ and $o_{m-1}(\ell-1) = o_{m-1}(\ell)$ is larger than both possible values of $e_{m-1}(\ell)$.
    Doing this analysis for all $r$ completes the proof.
\end{proof}
\begin{lemma}\label{lem:possible-gate-values}
    Let $g$ be the gate $(i, i+2^{k_\ell}) \in \gateSet_\ell$ for $\ell \geq 2$ and $r \equiv i (\bmod 2^{k_\ell})$.
    If Batcher's odd even mergesort receives some permutation of $[n]$ as its input, then the smaller input to $g$ is in the set $S_g \subseteq \{i+r-2^{k_\ell} +1, \ldots, i+r\}$ and the larger input to $g$ is in the set $L_g \subseteq \{i+r+1, \ldots, i+r+2^{k_\ell}\}$.
\end{lemma}
\begin{proof}
    We prove this by induction on $\ell$ decreasing.
    When $\ell = \log n$ we have that the smaller input to $g$ must be $i$ and the larger input to $g$ must be $i+1$ since the output after this layer is sorted.
    Now assume the claim holds for all gates in layer $\ell + 1$.
    Now we split into two cases.
    
    When $r < 2^{k_\ell - 1}$ then in layer $\ell+1$ we know that wire $i$ is an input to the gate $g_1 = (i-2^{k_\ell-1}, i)$ and wire $i+2^{k_\ell}$ is an input to the gate $g_2 = (i+2^{k_\ell -1}, i+2^{k_\ell})$.
    By our inductive hypothesis, this means that $S_{g} \subseteq S_{g_1} \cup L_{g_1} = \{i + r - 2^{k_\ell} + 1, \ldots, i+r\}$ and $L_g \subseteq S_{g_2} \cup L_{g_2} = \{i+r+1, \ldots, i+r+2^{k_\ell}\}$.

    When $2^{k_\ell-1} \leq r < 2^{k_\ell}$ then in layer $\ell+1$ we know that wire $i$ is an input to the gate $g_1 = (i, i+2^{k_\ell-1})$ and wire $i+2^{k_\ell}$ is an input to the gate $g_2 = (i+2^{k_\ell}, i+2^{k_\ell} + 2^{k_\ell -1})$.
    Again our inductive hypothesis gives us that $S_g \subseteq S_{g_1} \cup L_{g_1} \subseteq \{i + r - 2^{k_\ell} + 1, \ldots, i+r\}$ and $L_g \subseteq S_{g_2} \cup L_{g_2} = \{i+r+1, \ldots, i+r+2^{k_\ell}\}$.
\end{proof}

Next we state the key technical lemma that we will need to prove our lower bound on $\EP_\circuit$. We need the following definition of \emph{central gates} in order to state our key lemma, and will use the definition of the hypergeometric distribution in the proof:

\begin{definition}
    We say a gate $(i,i+2^{k_\ell}) \in \gateSet_\ell$ is \emph{central} if
        $\frac{3n}{8} \leq i  \leq \frac{5n}{8}$.
\end{definition}

\begin{definition}
\label{def:hyperGeo}
    Let $\hyperGeo{N}{K}{n}[k]$ denote the probability that when $n$ items are taken without replacement from a set of $N$ items of which exactly $K$ are marked that exactly $k$ marked items are taken. This is given by the hypergeometric distribution:
    \[
    \hyperGeo{N}{K}{n}[k] = \binom{N-K}{n-k}\binom{K}{k} / \binom{N}{n},
    \]
    which has mean $nK/N$. The distribution may have one or two modes, which differ by at most 1 from the mean.
\end{definition}

\begin{lemma}\label{lem:batcher-integer-swap-bound}
    Let $n$ be a sufficiently large power of two and $\ell > \ceil{\frac{1}{2} \log n} + 1$ be a layer in Batcher's odd-even mergesort on $n$ inputs.
    Let $g = (i,i+2^{k_\ell}) \in \gateSet_\ell$ be a central gate and $\dist$ be the marginal distribution over input pairs $(x_1,x_2)$ to $g$ induced by giving the sorting network a uniformly random permutation of the set $[n]$ as input.
    Then there is some constant $\newbounda > 1$ (independent of $n$) such that  for every $(x_1,x_2) \in \support(\dist)$,
    \begin{displaymath}
        \frac{1}{\newbounda} \leq \frac{\Pr_{\dist}[(x_2,x_1)]}{\Pr_{\dist}[(x_1,x_2)]} \leq \newbounda.
    \end{displaymath}
\end{lemma}
Before we more formally prove this lemma, we give a detailed and more intuitive outline of the strategy.
We want to show that if there are $k$ permutations of $[n]$ resulting in a gate $g = (i, i + 2^{k_\ell})$ receiving any feasible input $x_1,x_2$, then there must be at least $k/\newbounda$ permutations of $[n]$ resulting in $g$ receiving input $x_2,x_1$.
To prove this, we first note that due to the structure of Batcher's odd even mergesort, if
$r \equiv i (\bmod 2^{k_\ell})$ then
the top wire inputs to $g$ will have some fixed rank $t$ among the values of all wires after layer zero (the recursive step, which separately sorts the first and second half of the inputs independently) whose indices are congruent to $r$ mod $2^{k_\ell}$,
and bottom wire inputs to $g$ will have rank $t+1$ among the values of all wires after layer zero whose indices are congruent to $r$ mod $2^{k_\ell+1}$.
This observation is an immediate consequence of \cref{lem:batchers-partial-sorts}.
We call these sets $\oddSet$ and $\evenSet$ respectively, since they correspond to indices $r$ plus either odd or even multiples of $2^{k_\ell}$ respectively.

With this, we now know that the top input to $g$ is always the $t$'th ranked element of $\oddSet$ while the bottom input to $g$ is always the $t+1$'st ranked element of $\evenSet$.
To prove our lemma, it is sufficient to show that the number of permutations of $[n]$ that lead to $x_1$ being the $t$'th ranked element of $\oddSet$ and $x_2$ being the $t+1$'st ranked element of $\evenSet$ is within a factor of $\newbounda$ of the number of permutations resulting in $x_2$ being the $t$'th ranked element of $\oddSet$ and $x_1$ being the $t+1$'st ranked element of $\evenSet$.

Unfortunately, this is complicated by the fact that the indices of elements after layer zero (and hence our sets $\oddSet$ and $\evenSet$) are the result of independently sorting the first half and the second half of the input.
Thus the sets $\oddSet$ and $\evenSet$ are determined entirely by which inputs started on which half of the wires.
We introduce sets $A$ and $B$ containing the values stored on the first and second half of the wires at the beginning of the circuit respectively.
Note that, for the uniformly random permutation of $[n]$ as input, the induced equipartition of $[n]$ into $A$ and $B$ is selected uniformly at random.
Then the sets $\oddSet$ and $\evenSet$ are entirely determined by $A$ and $B$.

Let $\xmin$ and $\xmax$ be the smaller and larger values respectively among $x_1,x_2$.
We observe that shifting exactly $2^{k_\ell}$ elements from $A$ to $B$ (or from $B$ to $A$) that are smaller than $\xmin$ and moving $2^{k_\ell}$ elements larger than $\xmax$ back results in flipping the sets $\oddSet \cap [\xmin,\xmax]$ and $\evenSet \cap [\xmin,\xmax]$, which has the effect of swapping the order of the two inputs to $g$.
Moreover for any $A$ and $B$ resulting in $g$ receiving input $x_1,x_2$, swapping any number of elements smaller than $\xmin$ between $A$ and $B$ and doing likewise with elements larger than $\xmax$ will maintain the fact that $g$ receives input $x_1,x_2$.
Thus we want to compare the number of possible $A$ with $s$ elements smaller than $\xmin$ and a fixed subset $I$ of $\{\xmin, \ldots, \xmax\}$ in $A$ with the number of possible $A'$ with $s \pm 2^{k_\ell}$ elements smaller than $\xmin$ and the same fixed subset of $I$ in $A'$.
The remainder of the proof is a (somewhat cumbersome) argument that these values are within a factor of $\alpha$ of one another.

We use $\calA_I^s$ to denote the set of all $A$ that have $A \cap \{\xmin, \ldots, \xmax\} = I$ and exactly $s$ elements smaller than $\xmin$ in $A$.
The relative size of these sets for different values of $s$ are hyper-geometrically distributed; the `marked' items are elements smaller than $\xmin$ and we want to count how many of them end up in set $A$.
Since the probability mass function of the hyper-geometric distribution starts off monotonically non-decreasing until hitting some mode $s^*$ and then is monotonically non-increasing, we immediately have that when $\abs{s-s^*}>2^{k_\ell}$ it must be the case that $\abs{\calA_I^s} \leq \max(\abs{\calA_I^{s-2^{k_\ell}}}, \abs{\calA_I^{s+2^{k_\ell}}})$ and we are done.

When $\abs{s-s^*}<2^{k_\ell}$ the analysis is more involved.
We show that the ratio of the squared geometric mean of $\abs{\calA_I^{s-2^{k_\ell}}}$ and $\abs{\calA_I^{s+2^{k_\ell}}}$ to the value of $\abs{\calA_I^{s}}^2$ is at least some constant, $1/\newbounda^2$, by 
expressing $\abs{\calA^{s-2^{k_\ell}}_I}\abs{\calA^{s+2^{k_\ell}}_I}/\abs{\calA^{s}_I}^2$ as a telescoping product.
With this lower bound in hand,  we can use the arithmetic, geometric mean inequality to deduce that $\abs{\calA^{s-2^{k_\ell}}_I}+\abs{\calA^{s+2^{k_\ell}}_I}\geq \abs{\calA^{s}_I} (2/\newbounda)$.
This in turn gives us that if $C(a,b)$ is the number of choices for $A$ resulting in $g$ receiving input $a,b$ then we always have that $1/\newbounda \leq \frac{C(x_2,x_1)}{C(x_1,x_2)} \leq \newbounda$ as desired.

\begin{proof}[Proof of \cref{lem:batcher-integer-swap-bound}]
    Let $r \equiv i (\bmod 2^{k_\ell})$.
    As in the proof of \cref{lem:batchers-partial-sorts}, we define wires $e_j$ indexed at $r+2j \cdot 2^{k_\ell}$ and $o_j$ indexed at $r+(2j+1)2^{k_\ell}$.
    Let $\evenSet$ and $\oddSet$ be the set of values found in the $e_j(0)$ and the $o_j(0)$ respectively.
    By \cref{lem:batchers-partial-sorts} we know that the $e_j(\ell-1)$ and the $o_j(\ell-1)$ contain the values in $\evenSet$ and $\oddSet$ respectively in  sorted orders.
    Let $t$ be the index such that $g = (o_t, e_{t+1})$.
    Thus we know that the input to $g$ is the pair $(o_t(\ell-1), e_{t+1}(\ell-1))$.
    
    Next we characterize the sets $\evenSet$ and $\oddSet$.
    Let $A$ and $B = [n] \setminus A$ be the uniformly random equipartition of our input elements into the first and second halves of the wires respectively at the beginning of the circuit.
    If $a_0 < a_1 < \ldots < a_{n/2-1}$ are the elements of $A$ and $b_0 < b_1 < \ldots < b_{n/2-1}$ are the elements of $B$, then we know that $w_i(0) = a_i$ for all $i < n/2$ and $w_i(0) = b_{i-n/2}$ for all $i \geq n/2$ at the end of layer zero.
    Then we can characterize our sets as
    \begin{displaymath}
        \evenSet = \{a_j \; | \; j \equiv r (\bmod 2^{k_\ell + 1} )\} \cup \{b_j \; | \; j \equiv r (\bmod 2^{k_\ell + 1} )\}
    \end{displaymath}
    and
    \begin{displaymath}
        \oddSet = \{a_j \; | \; j \equiv r + 2^{k_\ell} (\bmod 2^{k_\ell + 1} )\} \cup \{b_j \; | \; j \equiv r + 2^{k_\ell} (\bmod 2^{k_\ell + 1} )\}.
    \end{displaymath}
    Thus we know that any $x \in \evenSet$ exactly when $x=a_j$ or $x=b_j$ for some $j \equiv r (\bmod 2^{k_{\ell}+1})$.
    Likewise $x \in \oddSet$ exactly when $x=a_j$ or $x=b_j$ for some $j \equiv r + 2^{k_\ell} (\bmod 2^{k_{\ell}+1})$.

Now define $\xmin$ and $\xmax$ as the min and the max respectively of $x_1$ and $x_2$. It follows from \cref{lem:possible-gate-values} that 
\begin{equation}
\label{eqn:xminmaxdiff}
\xmin =  i \pm O(\sqrt{n}) \;\;\;\;\; \xmax =  i \pm O(\sqrt{n}), \;\;\;\;\; \mbox{ and } 1 \le \xmax-\xmin \le 2\sqrt{n}. 
\end{equation}
Since $g = (i, i + 2^{k_{\ell}})$ is central, we have that 
\begin{equation}
\label{eqn:xminmaxbounds}
3n/8 - O(\sqrt{n}) \le \xmin \le \xmax \le 5n/8 + O(\sqrt{n}),
\end{equation}
and so both $\xmin$ and $\xmax$ are $\Theta(n)$.
    
Observe that moving precisely 
    $2^{k_\ell}$ elements from $A$ to $B$ (or from $B$ to $A$) that are smaller than $\xmin$ and moving $2^{k_\ell}$ elements larger than $\xmax$ back results in flipping the sets $\evenSet \cap [\xmin,\xmax]$ and $\oddSet \cap [\xmin,\xmax]$.\footnote{When shifting small elements from $B$ to $A$, each old $a_j$ with a value in this interval becomes $a_{j+2^{k_\ell}}$ and each old $b_j$ with a value in this interval becomes $b_{j-2^{k_\ell}}$.}
    In particular such a change would swap the values of $o_{t}(\ell-1)$ and $e_{t+1}(\ell -1)$, as the even and the odd sets swapped their values in this interval and the moved elements made the number of elements smaller than $\xmin$ now in $\evenSet$ one more than the number of elements smaller than $\xmin$ previously in $\oddSet$.\footnote{As an illustrative example of this, consider $\ell = 0$ so $\evenSet$ and $\oddSet$ are just the even and odd indices respectively of $A$ and $B$ when sorted. Fix $A=\{0,1,2,3\}$ and $B = \{4,5,6,7\}$. Then $\evenSet = \{0,2,4,6\}$ and $\oddSet = \{1,3,5,7\}$.
    If we swap values $0$ and $7$ between $A$ and $B$, then we get $A' = \{1,2,3,7\}, B' = \{0,4,5,6\}$ and $\evenSet' = \{0,1,3,5\}, \oddSet' = \{2,4,6,7\}$. We have that $e'_{t+1}= o_t$ and $o'_t = e_{t+1}$ for $t\in\{0,1,2\}$.}
    This swap would have the effect of replacing input $(x_1,x_2)$ with $(x_2,x_1)$ and vice versa.
    
    Let $\calA^s_I$ be the (possibly empty) subset of equipartitions $A$ where $I=[\xmin,\xmax] \cap A$ and $s=\abs{[0,\xmin) \cap A}$.
    Importantly, all equipartitions in $\calA^s_I$ result in $g$ receiving the same input.
    Then we can precisely characterize the size of these sets by:
    \begin{equation}\label{eqn:hypergeo-counts}
        \abs{\calA^s_I} = \hyperGeo{(n-\xmax + \xmin -1)}{\xmin}{n/2 - \abs{I}}[s] \cdot \binom{n- \xmax + \xmin -1}{n/2 - \abs{I}} = \binom{\xmin}{s} \binom{n - \xmax - 1}{n/2 - \abs{I} - s}.
    \end{equation}
        To see why the hypergeometric distribution applies, note that $\abs{I}$ elements of $A$ are already fixed, leaving $n/2-\abs{I}$ elements of $A$ that can be chosen from the set $[0,n-1] - [\xmin,\xmax]$, which has size $n - \xmax + \xmin -1$.
 Mark the $\xmin$ elements in the range $[0,\xmin-1]$; from these, $s$ must be chosen. The remaining $n/2 - \abs{I} - s$ items are chosen from the unmarked elements in the set $[\xmax+1,n-1]$ which has size $n - \xmax -1$.
 
 From \cref{def:hyperGeo}, the mean of this distribution is 
 \[
\frac{\xmin (n/2-\abs{I})}{n - (\xmax - \xmin +1)} .
\]
 Since $\abs{I} \le \xmax - \xmin + 1 \le 2^{k_{\ell}+1} \le 2 \sqrt{n}$, and since $\xmin = \Theta(n)$, the mean and modes (which differ from the mean by at most 1) are $\xmin/2 \pm O(\sqrt{n})$. We denote the larger mode of this distribution by $s^*$, and so we have
 \begin{equation}
 \label{eqn:modebound}
 s^* = \xmin/2 \pm O(\sqrt{n})
 \end{equation}
    
    For any $s$ and $I$ where $g$ receives the input pair $(x_1,x_2)$ we know that both $\calA^{s+2^{k_\ell}}_I$ and $\calA^{s-2^{k_\ell}}_I$ results in $g$ receiving the input pair $(x_2,x_1)$.
    We will want to show that there is some constant $\newbounda > 1 $ such that for every choice of $s$ and $I$,
    \begin{equation}\label{eqn:shifted-sets-large}
        \abs{\calA^{s-2^{k_\ell}}_I} + \abs{\calA^{s+2^{k_\ell}}_I} \geq  \abs{\calA^{s}_I}(2/ \newbounda).
    \end{equation}

    Since the hyper-geometric distribution is unimodal, i.e., strictly non-decreasing to $s^*$, the larger mode, and then strictly non-increasing
    after that point, for any $s \leq s^* - 2^{k_\ell}$ we immediately have that $\abs{\calA^{s}_I} \leq \abs{\calA^{s+2^{k_\ell}}_I}$.
    Likewise for any $s \geq s^* + 2^{k_\ell}$ we immediately have that $\abs{\calA^{s}_I} \leq \abs{\calA^{s-2^{k_\ell}}_I}$.

    When $\abs{s-s^*} < 2^{k_\ell}$ we give a lower bound on  $\abs{\calA^{s-2^{k_\ell}}_I}\abs{\calA^{s+2^{k_\ell}}_I}/\abs{\calA^{s}_I}^2$ and then use the arithmetic geometric mean equality to deduce that $\abs{\calA^{s-2^{k_\ell}}_I} + \abs{\calA^{s+2^{k_\ell}}_I} \geq (2/\newbounda) \abs{\calA^{s}_I}$.
Using a telescoping product, we have that:
    \[
        \frac{\abs{\calA^{s-2^{k_\ell}}_I}\abs{\calA^{s+2^{k_\ell}}_I}}{\abs{\calA^{s}_I}^2} \\
        = \prod_{j=1}^{2^{k_\ell}} \frac{\abs{\calA_I^{s-j}}}{\abs{\calA_I^{s-j+1}}} \frac{\abs{\calA_I^{s+j}}}{\abs{\calA_I^{s+j-1}}}. \\
\]
We first show that the telescoping product is well-defined, i.e., that all terms $\abs{\calA^{w}_I}$, $s-2^{k_{\ell}} \le w \le s + 2^{k_{\ell}}$, are positive.
The left binomial term of \cref{eqn:hypergeo-counts} for $\abs{\calA^{w}_I}$ is $\binom{\xmin}{w}$. 
From \cref{eqn:xminmaxbounds} we have that $\xmin = \Theta(n)$. Also, 
$\abs{w-s^*} \le \abs{w-s} + \abs{s - s^*} \le 2\sqrt{n}$ and so, since $s^* = \xmin/2 + O(\sqrt{n})$ by \cref{eqn:modebound}, we have 
\[
w = \xmin/2 \pm O(\sqrt{n}).
\]
Therefore for sufficiently large $n$, since $\xmin = \Theta(n)$, $w$ is positive and less than the upper index, $\xmin$, so the first binomial term, $\binom{\xmin}{w}$, is positive.

Now consider the right binomial term of \cref{eqn:hypergeo-counts} for $\abs{\calA^{w}_I}$, which is $\binom{n - \xmax - 1}{n/2 - \abs{I} - w}$.
We have that $\xmax \le 5n/8 + O(\sqrt{n})$, so the upper index of the binomial is at least $3n/8 - O(\sqrt{n})$. 
Using the expression for $w$ above, as well as the fact that $\abs{I} \le \xmax - \xmin + 1 = O(\sqrt{n})$,
\cref{eqn:xminmaxdiff}, we can bound the lower index as follows:
\[
n/2 - \abs{I} - w = n/2 - \xmin/2 \pm O(\sqrt{n}) \ge n/2 - \xmax/2 \pm O(\sqrt{n}),
\]
which again is positive and is less than the upper index for sufficiently large $n$.

We now expand the telescoping product above:
   \begin{align*}
        & \frac{\abs{\calA^{s-2^{k_\ell}}_I}\abs{\calA^{s+2^{k_\ell}}_I}}{\abs{\calA^{s}_I}^2} \\
        &= \prod_{j=1}^{2^{k_\ell}} \frac{\abs{\calA_I^{s-j}}}{\abs{\calA_I^{s-j+1}}} \frac{\abs{\calA_I^{s+j}}}{\abs{\calA_I^{s+j-1}}}\ \\
        &= \prod_{j=1}^{2^{k_\ell}} 
        \frac{(s-j+1)(n/2-\xmax + \abs{I} + s - j)}{(\xmin - s + j)(n/2 - \abs{I} - s + j)}\frac{\abs{\calA_I^{s+j}}}{\abs{\calA_I^{s+j-1}}} \\
        &= \prod_{j=1}^{2^{k_\ell}} 
        \frac{(s-j+1)(n/2-\xmax + \abs{I} + s - j)}{(\xmin - s + j)(n/2 - \abs{I} - s + j)}\frac{(\xmin - s-j+1)(n/2 - \abs{I} -s-j+1)}{(s+j)(n/2 - \xmax + \abs{I} + s+j-1)} \\
        &=  \prod_{j=1}^{2^{k_\ell}} \hspace{-2pt} \left( \frac{s-j+1}{s+j} \right) \hspace{-2pt} \left( \frac{\xmin - s -j+1}{\xmin - s + j}\right) \hspace{-2pt} \left( \frac{n/2 - \abs{I} -s -j +1}{n/2 - \abs{I} - s + j} \right) \hspace{-2pt} \left( \frac{n/2 - \xmax - 1 + \abs{I} + s - j + 1}{n/2 - \xmax - 1 + \abs{I} + s + j} \right) .\\
    \end{align*}
In the last line, inside the product there are four terms, each of the form $\frac{z-j +1}{z+j}$, where, from our previous derivations, $z = \Theta(n)$ and $j \le \sqrt{n}$. Therefore for some positive constant $c$, each of these terms is at least $1 - c/\sqrt{n}$. Therefore, since $k_{\ell} \le \sqrt{n}$,
   \[
 \frac{\abs{\calA^{s-2^{k_\ell}}_I}\abs{\calA^{s+2^{k_\ell}}_I}}{\abs{\calA^{s}_I}^2} \ge (1 - c/\sqrt{n})^{4 \sqrt{n}} \ge (1/4)^{4c},
   \]
where for the last inequality we use the fact that $1/4 \le (1 - 1/k)^k$ for all $k \ge 2$.

Using the arithmetic geometric mean inequality,
    \begin{align*}
        \abs{\calA^{s-2^{k_\ell}}_I} + \abs{\calA^{s+2^{k_\ell}}_I} &\geq 2\sqrt{\abs{\calA^{s-2^{k_\ell}}_I} \abs{\calA^{s+2^{k_\ell}}_I}}. \\
        &\geq 2 (1/4)^{2c}  \abs{\calA^{s}_I} = (2/\newbounda) \abs{\calA^{s}_I},
    \end{align*}
    where the constant $\newbounda$ is $4^{2c}$.
 We have now proven \eqref{eqn:shifted-sets-large} in all cases, 
    Let $C(x_1,x_2)$ denote the number of choices of set $A$ such that $g$ will receive input $(x_1,x_2)$.
    Then:
    \begin{displaymath}
        C(x_1,x_2) = \sum_{\substack{s,I\\ \calA_I^s \text{ gives input } (x_1,x_2)}} \abs{\calA^s_I} \leq (\newbounda/2) \sum_{\substack{s,I\\ \calA_I^s \text{ gives input } (x_1,x_2)}} \abs{\calA^{s-2^{k_\ell}}_I} + \abs{\calA^{s+2^{k_\ell}}_I} \leq \newbounda C(x_2,x_1),
    \end{displaymath}
    where the last inequality comes from the fact that this sum is counting each possible $A$ at most twice.
    By symmetry the same argument gives that $C(x_2,x_1) \leq \newbounda C(x_1,x_2)$.
    Since every equipartition is equally likely we can conclude that:
    \begin{displaymath}
        1/\newbounda \leq \frac{C(x_2,x_1)}{C(x_1.x_2)} = \frac{\Pr_{\dist}[(x_2,x_1)]}{\Pr_{\dist}[(x_1,x_2)]} \leq  \newbounda
    \end{displaymath}
    as desired.    
\end{proof}

\paragraph{Thermodynamic bounds}
Now that we have proven \cref{lem:batcher-integer-swap-bound}, we are ready to prove our thermodynamic bounds.
\begin{lemma}\label{lem:hard-gate-ef-min}
    Let $\domain$ be an ordered set and $\dist$ be the distribution resulting from picking $n$ elements of $\domain$ without replacement uniformly at random.
    Let $n$ be sufficiently large so that \cref{lem:batcher-integer-swap-bound} applies.
     Any gate in Batcher's odd-even mergesort where \cref{lem:batcher-integer-swap-bound} can be applied must release $\Omega(1)$ expected heat.
\end{lemma}
\begin{proof}
    Assume for now that $\domain = \{0,\ldots, n-1\}$.
    By \eqref{eqn:int-sort-heat-fn} and \cref{lem:batcher-integer-swap-bound} we know that the optimal heat function for such a gate with induced marginal input distribution $\dist$ is given by $\heat(x_1,x_2) = \log(1+\frac{\Pr_{\dist}[(x_2,x_1)]}{\Pr_{\dist}[(x_1,x_2)]}) \geq \log(1 + 1/\newbounda) = \Omega(1)$.
    Since this applies to every input pair $x_1,x_2$, the expected heat dissipation is $\Omega(1)$.
    This lower bound also applies to arbitrary sets $\domain$ due to \cref{lem:sorting-network-domain-extension}.
\end{proof}
\begin{theorem}\label{thm:batcher-entropy-integer-high}
    Let $\domain$ be an ordered set and $\dist$ be the distribution resulting from picking $n$ elements of $\domain$ without replacement uniformly at random.
    Batcher's odd-even mergesort must release $\Omega(n \log^2 n)$ expected heat on input distribution $\dist$.
    When using the optimal heat function for each gate, this value is $\Theta(n \log^2 n)$.
\end{theorem}
\begin{proof}
    We do this by counting the number of gates where we can apply \cref{lem:hard-gate-ef-min}.
    Let $C(n)$ denote the count of such gates in the circuit of size $n$.
    In the non recursive step, when $n$ is sufficiently large, the admissible gates are those in layers $\ell > \ceil{(\log n)/2} + 1$ that are central.
    As in the proof of \cref{lem:batcher-integer-swap-bound}, let $k_\ell = (\log n) - \ell$ be the width of gates in layer $\ell$.
    There are at least $n/4$ choices of $i$ such that $3n/8 \leq i \leq 5n/8$.
    By \eqref{eqn:gate-locs}, at least half of these $i$ have the correct value of $\floor{i/2^{k_\ell}}$ to be either the top or the bottom node of a gate.
    Of these $n/8$ potential gates, at most $2^{k_\ell}$ of them do not actually exist due to the constraint $2^{k_\ell} \leq i < n-2^{k_\ell}$ on the starting point of gates.
    After all this, the number of gates in layer $\ell$ is at least $n/8 - 2^{k_\ell} \geq n/9$ for sufficiently large $n$.
    The number of considered layers is $(\log n) - (\ceil{(1/2)\log n} + 1) \geq (\log n)/4$.
    Thus the number of relevant gates in the non recursive step is at least $(n \log(n))/36$ or $\Theta(n \log n)$ with the asymptotic upper bound coming from the total number of gates in the non-recursive step.
    Considering in the recursive step, we have:
    $$C(n) = 2C(n/2) + \Theta(n \log n)$$
    Which is $\Theta(n \log^2 n)$ by the Master theorem.

    By \cref{lem:hard-gate-ef-min}, each of these gates must release $\Omega(1)$ heat in expectation.
    Thus the total heat dissipation of the circuit is $\Omega(n \log^2 n)$.
    To see that this bound is asymptotically tight, observe that the heat function $\heat(x_1,x_2) = 1$ is valid.
    Since there are $O(n \log^2 n)$ total gates in Batcher's odd even mergesort, this choice for $Q$ gives a total of $O(n \log^2 n)$ expected heat released.
\end{proof}
\begin{corollary}\label{cor:batcher-integer-any-set}
    Let $\domain$ be an ordered set and $\dist$ be the distribution resulting from picking $n$ elements of $\domain$ without replacement uniformly at random.
    Batcher's odd-even mergesort has $\Omega(n \log^2 n)$ circuit entropy production on input distribution $\dist$.
\end{corollary}
\begin{proof}
    By \cref{thm:batcher-entropy-integer-high} we know this circuit releases $\Omega(n \log^2 n)$ expected heat.
    The drop in entropy by sorting an input is at most $\log(n!)$ which is only $O(n \log n)$.
    Thus the released heat dominates the drop in entropy term for this circuit.
\end{proof}
We have thus shown that Batcher's odd even mergesort network is thermodynamically inefficient.

\subsection{A more thermodynamically efficient sorting network}\label{sec:better-than-batchers}
In \cref{sec:batcher-integers} we showed that Batcher's odd even mergesort is a poor choice if one is interested in optimizing the thermodynamic costs of sorting more than just bits, as the minimal expected heat released is asymptotically larger than the drop in entropy.
If one is only interested in asymptotic bounds, then there is a trivial way to get a sorting network with $O(n \log n)$ expected heat release by using the AKS sorting network which has only $O(n \log n)$ gates \cite{AjtaiKS83}.
However, this construction has such a large constant hidden in the $O$ notation that it is unsuitable for actual implementation.

\paragraph{A network sorting almost all inputs}
In \cite{LP98} the authors present a comparator network that sorts almost every permutation of $[n]$ using only $(3.72 + o(1)) n \log n$ comparator gates.
We will show how this construction can be leveraged to give a sorting network that releases only $O(n \log n)$ expected heat.
To start, we cite the relevant result from their paper.
\begin{proposition}[\cite{LP98}]\label{prop:tournament-almost-sort}
    Let $n=2^d$.
    For every $\varepsilon(d) = 2^{-2^{o(\sqrt{d})}}$ there is a deterministic (non-uniform) network of depth $(7.44+o(1)) \cdot d$ that sorts a uniformly random permutation of $[n]$ with probability at least $1-\varepsilon(d)$.
\end{proposition}
\begin{remark}
    In \cite{LP98} the authors actually prove a bound that works for every $\varepsilon(d) = 2^{-2^{o(d)}}$ however, this stronger error bound requires invoking AKS on small subproblems. Thus the enormous constant associated with AKS would still be present in the $o(1)$ term in the depth bound.
    Instead, we use a weaker version of their result that they explicitly mention at the end of section 6 where AKS is replaced by bitonic sort.
    This slightly weakens the error bounds but avoids the large constant factor in the overall cost.
\end{remark}

In addition to a depth bound that is linear in $d$, \cref{prop:tournament-almost-sort} has only a linear number of gates per layer.
Thus the overall gate count is only $O(n \log n)$.
We will only need to instantiate this with $\varepsilon(d) = 2^{-\log(d)-\log \log(d)} = 1/(\log(n) \log \log (n))$.
The following corollary suffices for our use case.
\begin{corollary}\label{cor:almost-sort-network}
    Let $n$ be a power of two, $\domain$ be an arbitrary ordered set and $\dist$ be the distribution over inputs given by sampling $n$ elements of $\domain$ uniformly at random without replacement.
    Then there exists a deterministic (non-uniform) network with $O(n \log n)$ gates that correctly sorts inputs $x_1, \ldots, x_n \sim \dist$ except with probability at most $1/(\log(n) \log \log (n))$.
\end{corollary}

\paragraph{A sorting network with low expected released heat}
Now we show that the network in \cref{cor:almost-sort-network} can be used as a subroutine in a thermodynamically efficient sorting network.
Intuitively, this network sorts most inputs using only a small number of gates.
This allows it to be followed by a larger sorting network which uses a heat function that takes advantage of the fact that its input is often already in sorted order.
\begin{theorem}\label{thm:pre-sort-helps-upper-bound}
    Let $\domain$ be an arbitrary ordered set and $\dist$ be the distribution over inputs given by sampling $n$ elements of $\domain$ uniformly at random without replacement.
    There is a thermodynamic sorting network $\circuit$ with $O(n \log^2 n)$ gates that are each augmented with one of two heat functions $\heat_0,\heat_1$ such that, on input distribution $\dist$, the expected heat dissipation of $\circuit$ is $O(n \log n)$.
\end{theorem}
\begin{proof}
    We fix $\circuit$ to be the concatenation of the sorting network in \cref{cor:almost-sort-network} on all $n$ wires (call this $\circuit_0$) followed by Batcher's odd even mergesort on the same $n$ wires (call this $\circuit_1$).
    Since Batcher's odd even mergesort is a sorting network, we know that $\circuit$ is also a sorting network.
    We augment every gate in $\circuit_0$ with the heat function $\heat_0(x_1,x_2) = 1$.
    Since $\circuit_0$ has $O(n \log n)$ gates, the expected heat dissipation of this component is also $O(n \log n)$.
    Importantly, after running $\circuit_0$, the input is already sorted except with probability at most $1/(\log(n) \log \log (n))$.
    
    We augment every gate in $\circuit_1$ with the heat function:
    \begin{displaymath}
        \heat_1(x_1, x_2) = \begin{cases}
            1/\log(n) & x_1 \leq x_2\\
            \log\left(\frac{1}{1-2^{-1/\log(n)}}\right) & x_1 > x_2
        \end{cases}
    \end{displaymath}
    The validity of $\heat_1$ follows from \cref{lem:turgut-lemma}.
    When $\circuit_1$ receives an already sorted input from $\circuit_0$, this choice of $\heat_1$ implies that every gate will dissipate only $1/\log(n)$ heat.
    Since $\circuit_1$ has $O(n \log^2 n)$ total gates, the expected released heat conditioned on $\circuit_0$ sorting its input is only $O( n \log n)$.
    Conditioned on $\circuit_0$ failing to sort its input, some gates in the network might release up to
    \begin{align*}
        - \log (1-2^{-1/\log n}) &= -\log(1-e^{-\frac{\ln 2}{\log n}})\\
        &= -\log \left(\frac{\ln 2}{\log n}\right) - \log \left(\frac{1-e^{-\frac{\ln 2}{\log n}}}{\frac{\ln 2}{\log n}}\right)\\
        &= \log \log (n) - \log(\ln(2)) - \log \left(\frac{1-e^{-\frac{\ln 2}{\log n}}}{\frac{\ln 2}{\log n}}\right)\\
        &= \log \log (n) - \log(\ln(2)) + o(1) \qquad \textrm{since $\lim_{x \to 0} \frac{1-e^{-x}}{x} = 1$}\\
        &= O(\log \log n)
    \end{align*}
    expected heat.
    Even if every gate releases this much, the expected released heat conditioned on $\circuit_0$ failing to sort its input is only $O(n \log^2(n) \log \log(n))$.
    Thus the total expected heat released by $\circuit_1$ is at most:
    \begin{displaymath}
        \left(1-\frac{1}{\log(n) \log \log (n)}\right) \cdot O(n \log^2 (n)) / \log(n) + \frac{1}{\log(n) \log \log (n)} \cdot  O(n \log^2(n)) \cdot  \log \log (n)) = O(n \log n)
    \end{displaymath}
    as desired.
    Since both $\circuit_0$ and $\circuit_1$ release only $O(n \log n)$ expected heat, so does $\circuit$.
\end{proof}
It is important to note that this construction is heavily reliant on the input distribution $\dist$.
If $\dist$ gives heavy weight to inputs that $\circuit_0$ fails to sort, then the analysis presented here falls apart.
\section{Clean zero entropy circuits for compaction and sorting}
\label{sec:zero-entropy-compaction-and-sorting-circuits}
We start by considering \emph{bit compaction} as a special case of sorting with only two possible values for each key.
We give a detailed overview of this circuit with zero circuit entropy production and then show how the idea can be extended to more general sorting.

\subsection{A clean compaction circuit with zero entropy production}
\label{sec:clean-compaction-zero-ep}

We describe a bit-sorting (a.k.a., \emph{bit-compaction}) circuit with zero $\EP_\circuit$
(assuming the distribution of the inputs is \emph{exchangeable},
explained below).
For any binary string $x$, write $|x|_1$ to denote the number of 1's in $x$ 
and $|x|_0 = |x| - |x|_1$ to denote the number of 0's.

First, we note that we are considering compaction \emph{circuits}, rather than sorting networks:
these can use more general gates than the comparators used in \cref{sec:batcher-bits,sec:batcher-integers,sec:better-than-batchers}, 
and they are allowed ancilla wires beyond those needed to store inputs and outputs.
Although the problem of sorting bits with a sorting network as in \cref{sec:batcher-bits} is nontrivial
due to~\cref{prop:zero-one-principle},
general circuits can perform compaction without payloads by simply counting the number of times each key appears in the input.
Thus, we add payloads that must be sorted alongside their keys to make the problem non-trivial.

This problem was studied by Asharov, Lin, and Shi~\cite{asharov2021sorting}, who gave an $o(n \log n)$ sized circuit for compaction (when the keys are sufficiently small) without considering the thermodynamic costs.\footnote{
\label{footnote:compaction-history}
    See also~\cite{lin2022optimal}.
    The problem has a longer history in other models of computation,
    going back to Pippenger's self-routing superconcentrators~\cite{Pip96}.\footnotemark{}
    With comparator gates only, compaction is equivalent to sorting:
    by \cref{prop:zero-one-principle}, a comparator network that sorts every 0/1 key vector sorts every input,
    so $\Theta(n \log n)$ comparators are necessary and sufficient~\cite{AjtaiKS83}.
    With general Boolean gates,
    Asharov, Lin, and Shi~\cite{asharov2021sorting} gave the first compaction circuits of size $O(nm)$ for $m$-bit payloads,
    up to a $\mathrm{poly}(\log^* n - \log^* m)$ factor that disappears when $m = O(1)$ or $m \ge \log^{(c)} n$ for any constant $c$,
    but with superpolylogarithmic depth;
    Kouck\'{y} and Kr\'{a}l~\cite{KouckyK21} reduced the depth to $O(\log^3 n)$ with size $O(nm(1 + \log^* n - \log^* m))$,
    and Lin and Shi~\cite{lin2022optimal} to the optimal $O(\log n + \log m)$ with size $O(nm) \cdot \mathrm{poly}(\log^* n - \log^* m)$.
    (Pippenger's self-routing superconcentrators~\cite{Pip96} already give $O(n \log n + nm)$ gates~\cite{asharov2021sorting};
    in the RAM model, oblivious compaction takes linear time~\cite{AsharovKLNPS23},
    but word operations cost $\Omega(\log n)$ gates, so this does not yield linear-size circuits.)
    All of these circuits are non-stable, necessarily so:
    stable compaction needs $\Omega(n \log n)$ selector gates if payloads are only routed and never computed on~\cite{LinSX19},
    and $\Omega(nm \log n)$ gates in general if the \emph{Li-Li network coding conjecture}~\cite{LiLi04} holds~\cite{AfshaniFKL19,asharov2021sorting}.
    To our knowledge it remains open whether the $\mathrm{poly}(\log^* n - \log^* m)$ factor can be removed for every $m$,
    and no unconditional superlinear lower bound is known for stable compaction, or for sorting keys of length $\Omega(\log n)$, with general gates.
    The circuits in this paper are much larger than these; they are designed for zero circuit entropy production, not size.
}%
\footnotetext{
\label{footnote:superconcentrator}
    A \emph{superconcentrator} is a network of $2$-input, $2$-output \emph{switches},
    with $n$ inputs and $n$ outputs,
    such that for every $k$ and every choice of $k$ inputs and $k$ outputs,
    some setting of the switches joins them by $k$ vertex-disjoint paths;
    \emph{self-routing} means the network computes that setting itself
    from flags presented at its inputs,
    rather than being handed it by an external routing algorithm.
    A switch differs from a comparator in that it is told whether to swap by a separate control wire
    and never inspects the payloads it carries,
    so compaction in this model splits into a control computation over the $n$ keys
    and a data path moving the $m$-bit payloads as opaque cargo.
} 
The input is $n$ key-value pairs $(b_i,V_{i,:})$ called \emph{records},
where $b_i \in \{0,1\}$ is a one-bit \emph{key}
and $V \in \{0,1\}^{n\times m}$ is the payload matrix,
whose $i$'th row $V_{i,:}$ is the $m$-bit \emph{payload} value in record $i$.
The goal is to sort the records by the key;
for example sorting $(1,a), (0,b), (1,c), (0,d)$ could result in sorted order
$(0,d), (0,b), (1,c), (1,a)$.
(There is no requirement that the sort be stable, 
i.e., preserve the original order between two records with equal keys.)
The requirement to carry the payload values along with the keys means that 
we must really track the order of the records as we sort them; 
simply counting the 1's as above is not sufficient.

\Cref{thm:compaction-zero-entropy} assumes the distribution on the input record list is \emph{exchangeable},
which means that any permutation of the $n$ record positions gives the same distribution over inputs.
This allows payloads to be arbitrarily correlated with their keys but prevents, for example, a bias towards keys that are almost already sorted or payloads that encode the initial indices of their record.


Our circuit is based on first computing the mapping $(b,V) \mapsto (b,s,W)$,
where $b \in \{0,1\}^n$ represents the unsorted 1-bit keys,
$V \in \{0,1\}^{n\times m}$ represents the associated $m$-bit payloads,
$s \in \{0,1\}^n$ is the sorted keys, and
$W \in \{0,1\}^{n\times m}$ is the sorted payload matrix.
To make the computation clean, we then erase the bits $b$ with zero $\EP$,
which requires determining the probability $p$ that a bit $b_i$ is 1 
and then routing it to an erasure gate optimal for that probability.
It turns out that there are $\Omega(n^2)$ many such probabilities,
which implies that this technique requires $\Omega(n^2)$ erasure gates.
It remains an open question whether one can erase $n$ bits, 
which in total have quadratically many possible distributions,
using $o(n^2)$ gates,
or to compute the desired mapping $(b,V) \mapsto (s,W)$ with zero circuit entropy production using a different technique
to obtain $o(n^2)$ total gates.
We note that for any register $r$,
both copying 
$(r,\vec{0}) \mapsto (r,r)$ 
and moving
$(r,\vec{0}) \mapsto (\vec{0}, r)$
can be done reversibly.

The next theorem turns an arbitrary sorting network into a compaction circuit with zero circuit entropy production.

\begin{theorem}
\label{thm:compaction-zero-entropy}
For all positive integers $n$ and $m$, 
for any sorting network $C$ with $t(n)$ gates and depth $d(n)$,
there is a circuit $\circuit$ cleanly solving the 1-bit key, $m$-bit payload value,
compaction problem on $n$ inputs, 
with  $\EP_\circuit(\dist) = 0$ for every exchangeable distribution $\dist$ over input record lists,
using $O(t(n)m + n^2)$ gates, $O(d(n)\log m + n \log n)$ depth, and $O(t(n) + nm)$ ancilla wires.
\end{theorem}

\begin{proof}
Denote the input $x = (b,V)$, where $b \in \{0,1\}^n$ represents the $n$ 1-bit keys,
and $V \in \{0,1\}^{n\times m}$ is the payload matrix whose row $V_{i,:}$ is the value associated with key $b_i$.
The correct output of the compaction problem on input $(b,V)$ is $C(b,V) = (s,W)$,
where $s = 0^{n-h} 1^{h}$ represents the compacted keys (assuming there are $h$ 1's in $b$),
and $W \in \{0,1\}^{n\times m}$ represents the associated values, sorted according to the keys.
(That is, the final $h$ rows $W_{n-h+1,:}, \dots, W_{n,:}$ are exactly those associated with 1 keys.)

Most of the task of doing compaction with zero circuit entropy production 
will use standard reversible computing primitives 
to compute the mapping $(b,V) \mapsto (b,s,W)$.
The main difficulty is erasing $b$ with zero $\EP$, 
once the rest of the output $(s,W)$ is correct and the other input $V$ is reversibly erased.
This last erasure of $b$ is the only part of the circuit that is not logically reversible and thus the only part which requires a heat function with non-zero entries.
In order to obtain zero circuit entropy production, we must pick these heat functions, per \cref{prop:ef-min-heat-fn}, in a way that exactly matches the induced distribution over inputs for each bit being erased.
Unfortunately, this induced input distribution depends on values other than the bit being erased (see \cref{sec:modularity} for why this causes positive entropy production).
Thus, it is necessary to determine the induced marginal probability distribution over the bit being erased conditioned on the values stored on all other wires to obtain zero circuit entropy production.

First we explain how to compute the mapping $(b,V) \mapsto (b,s,W)$ reversibly.
In the sorting network $C$, replace each comparator gate between wires $i$ and $j$
with a reversible comparator between records that also outputs (for reversibility)
whether a swap occurred:
\[
\begin{array}{rllll}
& ((b_i, V_{i,:}), & (b_j, V_{j,:}), & 0 & ) \\
\mapsto{} & ((b_i', V'_{i,:}), & (b_j', V'_{j,:}), & g & ),
\end{array}
\]
where the primed values are the identity if $b_i \le b_j$ and swapped otherwise,
with $g=1$ if a swap occurred and 0 otherwise.
To implement this record comparator using constant-size gates,
we reversibly distribute $g$ to $m$ clean ancilla wires using a balanced tree of copy gates of depth $O(\log m)$,
use those copies to swap all $m$ pairs of payload bits in parallel,
and then reverse the copy tree.
Thus each comparator of $C$ becomes $O(m)$ gates of depth $O(\log m)$.
All record comparators belonging to one layer of $C$ act on disjoint records,
so their copy trees and bit swaps can be performed in parallel.
The resulting reversible sorting computation has gate complexity $O(t(n)m)$
and depth $O(d(n)\log m)$.
Running this modified sorting network $C$ on the records $(b,V)$ results in the sorted output $(s,W)$ described above,
along with ``garbage'' data $G$ indicating which swaps occurred,
reversibly computing the mapping
\[
\begin{array}{rllll}
& (b, & V, & \vec{0} & ) \\
\mapsto{} & (s, & W, & G & ).
\end{array}
\]
We now reversibly copy $s$:
\[
\begin{array}{rlllll}
& (s, & W, & G, & \vec{0} & ) \\
\mapsto{} & (s, & W, & G, & s & ).
\end{array}
\]
Now, ``uncompute'' $C$ via the mirror-image reverse of the first sorting step,
but using swap gates that \emph{only} swap the keys, leaving the sorted payloads in place:
\[
\begin{array}{rlllll}
& (s, & W, & G, & s & ) \\
\mapsto{} & (b, & W, & \vec{0}, & s & ).
\end{array}
\]
This completes the mapping
$
(b,V)
\mapsto
(b,s,W),
$
using $O(t(n)m)$ gates and $O(d(n)\log m)$ depth.

At this point we \emph{almost} have the clean output $y = (s,W)$,
but it remains to erase the unsorted key bits $b$.

Write $b_{\le i}$ to denote the $i$'th prefix $b_{1 \twodots i}$ of $b$
(similarly 
$b_{< i}=b_{1\twodots i-1}$,
$b_{\ge i} = b_{i \twodots n}$, and 
$b_{>i} = b_{i+1\twodots n}$).
Let $h = |b|_1 = |s|_1$.
We erase the bits one at a time starting with the first,
successively producing 
$
b
\mapsto
0 b_{\ge 2}
\mapsto
0 0 b_{\ge 3}
\mapsto
\ldots
\mapsto
0^{n-2} b_{n-1} b_n
\mapsto
0^{n-1} b_n
\mapsto
0^{n}.
$
Suppose we currently have
$0^{i-1} b_{\ge i}$
and now need to erase $b_{i}$;
call this \emph{stage} $i$.
The correct distribution to assume when erasing $b_i$ when we apply \cref{prop:ef-min-heat-fn} is such that the probability of a one is 
$|b_{\le i}|_1 / i$;
we justify this at the end of the proof.
The denominator $i$ is a fixed constant in stage $i$ of this circuit,
so the key is to make the numerator
$z_i = |b_{\le i}|_1$ reversibly available in a binary register $Z$,
which will then be used to route $b_i$ to a ``controlled erasure'' gate 
with an optimal heat function for this probability.\footnote{
    A controlled erasure gate takes two input bits: $c$ and $b$.
    If the control bit $c=0$, then it leaves $b$ alone (and releases 0 heat),
    otherwise it erases $b$.
    We need these because there will be a ``bank'' of several that might be used,
    but only one will actually be used to erase $b_i$, i.e., will have its $c=1$;
    the other gates will have $c=0$.
}
We explain below how to maintain $Z$ incrementally as we iterate through the bits of $b$.

An example of the erasure subcircuit for stage $i$ is shown in \Cref{fig:compaction-subcircuit}.
There are $i + 1$ pairs of clean wires, one pair for each possible numerator
$j\in\{0,\ldots,i\}$.
Set a clean wire to 1 with a NOT gate and,
using a reversible $(i+1)$-way demultiplexer,\footnote{
    A demultiplexer takes as input a value $v$ (one or more bits),
    and an address $a$ interpreted as a binary integer,
    and routes $v$ to one of several output lines based on the value of $a$.
}
route the pair $(1,b_i)$ (the first bit is a control bit $c$ for the controlled erasure gate downstream)
to output wires addressed by the binary value $z_i$ in register $Z$.
Thus the selected pair contains $(1,b_i)$ and every other pair contains $(0,0)$.
On pair $j$, apply a controlled erasure gate to each such pair for $j \in \{0,\ldots,i\}$,
using a heat function optimal for probability $j/i$.
Consequently only the selected gate performs an erasure, using the heat function optimal for
the probability $z_i/i$; all other gates act reversibly with zero heat since their control bits are 0.
After these gates act in parallel, reverse the demultiplexer (a.k.a., multiplexing) using the same address $z_i$.
It returns the pair $(1,0)$, after which a NOT gate restores the control wire to 0,
while the original $b_i$ wire and all routing wires are also 0.

\begin{figure}[ht]
    \centering
    \resizebox{\textwidth}{!}{\sffamily\begin{tikzpicture}[>={Stealth[length=1.6mm]}, every node/.style={font=\small}, scale=0.57]
    \foreach \y in {1.2, 0.7}{
      \node[font=\small, gray] at (0,\y) {0};
      \draw[wire, gray!60] (0.3,\y) -- (0.95,\y);
      \node[font=\small, gray] at (1.6,\y) {$\cdots$};
    }
    \node[font=\small, keyblue] at (0,0.2) {1};
    \node[font=\small] at (0,-0.3) {1};
    \draw[wire] (0.3,-0.3) -- (0.95,-0.3);
    \node[font=\small] at (1.6,-0.3) {$\cdots$};
    \node[font=\scriptsize, keyblue, above=1pt, inner sep=1pt] at (2.45,0.2) {$b_3$};
    \foreach \y/\v in {-2.70/0, -3.10/1}{\node[font=\small] at (0,\y) {\v};}
    \draw[decorate, decoration={brace, amplitude=3pt}]
      (-0.42,-0.02) -- (-0.42,1.42)
      node[midway, left=3pt, align=right, font=\scriptsize] {$b_{\le 3}$\\\textcolor{gray}{($b_{<3}$ erased)}};
    \draw[decorate, decoration={brace, amplitude=3pt}]
      (-0.42,-0.52) -- (-0.42,-0.08) node[midway, left=3pt, font=\scriptsize] {$b_{>3}$};
    \draw[decorate, decoration={brace, amplitude=3pt}]
      (-0.42,-3.32) -- (-0.42,-2.48) node[midway, left=3pt, font=\scriptsize] {$Z$};
    \node[gate, minimum width=0.95cm, minimum height=2.15cm, font=\footnotesize] (inc) at (3.6,-1.45) {$Z + b_3$};
    \draw[wire, keyblue] (0.3,0.2) -- (inc.west |- 0,0.2);
    \foreach \y in {-2.70, -3.10}{\draw[wire] (0.3,\y) -- (inc.west |- 0,\y);}
    \draw[wire, keyblue] (inc.east |- 0,0.2) -- (4.9,0.2);
    \node[wirelabel, text=keyblue] at (4.9,0.2) {1};
    \foreach \y/\v in {-2.70/1, -3.10/0}{
      \draw[wire] (inc.east |- 0,\y) -- (4.9,\y);
      \node[wirelabel] at (4.9,\y) {\v};
    }
    \foreach \k/\y/\l in {0/1.8/{$p = 0/3$}, 1/0.6/{$p = 1/3$},
                          2/-0.6/{$p = 2/3$}, 3/-1.8/{$p = 3/3$}}{
      \node[cerase, minimum width=1.55cm, minimum height=0.4cm, font=\footnotesize]
        (E\k) at (9.9,\y) {\l};
    }
    \foreach \y in {-2.70, -3.10}{
      \draw[wire] (5.05,\y) -- (14.3,\y);
    }
    \draw[thick, fill=keyblue!12] (6.6,-0.6) -- (6.6,0.6) -- (7.8,2.2) -- (7.8,-2.2) -- cycle;
    \draw[wire] (7.0,-1.13) -- (7.0,-2.70);
    \fill (7.0,-2.70) circle (2.4pt);
    \draw[wire] (7.4,-1.67) -- (7.4,-3.10);
    \fill (7.4,-3.10) circle (2.4pt);
    \node[font=\scriptsize] at (7.2,-3.53) {address $z_3 = 2$};
    \foreach \y/\v in {-2.70/1, -3.10/0}{\node[wirelabel] at (8.1,\y) {\v};}
    \draw[wire, keyblue] (5.05,0.2) -- (6.6,0.2);
    \draw[wire, keyblue, dashed] (6.6,0.2) to[out=0, in=180] (7.8,-0.4);
    \foreach \k/\y in {0/1.8, 1/0.6, 3/-1.8}{
      \draw[wire] (7.8,{\y+0.2}) -- ($(E\k.west)+(0,0.2)$);
      \node[wirelabel] at (8.1,{\y+0.2}) {0};
    }
    \draw[wire, keyblue] (7.8,-0.4) -- ($(E2.west)+(0,0.2)$);
    \node[wirelabel, text=keyblue] at (8.1,-0.4) {1};
    \draw[wire, garbagepurple] (5.55,-0.2) -- (6.6,-0.2);
    \node[wirelabel, text=garbagepurple] at (5.45,-0.2) {0};
    \node[tightgate] at (6.05,-0.2) {$\neg$};
    \draw[wire, garbagepurple, dashed] (6.6,-0.2) to[out=0, in=180] (7.8,-0.8);
    \foreach \k/\y/\c in {0/1.8/0, 1/0.6/0, 2/-0.6/1, 3/-1.8/0}{
      \draw[wire, garbagepurple] (7.8,{\y-0.2}) -- ($(E\k.west)+(0,-0.2)$);
      \node[wirelabel, text=garbagepurple] at (8.1,{\y-0.2}) {\c};
      \draw[wire, garbagepurple] ($(E\k.east)+(0,-0.2)$) -- (14.3,{\y-0.2});
      \node[wirelabel, text=garbagepurple] at (11.75,{\y-0.2}) {\c};
    }
    \node[font=\tiny, garbagepurple, above] at (13.0,1.64) {control bits};
    \node[font=\tiny, garbagepurple, right, align=left] at (12.2,-1.25) {needs\\cleaning};
    \draw[->, garbagepurple] (12.33,-1.2) -- (11.92,-0.9);
    \foreach \k/\y in {0/1.8, 1/0.6, 3/-1.8}{
      \draw[wire] ($(E\k.east)+(0,0.2)$) -- (11.6,{\y+0.2});
      \node[wirelabel] at (11.75,{\y+0.2}) {0};
    }
    \draw[wire] ($(E2.east)+(0,0.2)$) -- (11.6,-0.4);
    \node[wirelabel, text=heatred] at (11.75,-0.4) {0};
    \node[font=\scriptsize, heatred] at (13.08,-0.4) {$b_3$ erased};
    \foreach \y/\v in {-2.70/1, -3.10/0}{
      \draw[wire] (14.3,\y) -- (15.55,\y);
      \node[wirelabel] at (15.55,\y) {\v};
      \node[font=\small, right] at (15.57,\y) {$\cdots$};
    }
    \draw[decorate, decoration={brace, amplitude=3pt}] (16.45,-2.48) -- (16.45,-3.32)
      node[midway, right=2pt, align=left, font=\scriptsize] {initial $Z$ for next\\stage erasing $b_4$};
    \draw[thick, fill=keyblue!12] (14.3,-2.2) -- (14.3,2.2) -- (15.5,0.6) -- (15.5,-0.6) -- cycle;
    \draw[wire] (14.7,-1.67) -- (14.7,-2.70);
    \fill (14.7,-2.70) circle (2.4pt);
    \draw[wire] (15.1,-1.13) -- (15.1,-3.10);
    \fill (15.1,-3.10) circle (2.4pt);
    \draw[wire, garbagepurple, dashed] (14.3,-0.8) to[out=0, in=180] (15.5,0.0);
    \draw[wire, garbagepurple] (15.5,0.0) -- (16.25,0.0);
    \draw[wire] (16.25,0.0) -- (17.0,0.0);
    \node[tightgate] at (16.25,0.0) {$\neg$};
    \node[font=\small, right] at (17.05,0.0) {0};
    \node[font=\scriptsize, anchor=south] at (7.2,2.25) {demultiplexer};
    \node[font=\scriptsize, anchor=south] at (14.9,2.25) {multiplexer};
    \node[font=\tiny] at (7.2,0.45) {demux};
    \node[font=\tiny] at (14.9,0.45) {mux};
  \end{tikzpicture}%
}
    \caption{
        The subcircuit for erasing bit $b_3$ when $b = 1011$,
        so $z_3 = |b_{\le 3}|_1 = 2$ and the gate must assume $\Pr[b_3 = 1] = 2/3$.
        Reversible demultiplexer drawn without extra clean input wires needed to supply $0$'s to unaddressed output wires, and
        similarly for unshown output wires of multiplexer that simply pass through unaddressed inputs.
    }
    \label{fig:compaction-subcircuit}
\end{figure}
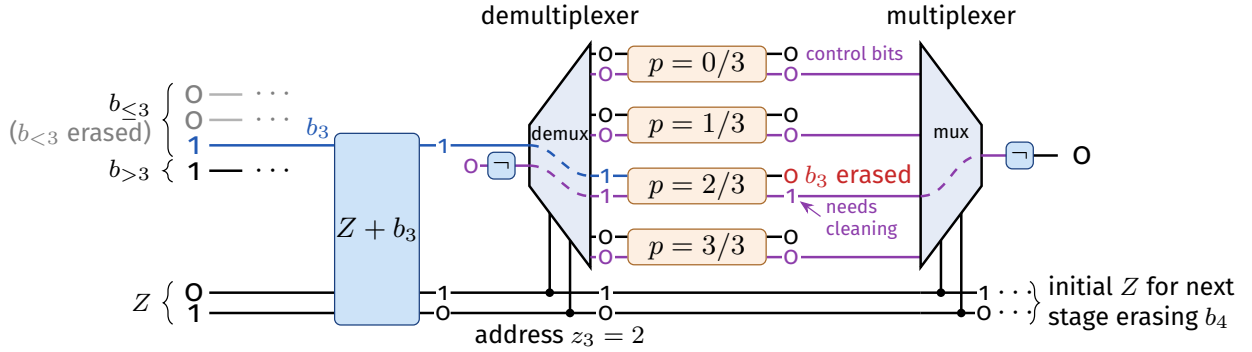

A binary-tree demultiplexer uses $O(i) = O(n)$ constant-size gates and $O(\log i) = O(\log n)$ depth.
Together with the $i+1$ controlled erasure gates, stage $i$ therefore uses
$O(n)$ gates and $O(\log n)$ depth for routing and erasure.

Computing $z_i$ from scratch with a straightforward sequential implementation would require $O(n)$ depth,
and since this must be done sequentially for each bit erasure,
the total depth would be $O(n^2)$ if computed na\"{i}vely.
Instead, we observe that $z_i$ can be updated incrementally using only $O(\log n)$ depth as we erase bits.
Maintain a register $Z$, initially 0.
Just before stage $i$, maintain the invariant that $Z=|b_{<i}|_1$, the number of 1's among
the already-erased prefix bits.
At stage $i$, do the reversible increment
$
(b_i, Z)
\mapsto
(b_i, Z+b_i)
$
to update $Z$ to represent $|b_{\le i}|_1 = h-|b_{>i}|_1 = z_i$.\footnote{
    \label{fn:increment-before-demux}
    The increment must precede the demultiplexing that selects the erasure gate.
    If the pre-increment value $Z = z_i - b_i$ were used as the address instead,
    then $(s, W)$, $b_{>i}$, and $Z$ would together determine $b_i = z_i - Z$,
    since $z_i = |s|_1 - |b_{>i}|_1$ is determined by the first two.
    Erasing $b_i$ would then destroy the mutual information between $b_i$ and the rest of the circuit,
    which by \Cref{lem:mutual-info-positive-ep} produces entropy no matter which heat function the erasure gate uses.
    After the increment, $Z = z_i$ carries no information about $b_i$ beyond what $(s, W)$ and $b_{>i}$ already determine.
}
This requires only $O(\log n)$ depth per bit erasure
(using a standard ripple-carry increment),\footnote{
    We could actually do better than $O(\log n)$ depth to increment $Z$,
    using $O(\log \log n)$ depth by using a parallel prefix circuit~\cite{ladner1980parallel}.
    However, we forsake such optimizations, 
    since we are already using $\Omega(\log n)$ depth 
    for the demultiplexer used to route the bits to erase,
    so there is no asymptotic benefit in optimizing the increment depth.
    In fact, given our setup, there is no way to demultiplex in $o(\log n)$ depth,
    since we need a single value $b_i$ to reach one of $\Omega(n)$ different wires,
    but due to $O(1)$ size gates,
    each layer only enables it to reach a constant multiple more wires than the previous.
    This does not rule out a more sophisticated approach that can achieve $O(1)$ amortized depth per bit erasure,
    using some sort of ``pipelining'' for instance,
    or perhaps some way to erase $\omega(1)$ bits in parallel.
}
so $O(n \log n)$ total depth for all $n$ bit erasures.
After the final erasure, $Z = h = |s|_1$.
To restore this work register to 0, as required for clean output,
apply for each $i \in \{1, \ldots, n\}$ in turn the reversible decrement
$
(Z, s_i) \mapsto (Z - s_i, s_i),
$
the inverse of the increment above, with the retained sorted key bit $s_i$ as the control.
Since $Z$ starts at $|s|_1$, $Z$ ends at 0.
This uses $n$ decrements of $O(\log n)$ gates and depth per bit,
i.e., $O(n\log n)$ gates and depth total.

The erasure portion uses
$\sum_{i=1}^n (i+1)=\Theta(n^2)$ controlled erasure gates,
one for each possible numerator $z_i$ at each stage $i$;
the forward and reverse demultiplexers use $O(n^2)$ additional reversible gates in total.
Together with the reversible sorting computation, the total gate complexity is
$O(t(n)m)$ (for the sorting network)
plus
$O(n^2)$,
i.e., $O(t(n)m + n^2)$.
The depth is $O(d(n)\log m + n \log n)$ since we have $n$ layers of single-bit erasure,
one for each $b_i$,
and each layer requires $O(\log n)$ depth as justified above
(plus the $O(d(n)\log m)$ depth of the reversible sorting computation.)
The circuit uses $O(t(n) + nm)$ wires beyond the $n(m+1)$ input wires:
$n(m+1)$ for the register that receives $(s, W)$;
$t(n)$ for the garbage bits $G$;
$O(nm)$ for the copy trees of the comparators in one layer of $C$,
which are returned to 0 after each layer and reused by the next;
and $O(n)$ for the register $Z$, the demultiplexers, and the pairs of wires feeding the controlled erasure gates,
which are likewise returned to 0 after each stage and reused by the next.

Now we show that this gives us zero circuit entropy production.
Overall, we are mapping $(b,V)\sim \dist$ to $C(b,V) = (s,W)$.

Since $C$ is a deterministic function, we have that $\entropy((b,V)) - \entropy(C((b,V))) = \entropy((b,V) \; | \; (s,W)) = \entropy(b \; | \; (s,W))$.
Since $\dist$ is exchangeable, $W$ is irrelevant to the conditional probability distribution,\footnote{The distribution over payloads must be independent of their input index and $b$ is a permutation of $s$.} so this equals $\entropy(b \; | \; s) = \entropy(b \; | \; \abs{b}_1) = \sum_{k \leq n} \Pr\left[\abs{b}_1 = k\right] \cdot \entropy(b \; | \; \abs{b}_1 = k)$.
Since $\dist$ is exchangeable, each $b$ with a fixed value of $\abs{b}_1$ is equally likely, so $\entropy(b \; | \; \abs{b}_1 = k) = \log \binom{n}{k} = \log\frac{n!}{k!(n-k)!}$ and the total drop in entropy is:
\begin{equation}\label{eqn:compaction-entropy-drop}
    \sum_{k \leq n} \Pr[\abs{b}_1 = k] \log \frac{n!}{k!(n-k)!}.
\end{equation}

We will now show that our choices of heat function $\heat_i$ for erasing each bit of any fixed $b$ with $\abs{b}_1 = k$ result in the circuit releasing exactly this much heat in expectation.
We now fix $k$ and any $b$ such that $\abs{b}_1 = k$.
As a reminder, we are picking the optimal heat function per \cref{prop:ef-min-heat-fn} for erasing bit $b_i$ based on the probability that $b_i$ is $1$ is $\abs{b_{\leq i}}_1 / i$.
This gives us the heat function
\begin{displaymath}
    \heat_i(b_i, \abs{b_{\leq i}}_0, \abs{b_{\leq i}}_1) = \begin{cases}
    \log(i/\abs{b_{\leq i}}_0) & b_i = 0\\
    \log(i/\abs{b_{\leq i}}_1) & b_i = 1
\end{cases}
\end{displaymath}
for erasing the $i$'th bit of $b$.
Let $I = i_1, \ldots, i_{n-k}$ denote the indices of $b$ where $b_{i_{\ell}} = 0$ and $J = j_1, \ldots, j_{k}$ denote the indices of $b$ where $b_{j_{\ell}} = 1$.
Then the expected heat released on input $b$ is exactly
\begin{displaymath}
    \sum_{i_\ell \in I} \log(i_\ell / \ell) + \sum_{j_\ell \in J} \log(j_\ell / \ell)
    = \sum_{\ell \in [n]} \log(\ell) + \sum_{i \in [n-k]} \log(1/i) + \sum_{j \in [k]} \log(1/j)
    = \log \frac{n!}{k!(n-k)!}.
\end{displaymath}
Due to the linearity of expectation, the total expected heat released during compaction is exactly
\begin{equation}\label{eqn:compaction-heat-released}
    \sum_{k \leq n} \Pr[\abs{b}_1 = k] \log \frac{n!}{k!(n-k)!}.
\end{equation}
Since the drop in entropy~\eqref{eqn:compaction-entropy-drop} and the expected heat release~\eqref{eqn:compaction-heat-released} are identical, we can conclude that $\EP_\circuit = 0$.
\end{proof}
\begin{remark}
The analysis in \cref{thm:compaction-zero-entropy} has a natural extension to compaction on $k$-bit keys by instead applying \cref{prop:ef-min-heat-fn} to pick heat functions where, for every $d \in \{0,1\}^k$, the chosen distribution over inputs has $\Pr[b_i = d] = \abs{b_{\leq i}}_d/i$ where $\abs{x}_d$ denotes the count of key $d$ in $x$.
\end{remark}
\subsection{A clean sorting circuit with zero entropy production}
\label{sec:sorting-zero-ep}
Our method of zero circuit entropy production erasure in \cref{sec:clean-compaction-zero-ep} can be extended to give a circuit with zero entropy production for sorting.
As we are now doing general sorting, we can always assume that payloads are just low order bits of the keys and therefore ignore them altogether.
Thus the task is now, given an input $x = x_1, \ldots, x_n$ with each $x_i \in \{0,1\}^m$, output $\sort(x) = y_1, \ldots, y_n$ where the $y_i$ are a permutation of the $x_i$ such that, for every $i$, $y_i \leq y_{i+1}$.
For simplicity, we will assume that our larger keys can be compared with a single comparator gate, which is implemented using multiple constant-arity gates.

\begin{theorem}\label{thm:sorting-zero-entropy}
    Let $\dist$ be any exchangeable distribution over $(\{0,1\}^m)^n$.
    There exists a polynomial sized (in terms of $n$ and $m$) circuit that sorts $n$ inputs that are each $m$ bits with zero circuit entropy production.
\end{theorem}
\begin{proof}
    We will use the same circuit as in \cref{sec:clean-compaction-zero-ep} to map input $x$ to the state $(x,\sort(x))$ using only logically reversible gates.
    Now, the task is to erase the copy of $x$ with zero circuit entropy production.
    For $b = b_1, \ldots, b_\ell$, and $b_i, v \in \{0,1\}^k$, we use $\abs{b}_v$ to denote the number of $b_i$ such that $b_i = v$.

    We will erase the highest order bit of every $x_i$ before moving on to the next highest order bit and so on.
    In other words, the order in which we erase bits will be $x_1^1, x_2^1, \ldots, x_n^1, x_1^2, x_2^2, \ldots, x_n^m$.
    Let $x_i^j$ denote the $j$'th bit of $x_i$.
    Analogous to before, we now use $x_{\leq i}^{\geq j}$ to denote the suffixes starting on the $j$'th bit of $x_{i'}$ for all $i' \leq i$.
    When erasing $x_i^j$, we will use the optimal heat function, per \cref{prop:ef-min-heat-fn}, for the distribution with:
    \begin{displaymath}
    \Pr[x_i^j = 1] = \frac{\abs{x_{\leq i}^{\geq j}}_{1\cdot x_i^{\geq j+1}}}{\abs{x_{\leq i}^{\geq j+1}}_{ x_i^{\geq j+1}}},
    \end{displaymath}
    where $1 \cdot x_i^{\geq j}$ denotes the suffix of $x_i$ beginning with index $j$ with a $1$ appended to the front.
    In other words, the probability of a $1$ is exactly the probability one would get a $1$ by sampling uniformly at random the $j$'th bit of any $x_{i'}$ with $i'\leq i$  that shares its suffix after index $j$ with $x_i$.
    While it would be unfeasible to track these counts for every possible suffix, this probability can be computed on the fly for each erasure by paying a polynomial cost in the depth of the circuit and the number of gates.
    This follows from the fact that these counts can be computed (and uncomputed) in a logically reversible fashion even after the actual value of $x_i^j$ has been erased.\footnote{The counts of $0$'s versus $1$'s for erased prefixes can be deduced from the sorted output and the non erased part of the input.
    For example, $\abs{x_{\leq i}^{\geq j}}_{0 \cdot x_i^{\geq j+1}} = \abs{\sort(x)^{\geq j}}_{0 \cdot x_i^{\geq j+1}} - \abs{x_{>i}^{\geq j}}_{0 \cdot x_i^{\geq j+1}}$ and the terms on the right can be computed after $x_i^j$ has been erased.
    Control gates can be used to ensure that we only count the number of $0$'s and $1$'s for inputs with the desired suffix.}
    We will use the same reversible demultiplexer / multiplexer gadget as in the compaction circuit with controlled erasure gates to selectively pick the optimal heat function to use for each erasure.

    Now we argue that this approach gives us a circuit with zero circuit entropy production.
    First we consider the drop in entropy from the mapping $x \mapsto \sort(x)$.
    Let $C(x) \in \Z_{\geq 0}^{\{0,1\}^m}$ be the vector of counts for the number of times each possible $b \in \{0,1\}^m$ appears in the input $x$.
    Then we have that $\entropy(\dist) - \entropy(\sort(\dist)) = \entropy(\dist \; | \; \sort(\dist)) = \entropy(\dist \; | \; C(\dist))$ and thus when $x \sim \dist$:
    \begin{align}
        \entropy(\dist) - \entropy(\sort(\dist)) &=  \sum_{c \in \Z_{\geq 0}^{\{0,1\}^m}} \Pr_{x \sim \dist}[C(x) = c] \entropy(x \; | \; C(x) = c) \nonumber \\
        &= \sum_{c \in \Z_{\geq 0}^{\{0,1\}^m}} \Pr_{x \sim \dist}[C(x) = c] \log \frac{n!}{\prod_{b \in \{0,1\}^m} c_b !} \label{eqn:sort-entropy-drop}
    \end{align}
    Note that this term only has support on $c$ where the $1$-norm (i.e. Manhattan norm) of $c$ is exactly $n$.
    
    As before, we will now show that our choices of heat functions $\heat_{i,j}$ for erasing each bit release $\log(n!/\prod_b c_b!)$ expected heat on every input with count vector $c$.
    Our heat function for erasing the $j$'th bit of the $i$'th input will take the form:
    \begin{align*}
        \heat_{i,j}\left(x_i^j, \abs{x_{\leq i}^{\geq j+1}}_{x_i^{\geq j+1}}, \abs{x_{\leq i}^{\geq j}}_{0 \cdot x_i^{\geq j+1}}, \abs{x_{\leq i}^{\geq j}}_{1 \cdot x_i^{\geq j+1}}\right) &= \begin{cases} \log\left(\frac{\abs{x_{\leq i}^{\geq j+1}}_{x_i^{\geq j+1}}}{\abs{x_{\leq i}^{\geq j}}_{0 \cdot x_i^{\geq j+1}}}\right) & x_i^j = 0 \vspace{5pt}\\
        \log\left(\frac{\abs{x_{\leq i}^{\geq j+1}}_{x_i^{\geq j+1}}}{\abs{x_{\leq i}^{\geq j}}_{1 \cdot x_i^{\geq j+1}}}\right) & x_i^j = 1
        \end{cases} \\
        &= \log\left(\frac{\abs{x^{\geq j+1}_{\leq i}}_{x_i^{\geq j+1}}}{\abs{x^{\geq j}_{\leq i}}_{x_i^{\geq j}}}\right)
    \end{align*}
    Let's fix any $x$ and let $c = C(x)$.
    Then the total expected heat released on input $x$ is exactly
    \begin{align*}
        \sum_{i \in [n], j \in [m]} &\heat_{i,j}\left(x_i^j, \abs{x_{\leq i}^{\geq j+1}}_{x_i^{\geq j+1}}, \abs{x_{\leq i}^{\geq j}}_{0 \cdot x_i^{\geq j+1}}, \abs{x_{\leq i}^{\geq j}}_{1 \cdot x_i^{\geq j+1}}\right)  = \sum_{i \in [n], j \in [m]} \log\left(\frac{\abs{x^{\geq j+1}_{\leq i}}_{x_i^{\geq j+1}}}{\abs{x^{\geq j}_{\leq i}}_{x_i^{\geq j}}}\right)\\
        &=\sum_{i \in [n]} \sum_{j \in [m]} \log \left(\abs{x^{\geq j+1}_{\leq i}}_{x_i^{\geq j+1}}\right) - \log \left(\abs{x^{\geq j}_{\leq i}}_{x_i^{\geq j}}\right)\\
        &= \sum_{i \in [n]} \log\left(\abs{x^{\geq m + 1}_{\leq i}}_{x_i^{\geq m+1}}\right) - \log \left(\abs{x^{\geq 1}_{\leq i}}_{x_i^{\geq 1}}\right) \\
        &= \sum_{i \in [n]} \log(i) - \log(C_{x_i}(x_{\leq i})) \qquad \textrm{since $x_i^{\geq m+1}$ is the empty string}\\
        &=\log \frac{n!}{\prod_{b \in \{0,1\}^m} c_b!}.
    \end{align*}
    Due to the linearity of expectation, the total expected heat released during sorting is exactly:
\begin{equation}\label{eqn:sort-heat-released}
    \sum_{c \in \Z_{\geq 0}^{\{0,1\}^m}} \Pr_{x \sim \dist}[C(x) = c] \log \frac{n!}{\prod_{b \in \{0,1\}^m} c_b !}.
\end{equation}
    Since the drop in entropy \eqref{eqn:sort-entropy-drop} exactly matches the expected heat released \eqref{eqn:sort-heat-released} by our circuit, we can conclude that the circuit entropy production is zero.
\end{proof}

\section{Conclusion}
\label{sec:conclusion}

In this paper we posit that the concepts of entropy production and heat release provide a principled foundation for
understanding inherent energy costs of computations. 
We introduce a thermodynamic logic circuit model in which gates  (e.g., Boolean logic gates or comparator gates) are augmented with heat functions, which
determine a circuit's entropy flow and entropy production with respect to a given input distribution. We show that in general, determining optimal heat functions
is computationally intractable. In contrast, we show that for specific problems (i.e. sorting) our model can be used both to analyze the thermodynamic costs of existing circuits and to design new circuits that are thermodynamically efficient.

There are many directions for future work. Our analysis of Batcher's odd-even mergesort on bits uses
a set of gate heat functions that scales with the circuit size; what are the penalties for using only a constant number of distinct heat functions for this problem?
How can we analyze heat dissipation of circuits in other problem domains, such as graph algorithms? What are good general design principles for energy-efficient computation?

\Cref{thm:compaction-zero-entropy} shows a zero $\EP$ circuit for bit-compaction,
the problem of sorting 1-bit keys with $m$-bit payloads,
using $O(n \log n)$ depth (as long as $\log m$ is $O(n)$).
Can we design a zero $\EP$ compaction circuit with polylog($n$) depth?
This rules out erasing key bits one at a time as we do, 
which requires $\Omega(n)$ depth.\footnote{
    It is even worse for our circuit, $\Theta(n \log n)$,
    due to requiring $\Theta(\log n)$ depth per bit to route it through a $\Theta(\log(n))$-layer demultiplexer to the erasure gate with appropriate heat function.
}
Is it possible to show a lower bound of $\Omega(n)$ depth for the clean compaction problem, arguing that zero $\EP$ erasure with constant arity gates must do so sequentially?

Another limitation of our technique to do compaction in \Cref{thm:compaction-zero-entropy} is that we require $\Omega(n^2)$ erasure gates.
This is necessary given our specific erasure technique:
we get the circuit to a point where it has the intended output 
(compacted key-payload records),
as well as the uncompacted keys $b \in \{0,1\}^n$,
and we erase each key bit one at a time.
There are $\Omega(n^2)$ conditional probabilities that are possible for each bit to be a 1 across all stages during this process, necessitating $\Omega(n^2)$ erasure gates.
However, this does not rule out a more sophisticated way to erase the bits in which we do not encounter so many different conditional probabilities.
Perhaps some efficient reversible circuit computation could be done on $b$ to transform it into a different bit string $b'$ of length $O(n)$,
such that while erasing the bits of $b'$ in order,
each bit would only have one of,
say, $O(\log n)$ different (possibly position-dependent) probabilities.
In this case $O(n \log n)$ total erasure gates would suffice to erase all bits.

Suppose the number of erasure gates could be reduced to $O(n \mathrm{polylog}(n))$,
or if $m \gg n$ so that the $O(n^2)$ erasure gates are negligible.
(To simplify discussion below assume $m = O(1)$, so that swapping $m$-bit payloads takes $O(1)$ gates in our model.)
Then it becomes more important to reduce the size of the compaction/sorting circuit implementing the
underlying routing logic in \Cref{thm:compaction-zero-entropy,thm:sorting-zero-entropy}.
For instance, \Cref{thm:compaction-zero-entropy} can very likely be improved to use
the $O(n \mathrm{polylog}^*(n))$-gate compaction circuit of Asharov, Lin, and Shi~\cite{asharov2021sorting},
instead of a sorting network that requires $\Omega(n \log n)$ gates by the 0/1 principle.\footnote{
    Our main reason for using sorting networks is that it is more straightforward to transform them to be reversible in a black-box way, 
    whereas more sorting circuits allowing more general gates than comparators require more complex hypotheses about how the sorting circuit treats its keys and payloads.
    For example it is important for our technique that routing decisions are made based only on keys,
    never payloads, 
    which is true of the networks of Asharov, Lin, and Shi~\cite{asharov2021sorting} 
    and Lin and Shi~\cite{lin2022optimal}.
}

Rather than insisting on $\EP = 0$,
one can ask for the minimum $\EP$ achievable by compaction circuits of near-linear size or polylogarithmic depth.
We know of no lower bound of the form ``depth $d$ forces $\EP \ge g(n,d)$'' for any $g>0$;
such a bound would be the first general quantitative trade-off between a standard circuit resource and entropy production.

\Cref{thm:compaction-zero-entropy} uses a single circuit, with finitely many heat functions,
that has zero $\EP$ for every exchangeable input distribution simultaneously,
because the conditional probability that each erased bit is 1 depends only on exchangeability.
For which other classes of distributions is this possible?
If, for instance, the keys are independent but position-dependent,
the conditional probabilities depend on the bias of every position,
and it is unclear whether finitely many heat functions suffice.
Conversely, what is the $\EP$ of our circuit when the input is only approximately exchangeable,
say as a function of a divergence between the actual distribution and the nearest exchangeable one?

Analogous questions about the gate and depth complexity for zero circuit entropy production sorting are natural.
In general, a more detailed study of gate-depth-entropy tradeoffs for sorting is worthy of further investigation.

Our results concern expected $\EP$ and expected heat. What are the variance, and the worst case over inputs, of the heat released by our circuits,
and is there a trade-off between the fluctuations of heat and circuit size or depth,
in the spirit of the thermodynamic uncertainty relations of stochastic thermodynamics~\cite{HorowitzGingrich2020}?
Can we obtain concentration guarantees on the distribution of expected heat and entropy production?

\Cref{sec:batcher-bits} shows that Batcher's network has $\EP = \Omega(n)$ on uniformly random bits
and $\Omega(n \log^2 n)$ on uniformly random distinct inputs,
and \Cref{sec:better-than-batchers} gives a network releasing only $O(n \log n)$ expected heat,
but, as noted in \Cref{sec:batcher-bits},
it is open whether every sorting network has $\EP = \Omega(n)$, or even $\EP > 0$, on uniformly random bits.
(Every sorting network must release $\Omega(n)$ heat, since that is the Landauer cost of sorting bits;
zero $\EP$ means releasing exactly that much.)
A positive lower bound would show that comparators alone cannot achieve zero $\EP$,
so that the more general gates of \Cref{thm:compaction-zero-entropy} are necessary.\footnote{
    Ancilla wires alone would not help a comparator network on bits:
    a comparator between a data wire and a wire holding a known constant is a data-independent wire crossing
    (with constant 0 it maps $(x, 0) \mapsto (0, x)$, and with constant 1 it maps $(x, 1) \mapsto (x, 1)$),
    which is a bijection on its reachable inputs,
    so it can be assigned zero heat and removed by renaming wires,
    without changing the outputs or increasing $\EP$.
}
In the other direction, \Cref{thm:compaction-zero-entropy} uses $O(t(n) + nm)$ ancilla wires,
where $t(n)$ is the size of the underlying sorting network.
What is the minimum number of ancilla wires in a zero $\EP$ sorting or compaction circuit?
In particular, can zero $\EP$ be achieved in place, with no ancilla wires,
using $O(1)$-size gates more general than comparators?


Our thermodynamic circuit model is deterministic; we are interested in adapting it to circuits with probabilistic gates, where the output of a gate on a given input is specified as a distribution over the possibilities. We also would like to understand
heat dissipation in other models of computation, such as random access machines. Can we show  trade-offs between time, space, and entropy production? Might it be possible to show average-case hardness results for heat dissipation?

\section{Acknowledgments}
Niels Kornerup's contribution to work was supported by the Laboratory Directed Research and Development program (project
240650) at Sandia National Laboratories, a multimission laboratory managed and operated by National
Technology and Engineering Solutions of Sandia LLC, a wholly owned subsidiary of Honeywell International
Inc. for the U.S. Department of Energy’s National Nuclear Security Administration under contract DE-
NA0003525.
This paper describes objective technical results and analysis. Any subjective views or opinions that might be expressed in the paper do not necessarily represent the views of the U.S. Department of Energy or the United States Government.
The publisher acknowledges that the U.S. Government retains a non-exclusive, paid-up, irrevocable, world-wide license to publish or reproduce the published form of this written work or allow others to do so, for U.S. Government purposes. The DOE will provide public access to results of federally sponsored research in accordance with the DOE Public Access Plan. 

David Wolpert thanks the Santa Fe Institute for their support.
\subsection{Discussion of substantive LLM use}

In \cref{sec:approx-heat-funcs}, we used ChatGPT 5.6 to help us strengthen~\cref{lem:mult-approx-hard} to inputs where it is promised that $\E_{x \sim \dist}[\heatoptimal(x)] > 0$.
In particular, the LLM pointed us towards \cref{lem:sat-promise-conv} as a black-box way to convert any $\SAT$ instance $\phi$ into a new one $\psi$ where (1) there is always some satisfying assignment and (2) there is a programmable gap between the number of satisfying assignments for $\phi$ depending on whether or not $\phi$ had some satisfying assignment.
We uncovered that the LLM likely arrived at this exact construction based on a similar construction by Blass and Gurevich~\cite{DBLP:journals/iandc/BlassG82} designed to handle (1) and then observing that (2) follows directly from adding additional dummy variables.
The LLM also generated a full proof of~\cref{lem:mult-approx-hard} which we have heavily simplified into the version presented in this paper.

We suspected that the add-one smoothing estimator would be a natural candidate to use in the proof of \cref{lem:additive-approx-easy}, as it is a simple estimator which avoids the catastrophic pitfall of assigning zero probability on any event with positive probability, and asked the LLM to show that this estimator gives our polynomial time additive approximation.
The LLM observed that the excess heat dissipation associated with using an estimator is upper-bounded by the expected KL-divergence between the input distribution and its estimated version.
We identified that this is an immediate consequence of what is now known as \emph{mismatch costs} that were first identified by Kolchinsky and Wolpert~\cite{KW17}.
Next, the LLM proceeded to give a rather convoluted proof that upper bounds the expected KL-divergence between an input distribution and it's add-one smoothed estimator.
We instead identified that Braess, Forster, Sauer, and Simon~\cite{DBLP:conf/alt/BraessFSS02} directly proved that the add-one smoothing estimator is well suited for minimizing the expected KL-divergence.

We independently discovered proofs of \cref{lem:hard-gate-ef-min,thm:batcher-entropy-integer-high} where the heat function was restricted to only depend on the value of the comparison (i.e. is $x_1 \leq x_2$?) rather than the actual input values.
The LLM helped us find that shifting exactly $2^{k_\ell}$ small elements between sets $A$ and $B$ causes a central gate to receive input $x_2,x_1$ instead of input $x_1,x_2$, which was the key ingredient in our proof of \cref{lem:batcher-integer-swap-bound}, that let us extend our lower bound to general heat functions.

The high level strategy in \cref{sec:better-than-batchers} of combining an $O(n \log n)$ gate almost sorting network with an $O(n \log^2 n)$ gate network to reduce the expected heat dissipation was the result of interactions between the authors and the LLM when searching for a more thermodynamically efficient sorting network on integers than Batcher's.

ChatGPT 5.6 and ChatGPT 6 Astra were used to help proofread this document and catch basic errors, which were fixed by the authors.
\cref{fig:batcher-n-32-heat} was generated by code produced with the assistance of Claude code.
All text, proofs, and the remaining figures in this paper were produced by the authors, who take full responsibility for the correctness and originality of all content therein.

\bibliographystyle{alpha}
\bibliography{sources}

\appendix
\section{Basic lemmas}
\label{sec:basic-lemmas}
\gibbsopt*
\begin{proof}
    Observe that
    \begin{align*}
        \E_{x \sim \dist_{|S}}[f(x)] - \entropy(\dist_{|S}) &= \E_{x \sim \dist_{|S}}[f(x) + \log \Pr_{\dist_{|S}}[x]]\\
        &=\E_{x \sim \dist_{|S}}[\log (\Pr_{\dist_{|S}}[x] /2^{-f(x)})]\\
        &=\E_{x \sim \dist_{|S}}[\log (\Pr_{\dist_{|S}}[x] / \Pr_{(\gibbs_f)_{|S}} [x])] - \log \left(\sum_{x \in S} 2^{-f(x)}\right)\\
        &= \DKL(\dist_{|S} || (\gibbs_f)_{|S}) - \log \left(\sum_{x \in S} 2^{-f(x)}\right).
    \end{align*}
    By \cref{prop:gibbs-inequality}, we know that $\DKL(\dist_{|S} || (\gibbs_f)_{|S})$ and hence $\E_{\dist_{|S}}[f(x)] - \entropy(\dist_{|S})$ is minimized exactly when $\Pr_{\dist_{|S}}[x] = \Pr_{(\gibbs_f)_{|S}}[x]$ for every $x \in S$.
\end{proof}
\turgut*
\begin{proof}
    By \cref{def:valid-heat} we know that $\heat$ is valid precisely when
    \begin{displaymath}
        \max_{\dist} \sum_{x \in \support(\dist)} \Pr_{\dist}[x] \left( \log(\Pr_{f(\dist)}[f(x)]) - \log(\Pr_{\dist}[x]) -  \heat(x) \right) \leq 0.
    \end{displaymath}
    Or equivalently when
    \begin{displaymath}
        \max_{\dist} \sum_{y \in \support(f(\dist))} \Pr_{f(\dist)}[y] \left( \entropy(\dist_{|f^{-1}(y)}) - \E_{\dist_{|f^{-1}(y)}}[\heat(x)]\right) \leq 0.
    \end{displaymath}
    By convexity, we can assume without loss of generality that the maximum $\dist$ has $\Pr_{f(\dist)}[y] = 1$ for some fixed choice of $y$.
    We can also assume that the maximum $\dist$ has no probability mass on any $x$ where $\heat(x) = \infty$ unless $\heat$ is $\infty$ on every point and the validity condition is trivial.
    Additionally, by \cref{lem:gibbs-opt} we know that the maximal $\dist_{|f^{-1}(y)}$ is $(\gibbs_\heat)_{|f^{-1}(y)}$.
    Thus the above is equivalent to the claim
    \begin{displaymath}
        \max_y \log\left(\sum_{x \in f^{-1}(y)} 2^{-\heat(x)}\right) \leq 0.
    \end{displaymath}
    Taking $2$ to the power of both sides of this inequality proves our desired bound.
\end{proof}
\section{Analyzing entropy flow in Batcher's sorting network: overwriting model with binary inputs}
\label{sec:batcher-overwriting}

Here we show that in the overwriting model, the expected heat dissipation of Batcher's sorting network is $O(n)$, matching our bound for our thermodynamic circuit model.
We first describe useful heat functions for our analysis,  then use these to bound the expected heat dissipation for layers 1 through $\log n$ of Batcher's network, and finish with a recurrence to bound the total heat dissipation.
This proof is very similar to the one for our standard circuit model (i.e. \cref{thm:batcher-heat-released-upper-bound}).

\paragraph{The overwriting model and its heat functions}
In the overwriting model, we view a circuit as a DAG where every node represents a gate and edges indicate input dependence.
Unlike in our standard circuit model, the gates in the overwriting model retain their output values in a persistent storage between runs and are not required to erase their inputs.
As such, the old output to every gate must be \emph{overwritten} by a new value during each run of the circuit.
For this section, we will be assuming that inputs to our overwriting circuit are independently and identically distributed and will use $\heato$ to denote heat functions for circuits in the overwriting model.
If a gate in the overwriting model computes a function $f: \domain \to \range$, then it is actually computing the function $h(x,y) = (x,f(x))$ for every $x \in \domain, y \in \range$.
As such, the possible heat functions for this gate is given by functions $\heato: \domain \times \range \to \R_{\geq 0}$ and the validity criteria matches that of \cref{def:valid-heat} when the function being computed is $h$.

Since our inputs are independently and identically distributed, there is a significantly simpler expression for the optimal heat function given by \cref{prop:ef-min-heat-fn}.
Let $\dist$ be the marginal induced distribution over inputs to our gate (including the output from the previous run), $\dist_\domain$ be the marginal distribution over new input $x\in \domain$, and $\dist_\range$ be the marginal distribution over the old output $y \in \range$.
Then, since $\dist$ is independently and identically distribution, we have that $\dist = \dist_\domain \otimes \dist_\range$.
Thus the optimal heat function for our gate is given by:
\begin{displaymath}
\heato(x,y) = \log \frac{\Pr_{x,y\sim h(\dist)}[(x,f(x))]}{\Pr_{x,y \sim \dist}[x,y]} = \log \frac{\Pr_{x \sim \dist_\domain}[x]}{\Pr_{x \sim \dist_\domain}[x]\Pr_{y \sim \dist_\range}[y]} = \log \frac{1}{\Pr_{y \sim \dist_\range}[y]}.
\end{displaymath}
So the optimal heat function for gates run on independently identically distributed inputs in the overwriting model is equivalent to the optimal heat function for erasing the value of $f(x)$.
Thus we will be specifying heat functions in the overwriting model whose only inputs are the values that are about to be overwritten.

\paragraph{Batcher's odd-even mergesort on bits}
Rather than use the optimal heat function for every gate, in our analysis, we'll use the three heat functions (parameterized by $\epsilon$) shown in the following table:

           \begin{center}
              \begin{tabular}{l|l||l|l||l||l|l|l}
                  $x_1,x_2$  & $y_1', y_2'$ & $\Prob$ & $\heato^{(0,0)}_\epsilon$ &  $\Prob$ & $\heato^{(0,1)}_\epsilon$ &  $\Prob$ & $\heato^{(1,1)}_\epsilon$ \\ \hline
                   0,0        & 0,0          & $1-\epsilon$ & $\log_2 (1/ (1 - \epsilon))$    & $\epsilon/2$ & $\log_2 (2 / \epsilon)$            &  $\epsilon/2$ & $\log_2 (2 / \epsilon)$ \\
                  0,1 or 1,0 & 0,1       & $\epsilon/2$ & $\log_2 (2 / \epsilon)$           & $1-\epsilon$ & $\log_2 (1/ (1 - \epsilon))$  & $\epsilon/2$ & $\log_2 (2 / \epsilon)$      \\
                  1,1           & 1,1        & $\epsilon/2$ & $\log_2 (2 / \epsilon)$          & $\epsilon/2$ & $\log_2 (2 / \epsilon)$         & $1-\epsilon$ & $\log_2 (1/ (1 - \epsilon))$      \end{tabular}
          \end{center}
          
In the table, preceding the column for each heat function $\heato^{(b,b')}_\epsilon$, are the probabilities for which that function is optimal. For example. 
$\heato^{(0,0)}_\epsilon$ is optimal when $(y_1',y_2')= (0,0)$ with probability $1-\epsilon$, $(0,1)$ with probability $\epsilon/2$, and $(1,1)$ with probability $\epsilon/2$ (or equivalently when $(x_1,x_2) = (0,0)$ with probability $1-\epsilon$, $(0,1)$ or $(1,0)$ with probability $\epsilon/2$, and $(1,1)$ with probability $\epsilon/2$).
The expected heat release of a gate with one of these heat functions, when the input distribution is optimal, is
     \[
           (1 - \epsilon)  \log_2 (1/(1-\epsilon))  + \epsilon \log_2 (2/\epsilon) = O(\epsilon \log (1/\epsilon)),
          \]
          using the fact that $\log_2 (1/(1-\epsilon)) = O(\epsilon)$ by application of
          Taylor's expansion of $\log (1-\epsilon)$. 
Moreover, on distributions where $(x_1,x_2) = (b,b')$ with probability {\em at least} $1-\epsilon$, the expected heat release of a gate with heat function $\heato^{(b,b')}_\epsilon$ is $O(\epsilon \log 1/\epsilon)$. 
 
We will also use heat function $\heato^{(\unif)}$ that is optimal on the uniform input distribution. Its values are $\heato^{(\unif)}(0,0) = \heato^{(\unif)}(1,1) = 2$, and $\heato^{(\unif)}(0,1) = 1$ and the expected heat release is
          \[
          (1/4) 2 + (1/4) 2 + (1/2)  = 3/2.
          \]
          Moreover, for any input distribution, the expected heat release with respect to  $\heato^{(\unif)}$
          is at most 2; it approaches 2 as the probability of
          inputs (0,1) or (1,0) approaches 0.

          A key difference between the overwriting model and our circuit model is that \emph{logically reversible} operations do not buy you anything extra in the overwriting model.
          In \cref{sec:batcher-bits} we saw that we could always set $\heat(0,0) = \heat(1,1) = 0$ since these inputs give unique outputs.
          On the other hand, we need to dissipate heat when we write over these outputs in the overwriting model.

As in the analysis of Batcher's network for our 
model, we partition the comparator gates into $\log_2 n +1$ layers, with layer $0$ encompassing
the recursive sub-networks, as in Figure \ref{fig:batcher}.
Within each layer $\ell$, we will partition the gates into the same three regions as in the proof of \cref{thm:batcher-heat-released-upper-bound}, but will pick heat functions that are amenable to the overwriting model.
Again $\alpha \in (0,1)$ is a parameter we will specify later.

\paragraph{Layer 1}
In this layer, there are $n/2$ comparator gates.
The $i$'th gate has input wires $i$ and $i+n/2$.
\begin{itemize}
    \item
Upper region: $i \in [0, (n-n^{(1+\alpha)/2})/4]$.
Since just after layer $0$, bits $0,1,\ldots, n/2-1$ are sorted, and also the remaining
$n/2$ bits are independently sorted,
it follows from \cref{prop:hoeffding} that for $i$ in this region,
\begin{equation}
    \label{eqn:pi0}
    \Prob[\valil{i}{0}=1] \le \exp(-n^{\alpha}/4) \mbox{ and also }
    \Prob[\valil{n/2+i}{0}=1]  \le \exp(-n^{\alpha}/4).
\end{equation}
\begin{sloppypar}
Consequently, the $i$th comparator gate has the input $(0,0)$ with probability at least $(1-\exp(-n^{\alpha}/4))^2 \ge 1 - 2\exp(-n^{\alpha}/4)$. 
If the gate has heat function $\heato^{(0,0)}_{2\exp(-n^{\alpha}/4)}$,
the expected heat release is $O(n^{\alpha} \exp(-n^{\alpha}/4)) = O( \exp(-n^{\alpha}/5))$.
\end{sloppypar}
\item
Lower region: $i \in [(n+n^{(1+\alpha)/2})/4,n/2-1]$. A similar analysis shows that the $i$th comparator gate in this region gets $(1,1)$
as input with probability at least  $1 - 2\exp(-n^{\alpha}/4)$.
If the gate has heat function $\heato^{(1,1)}_{2\exp(-n^{\alpha}/4)}$,
the expected heat release is $O( \exp(-n^{\alpha}/5))$. 
\item
Middle band: $i \in [(n-n^{(1+\alpha)/2})/4, (n+n^{(1+\alpha)/2})/4]$.
The expected heat release each gate in this region is at most 2 if the gates have heat function $\heato^{\unif}$.
\end{itemize}
Summing up over all three regions, the total expected heat dissipation for layer 1 is dominated by the middle band, and is $O(n^{(1+\alpha)/2})$, using three different types of heat functions.

\paragraph{Layer 2}

There are $n/4$ comparator gates in layer 2; we number them by the index of their upper wire, from $n/4$
to $n/2-1$. The $i$'th gate has input wires $i$ and $i+n/4$.
We partition the gates into three regions: a narrow band of gates at the top, a large middle region,
and a narrow band at the bottom.
\begin{itemize}
\item
Top narrow band: $i \in [n/4,(n+n^{(1+\alpha)/2})/4]$.
Using the heat function $\heato^{\unif}$, the expected
heat dissipation for each gate is at most 2. 
\item
Middle region: $i \in [(n+n^{(1+\alpha)/2})/4, (n-n^{(1+\alpha)/2})/2]$.
It follows from \cref{prop:hoeffding} that $\Prob[\valil{i}{1}=1] \hspace{-1.5pt}= \hspace{-1.5pt}\Prob[ \min \{ \valil{i}{0},\valil{n/2+i}{0} \} \hspace{-1.5pt} = \hspace{-1.5pt}1] \hspace{-1.5pt} = \hspace{-1.5pt}  \Prob[ \valil{i}{0} \hspace{-1.5pt}= \hspace{-1.5pt} 1] \Prob[\valil{i}{0} = 1]
\hspace{-3pt} \ge \hspace{-3pt} 1 - 2\exp(-n^{\alpha}/4)$,
and
$\Prob[\valil{i+n/4}{1}=0] = \Prob[ \max \{ \valil{i-n/4}{0},\valil{i+n/4}{0} \} = 0] \hspace{-1.5pt} \ge \hspace{-1.5pt} \Prob[ \valil{i-n/4}{0} = 0] \Prob[ \valil{i+n/4}{0} = 0] \hspace{-2.5pt} \ge \hspace{-2.5pt} 1 - 2\exp(-n^{\alpha}/4)$.
So each comparator gate in this range has input $(1,0)$ with probability at least  $1 - 4\exp(-n^{\alpha}/4)$. Using the heat function $\heato^{(0,1)}(4\exp(-n^{\alpha}/4))$, the expected heat release of each gate is $O(n^{\alpha} \exp(-n^{\alpha}/4)) = O( \exp(-n^{\alpha}/5))$. 
\item
Bottom narrow band: $i \in [(n-n^{(1+\alpha)/2})/2, n/2-1]$. Using the heat function $\heato^{\unif}$, the the expected
heat dissipation for each gate is at most 2. 
\end{itemize}
Therefore, the total expected heat dissipation for layer 2 is dominated by the narrow bands, and is $O(n^{(1+\alpha)/2})$, using two different types of heat functions.

\paragraph{Layers 3 and beyond}
We will reuse \cref{claim:gates-have-biased-halves} from the proof of \cref{thm:batcher-heat-released-upper-bound} which shows that with high probability, the number of 1-valued wires in the range $[0,n/2-1]$, and the number of 0-valued wires in the range $[n/2,n-1]$, is small.

\gateshavebiasedhalvesclaim*

We partition gates of each layer $\ell, 3 \le \ell \le \log_2 n$, into three groups: upper gates,
lower gates, and crossover gates, defined in the cases below.
In Figure \ref{fig:batcher}, where $n = 16$, when $l=3$ there are two upper gates,
two lower gates, and two crossover gates.

\begin{itemize}
 \item
 Upper gates: Both inputs are from wires in the range $[0,\ldots,n/2-1]$. We assign them the heat function $\heato^{(0,0)}(2\exp(-n^{\alpha}/4))$. On input $(0,0)$, such gates have expected heat dissipation $O( \exp(-n^{\alpha}/4))$, and otherwise they have expected heat dissipation $O(n^{\alpha} \log_2 (1/\exp(-n^{\alpha}/4))) = O(n^{2\alpha})$. It follows from \cref{claim:gates-have-biased-halves}
 that with probability $1 - O(\exp(-n^{\alpha}/5))$, the number of upper gates that don't have input (0,0) is $O(n^{(1+\alpha)/2})$.
 In this scenario, the total (conditional) expected heat dissipation of the upper gates of layer 3 is
          \[
              O(n \exp(-n^{\alpha}/4)) + O(n^{(1+\alpha)/2} n^{2\alpha}) = O(n^{(1+5\alpha)/2}).
          \]
          With the remaining probability $O(\exp(-n^{\alpha}/5))$, there may be more
          gates that don't have input (0,0); in the worst case none of the $n/8$ gates
          has input $(0,0)$. In this alternative scenario the conditional expected heat dissipation is
          $O(n^{1 + 2 \alpha})$. Taking the expectation over the two scenarios, the total
          heat dissipation for the upper gates of layer 3 is
          \[
              (1 - O(\exp(-n^{\alpha}/5))) O(n^{(1+5\alpha)/2}) + O(\exp(-n^{\alpha}/5) n^{1 + 2\alpha}) = O(n^{(1+5 \alpha)/2}).
          \]
    \item
    Lower gates: Both inputs are from wires in the range $[n/2,\ldots,n-1]$. The reasoning is similar to that for the upper gates, with heat function $\heato^{(1,1)}(2\exp(-n^{\alpha}/4))$ and expected heat dissipation $O(n^{(1+5 \alpha)/2})$.
    \item
 Crossover gates: One input wire is in the range $[0,\ldots,n/2-1]$ and the other is in the range $[n/2,\ldots,n-1]$.
  The reasoning is similar to the previous two cases, with heat function $\heato^{(0,1)}(2\exp(-n^{\alpha}/4))$ and expected heat dissipation $O(n^{(1+5 \alpha)/2})$. 
\end{itemize}
Summing up over upper, lower and crossover gates, the total heat dissipation per layer in the range 3 to $\log_2 n$ is $O(n^{(1+5 \alpha)/2})$, and three different heat functions are used.

\begin{theorem}
    \label{thm:batcher-heat-released-upper-bound-overwriting}
    With respect to the overwriting circuit model, when given the uniform distribution over inputs in $\{0,1\}^n$,
    there exist valid heat functions for each gate in Batcher's odd-even mergesort such that the expected total heat release across all gates is $O(n)$.
\end{theorem}

\begin{proof}
Adding up the expected heat dissipation for layers 1 through $\log_2 n$, we see that the total expected heat dissipation is
\[
    O(n^{(1+\alpha)/2}) + O(n^{(1+5\alpha)/2} \log_2 n),
\]
where the first term accounts for layers 1 and 2, and the second term accounts for gates in layers 3 to $\log_2 n$.
The total heat dissipation for the whole circuit is therefore given by the
recurrence
\[
    \EF(n) = 2 \EF(n/2) + O(n^{(1+5\alpha)/2} \log_2 n)
\]
and so is $O(n)$ when $\alpha < 1/5$. The number of different heat functions used is $O(1)$ per layer, and so is  $O(\log n)$ in total for the network.
\end{proof}

\section{Re-deriving mismatch costs}
\label{sec:rederive-mismatch-cost}
In \cite{WK20,YCW24,YYW25} the authors use the language of \emph{mismatch costs} to lower bound the entropy production of implementing gates in a Boolean circuit.
Instead of explicitly defining heat functions, they start with the assumption that some distribution over inputs $\dist^*$ minimizes the entropy production of the desired computation.
Their main observation is that entropy production can be decomposed as the sum of a term characterizing how far the actual input distribution is from $\dist^*$ and a term characterizing the inefficiencies even when using $\dist^*$ as the input distribution.
We start by characterizing how $\dist^*$ can be derived from a given heat function.

\begin{lemma}
\label{lem:opt-heat-fun-for-dist-marginal}
    Let $f$ be an arbitrary function with domain $\domain$ and $\heat: \domain \to \R_{\geq 0} \cup \{\infty\}$ be a valid heat function for computing $f$.
    Then any conditional input distribution $\dist_{|f^{-1}(y)}^*$ that minimizes the entropy production of applying a gate $g$ computing $f$ with heat function $\heat$ must satisfy
    \begin{displaymath}
        \dist_{|f^{-1}(y)}^* = (\gibbs_\heat)_{|f^{-1}(y)}
    \end{displaymath}
    on all $y$ where $(\gibbs_\heat)_{|f^{-1}(y)}$ is defined.
    If $(\gibbs_\heat)_{|f^{-1}(y)}$ is undefined, then all conditional input distributions have infinite entropy flow.
\end{lemma}

\begin{proof}
    By the definition of entropy production, we have that for every $y$:
    \begin{displaymath}
        \EP_{g}(\dist_{|f^{-1}(y)}) = -\entropy(\dist_{|f^{-1}(y)}) + \E_{x \sim \dist_{|f^{-1}(y)}}[\heat(x)].
    \end{displaymath}
    By \cref{lem:gibbs-opt}, we know that $(\gibbs_\heat)_{|f^{-1}(y)}$ is the distribution that minimizes this quantity.
    When $(\gibbs_\heat)_{|f^{-1}(y)}$ is undefined, it must be the case that $\heat(x) = \infty$ for every $x \in f^{-1}(y)$.
    In such cases, any distribution over these inputs must release infinite expected heat.
\end{proof}

\begin{corollary}\label{cor:opt-dist-has-full-support-on-opt-inputs}
    Let $f$ be an arbitrary function with domain $\domain$ and $\heat: \domain \to \R_{\geq 0} \cup \{\infty\}$ be a valid heat function for computing $f$ such that $\heat(x) \in \R_{\geq 0}$ for some $x \in \domain$.
    Then any input distribution $\dist^*$ over $\domain$ that minimizes the entropy production of applying a gate computing $f$ with heat function $\heat$ must have conditional distributions $\dist^*_{|f^{-1}(y)} = (\gibbs_\heat)_{|f^{-1}(y)}$ and the distribution $f(\dist^*)$ must only have support on the $y$ where $(\gibbs_\heat)_{|f^{-1}(y)}$ is defined and that minimize $\EP_{g}((\gibbs_\heat)_{|f^{-1}(y)})$.
\end{corollary}
\begin{proof}
    The entropy production on an input distribution $\dist$ is given by:
    \begin{align*}
        \EP_{g}(\dist) &= \E_{x \sim \dist}[\heat(x)] + \entropy(f(\dist)) - \entropy(\dist)\\
        &= \E_{y \sim f(\dist)}\left[ \log(1/\Pr_{f(\dist)}[y]) + \E_{x \sim \dist_{|f^{-1}(y)}}\left[Q(x) - \log(1/\Pr_{\dist}[x])\right]\right]\\
        & =\E_{y \sim f(\dist)}\left[\E_{x \sim \dist_{|f^{-1}(y)}}\left[Q(x)\right] -\entropy(\dist_{|f^{-1}(y)})\right]\\
        &= \E_{y \sim f(\dist)}[\EP_{g}(\dist_{|f^{-1}(y)})]
    \end{align*}
    The entropy production can be expressed as a convex combination of the entropy production for the different conditional distributions.
    \cref{lem:opt-heat-fun-for-dist-marginal} exactly characterizes what these conditional distributions look like.
    Due to convexity, any optimal distribution must have all of its probability mass on $x$ where $\EP_{g}(\dist_{|f^{-1}(f(x))})$ is minimized.
\end{proof}
We now use heat functions to reprove the decomposition of entropy production into the sum of terms first presented in \cite{WK20}.
Intuitively, in \Cref{prop:mismatch-cost}, $\dist^*$ represents the distribution on inputs to gate $g$ for which the heat function $\gibbs$ is optimal (i.e., that minimizes entropy production), and $\dist$ represents the true distribution, incurring a higher entropy production. 
The difference is known as the \emph{mismatch cost}. The significance of \Cref{prop:mismatch-cost} is that this mismatch cost can be characterized purely information theoretically in terms of $\dist$ and $\dist^*$ (specifically by the KL-divergence terms $\DKL(\dist || \dist^*) - \DKL(f(\dist) || f(\dist^*)$). In other words the way the physics of the device implementing the gate influences the extra entropy production is entirely contained in its optimal distribution $\dist^*$.
Our proof of this result uses the language of heat functions, but is similar to that in \cite{WK20}.

\begin{proposition}[\cite{WK20}]
\label{prop:mismatch-cost}
    Let $g$ be a gate computing function $f$ with domain $\domain$, and valid heat function $\heat$. Let $\dist$ be any input distribution over $\domain$, and $\dist^*$ be another input distribution over $\domain$ where $\support(f(\dist)) \subseteq \support(f(\dist^*))$ and $\dist^*_{|f^{-1}(y)} = (\gibbs_\heat)_{|f^{-1}(y)}$ for every $y$ where both of these distributions are defined.
    If $(\gibbs_\heat)_{|f^{-1}(y)}$ is undefined for some $y \in \support(f(\dist))$ then $\EP_{g}(\dist) = \infty$.
    Otherwise, the entropy production of applying $g$ on $\dist$ is:
    \begin{displaymath}
        \EP_{g}(\dist) = \E_{y \sim f(\dist)}[ \EP_{g}((\gibbs_\heat)_{|f^{-1}(y)})] + \DKL(\dist || \dist^*) - \DKL(f(\dist) || f(\dist^*))
    \end{displaymath}
\end{proposition}

\begin{proof}
    If $\support(\dist) \not \subseteq \support(\dist^*)$ then there must be some $x \in \support(\dist)$ where $\heat(x) = \infty$.
    We would immediately have that $\EP_{g}(\dist) = \infty$ and also $\DKL(\dist || \dist^*) = \infty$ while the remaining terms on the right-hand side of the equality would be finite.
    Thus, we can assume without loss of generality that $\support(\dist) \subseteq \support(\dist^*)$.
    \begin{align}
        \EP_{g}(\dist) &= \entropy(f(\dist)) - \entropy(\dist) + \E_{x \sim \dist}[\heat(x)] \label{eqn:ep-mmc-islands-1}\\
        &= \E_{y \sim f(\dist)}\left[\E_{x \sim \dist_{|f^{-1}(y)}}[ \log(\Pr_{\dist}[x]/\Pr_{f(\dist)}[y]) + \heat(x)]\right] \nonumber\\
        &= \E_{y \sim f(\dist)}\left[\E_{x \sim \dist_{|f^{-1}(y)}}[ \log(\Pr_{\dist_{|f^{-1}(y)}}[x]) + \heat(x)]\right] \nonumber\\
        &= \E_{y \sim f(\dist)}\left[ \E_{x \sim \dist_{|f^{-1}(y)}} [\log(2^{\heat(x)}\Pr_{\dist_{|f^{-1}(y)}}[x])]\right]. \nonumber
    \end{align}
    Since $\EP_{g}((\gibbs_\heat)_{|f^{-1}(y)}) = \E_{x \sim (\gibbs_\heat)_{|f^{-1}(y)}} [\log(2^{\heat(x)} \Pr_{(\gibbs_\heat)_{|f^{-1}(y)}}[x])]$ we can rewrite \cref{eqn:ep-mmc-islands-1} as:
    \begin{align}
        &\E_{y \sim f(\dist)}\left[\EP_{g}((\gibbs_\heat)_{|f^{-1}(y)}) + \E_{x \sim \dist_{|f^{-1}(y)}} [\log(2^{\heat(x)}\hspace{-3pt}\Pr_{\dist_{|f^{-1}(y)}} \hspace{-3pt}[x])] - \E_{x \sim (\gibbs_\heat)_{|f^{-1}(y)}} [\log(2^{\heat(x)} \hspace{-10pt} \Pr_{(\gibbs_\heat)_{|f^{-1}(y)}}\hspace{-10pt}[x])]\right]\label{eqn:ep-mmc-islands-2}\\
        &= \E_{y \sim f(\dist)}\left[ \EP_{g}((\gibbs_\heat)_{|f^{-1}(y)}) \right] \nonumber
        \\&
        \quad + \E_{y \sim f(\dist)} \left[  \E_{x \sim \dist_{|f^{-1}(y)}} [\log(2^{\heat(x)}\Pr_{\dist_{|f^{-1}(y)}}[x])] - \E_{x \sim (\gibbs_\heat)_{|f^{-1}(y)}} [\log(2^{\heat(x)} \Pr_{(\gibbs_\heat)_{|f^{-1}(y)}}[x])]\right].
        \nonumber
    \end{align}
    Now by the definition of $(\gibbs_{\heat})_{|f^{-1}(y)}$, for every $x\in \support((\gibbs_{\heat})_{|f^{-1}(y)})$, $2^{\heat(x)} \Pr_{(\gibbs_\heat)_{|f^{-1}(y)}}[x] = 1/\sum_{x' \in f^{-1}(y)} 2^{-\heat(x')}$.
    This implies that this quantity has the same value for every element of the set $f^{-1}(y)$.
    Therefore, any convex combination over the terms with the same values of $y$ will give the same value.
    This lets us rewrite \cref{eqn:ep-mmc-islands-2} as:
    \begin{align}
        &= \E_{y \sim f(\dist)}[ \EP_{g}((\gibbs_\heat)_{|f^{-1}(y)})] + \E_{x \sim \dist} [\log (2^{\heat(x)} \Pr_{\dist_{|f^{-1}(f(x))}}[x]) - \log(2^{\heat(x)} \Pr_{(\gibbs_\heat)_{|f^{-1}(f(x))}}[x])] \nonumber\\
        &= \E_{y \sim f(\dist)}[ \EP_{g}((\gibbs_\heat)_{|f^{-1}(y)})]  + \E_{x \sim \dist}[\log (\Pr_{\dist_{|f^{-1}(f(x))}}[x] / \Pr_{(\gibbs_\heat)_{|f^{-1}(f(x))}}[x])]] \nonumber\\
        &= \E_{y \sim f(\dist)}[ \EP_{g}((\gibbs_\heat)_{|f^{-1}(y)})] + \E_{x \sim \dist} \left[\log \left(\frac{\Pr_{\dist}[x] \Pr_{f(\dist^*)}[f(x)]}{\Pr_{f(\dist)}[f(x)]\Pr_{\dist^*}[x]}\right)\right] \nonumber\\
        &= \E_{y \sim f(\dist)}[ \EP_{g}((\gibbs_\heat)_{|f^{-1}(y)})] + \DKL(\dist || \dist^*) - \DKL(f(\dist)||f(\dist^*)) \nonumber
    \end{align}
    as desired.
\end{proof}
There are actually some subtlety in our statement of \cref{prop:mismatch-cost} that becomes obvious with the language of heat functions.
In particular, we cannot always take $\dist^*$ to be the distribution over inputs that minimizes the entropy production.
Due to \cref{cor:opt-dist-has-full-support-on-opt-inputs}, we know that such a $\dist^*$ would only have support on $y$ that minimize $\EP_{g}((\gibbs_\heat)_{|f^{-1}(y)})$, and thus for many choices of $\dist$ we will not have that $\support(f(\dist)) \subseteq \support(f(\dist^*))$.
\cref{prop:mismatch-cost} actually has a nice intuitive interpretation in terms of heat functions.
The term $\DKL(\dist || \dist^*) - \DKL(f(\dist) || f(\dist^*))$ corresponds to how well suited $\heat$ is for the input distribution $\dist$ while $\E_{y \sim f(\dist))} [\EP_{g}((\gibbs_\heat)_{|f^{-1}(y)})]$ corresponds to how close $\heat$ is to being optimal on the best possible distributions over the preimage of each $y$.

\section{Entropy production from modularity}
\label{sec:modularity}
Heat functions only consider the inputs to a specific gate to determine the corresponding heat release.
Thus, they do not account for correlations between the inputs to this gate and the rest of the circuit.
This leads to a source of entropy production due to the locality of operations that are not injective.
This source of entropy production was first identified in~\cite{WK20}, where it was shown that the entropy production of applying a gate in a circuit is at least the drop in mutual information between the wires in this gate and the state of the remaining wires of the circuit.
We prove essentially the same result, showing that it is an immediate consequence of the locality of valid heat functions.

\begin{proposition}[\cite{WK20}]
\label{lem:mutual-info-positive-ep}
    Let $\circuit$ be a thermodynamic circuit with $\varSpace$ wires,
    and let $X,Y$ be a partition of $[\varSpace]$ and $\circuit_{t-1}[X], \circuit_{t-1}[Y] \sim \dist_{t-1}$ be random variables for the state of the wires in $X,Y$ after applying gate $t-1$ in $\circuit$. Let $\dist_X$ be the marginal distribution of random variable $\circuit_{t-1}[X]$ according to distribution $\dist_{t-1}$.
    Then if gate $g_t$ computes the function $f^{g_t}: \alphabet^X \to \alphabet^X$ and has valid heat function $\heat_{g_t}$, the entropy production of applying $g_t$ as the next gate in $\circuit$ is $\EPlocal_{g_t}(\dist_X)+ \mutualInformation(\circuit_{t-1}[X];\circuit_{t}[Y]) - \mutualInformation(\circuit_{t}[X];\circuit_{t}[Y])$.
\end{proposition}

\begin{proof}
    By \cref{def:local-ep},
    \begin{displaymath}
        \E_{x \sim \dist_X}[\heat_{g_t}(x)] = \EPlocal_{g_t}(\dist_X) + \entropy(\circuit_{t-1}[X]) - \entropy(\circuit_{t}[X])
    \end{displaymath}
    Let $\dist_t$ be the distribution over states of the circuit after applying $g_t$.
    Since $\circuit_{t-1}[Y]$ is the same as $\circuit_t[Y]$, we know that according to \cref{def:ep-gate}:
    \begin{align*}
        \EP_{g_t}(\dist_{t-1}) &= \E_{x \sim {\dist_{t-1}}}[\heat_{g_t}(x_{X})] + \entropy(\dist_t) - \entropy(\dist_{t-1})\\
        &= \E_{x \sim \dist_X}[\heat_{g_t}(x)] + \entropy(\dist_t) - \entropy(\dist_{t-1})\\
        &= \EPlocal_{g_t}(\dist_X) + \entropy(\circuit_{t-1}[X]) - \entropy(\circuit_t[X]) + \entropy(\dist_t) - \entropy(\dist_{t-1})\\
        &= \EPlocal_{g_t}(\dist_X) + \entropy(\circuit_{t-1}[X]) - \entropy(\circuit_t[X]) + \entropy(\circuit_{t}[X],\circuit_{t}[Y]) - \entropy(\circuit_{t-1}[X],\circuit_{t}[Y])\\
        &=\EPlocal_{g_t}(\dist_X) + \entropy(\circuit_{t-1}[X]) - \entropy(\circuit_t[X]) + \entropy(\circuit_{t}[X] \given \circuit_t[Y]) - \entropy(\circuit_{t-1}[X] \given \circuit_{t}[Y])\\
        &=\EPlocal_{g_t}(\dist_X) + \mutualInformation(\circuit_{t-1}[X];\circuit_{t}[Y]) - \mutualInformation(\circuit_{t}[X];\circuit_{t}[Y])
    \end{align*}
    as desired.
\end{proof}
Thus we can decompose entropy production into a term based off of sub-optimal computation on the marginal input distribution plus the drop in mutual information caused by the locality constraint on heat functions.

\end{document}